\documentclass[runningheads]{llncs}
\usepackage[utf8]{inputenc}
\usepackage{wrapfig}
\usepackage{array}
\usepackage{amssymb,amsmath,amsfonts}
\usepackage[scr]{rsfso}
\usepackage{graphicx}
\usepackage{xcolor}
\usepackage{ifthen}
\usepackage{tikz}  
\usetikzlibrary{positioning,calc}
\usetikzlibrary{decorations.pathmorphing,shapes,decorations.pathreplacing}
\usepackage{enumitem}
\setlist{itemsep=0pt,topsep=4pt}

\newcommand{\comment}[4][inline]{
  \ifthenelse{\equal{#1}{margin}}
  {
    \marginpar{\scriptsize \hrule\noindent
    {\bf #2:\\} {\em\textcolor{#3}{#4}} \hrule}
  }
  {
    \hrule\noindent {\bf #2:} {\em\textcolor{#3}{#4}} \hrule
  }
}

\newcommand{\Na}{\mathbf{a}} 
\newcommand{\Nb}{\mathbf{b}} 
\newcommand{\Nc}{\mathbf{c}} 
\newcommand{\Nr}{\mathbf{r}} 
\newcommand{\Nn}{{n}} 
\newcommand{\cards}{D} 
\newcommand{\deal}{\mathit{deal}} %
\newcommand{\signature}{\gamma} %
\newcommand{\MA}{\mathcal{M}} 

\renewcommand{\phi}{\varphi}
\newcommand{\powerset}[1]{\mathscr{P}(#1)}
\newcommand{\powersetS}[2]{\mathscr{P}_{#1}(#2)} 

\def\cG{{\cal G}}

\def\cI{{\cal I}}

\def\cO{{\cal O}}

\def\cT{{\cal T}}

\def\cI{{\cal I}}

\newcommand{\ceil}[1]{\left\lceil #1 \right\rceil}

\newcommand{\floor}[1]{\left\lfloor #1 \right\rfloor}

\newcommand{\set}[1]{\left\{ #1 \right\}}

\newcommand{\var}[1]{\lstinline+#1+}

\begin{document}
\title{A  Distributed Computing Perspective of  Unconditionally Secure Information
Transmission in Russian Cards Problems}
%
\titlerunning{Unconditionally Secure Information Transmission in Russian Cards Problems}
%
\author{Sergio Rajsbaum}
\authorrunning{S. Rajsbaum}
%
\institute{Instituto de Matem\'aticas \\
Universidad Nacional Aut\'onoma de M\'exico (UNAM)\\
Mexico City 04510, Mexico\\
\email{rajsbaum@im.unam.mx}
}
\maketitle              
\begin{abstract}
The  problem of $A$ privately transmitting information  to
$B$ by a  public announcement overheard by an eavesdropper $C$ is considered.
To do so by a deterministic protocol, their inputs must be correlated.
Dependent inputs are represented using a deck of  cards.
There is a publicly known \emph{signature} $(\Na,\Nb,\Nc)$,
  where $\Nn=\Na+\Nb+\Nc+\Nr$, and  $A$ gets $\Na$ cards, 
  $B$ gets $\Nb$ cards, and $C$ gets  $\Nc$ cards, out of the deck of $\Nn$ cards.
Using a deterministic protocol, $A$ decides its announcement  based on  her hand.\\

Using techniques  from coding theory, Johnson graphs,
and additive number theory,  a novel perspective inspired
by distributed computing theory is provided, to
analyze the  amount
of information that $A$ needs to send, while  preventing $C$
from learning a single card  of her hand.
In one extreme, the generalized Russian cards problem, 
$B$ wants to learn all of $A$'s cards, and in the other,  
 $B$ wishes to learn \emph{something} about $A$'s hand. 
\end{abstract}

\section{Introduction}

The idea that card games could be used to achieve
security in the presence of computationally unbounded adversaries 
proposed by  Peter Winkler~\cite{Winkler83}
led to an active research line
e.g.~\cite{FischerPR89secBitTr,FischerW92crypto,FWeffic93,FW96,KMizukiTN08,MaurerW99,MizukiT99,MizukiSN02journal,Winkler83}.
It motivated Fischer and Wright~\cite{FWeffic93}  to consider
  \emph{card games}, where $A,B,C$ draw cards from a deck $\cards$ of $\Nn$ cards,
as specified by a \emph{signature} $(\Na,\Nb,\Nc)$, with
 $\Nn=\Na+\Nb+\Nc+\Nr$.
Nobody gets $\Nr$ cards, while $A$ gets $\Na$ cards, 
  $B$ gets $\Nb$ cards, and $C$ gets  $\Nc$ cards.
  
 Fischer and Wright thought of the cards as representing
 correlated random initial local variables for the players, that have a simple structure. 
  They were interested in knowing
 which distributions of private initial values allow $A$ and $B$ to obtain 
a key, that remains secret to $C$. Their protocols mostly use randomization, and they 
are information-theoretic secure. However, they do not keep the cards of $A$ and $B$
secret from $C$.

Another research line started with an in depth, combinatorial and epistemic logic study of  van Ditmarsch~\cite{DitmarschRC03} of the~\emph{Russian cards}  problem,  presented at the Moscow Mathematics Olympiad in 2000, where the cards of $A$ and $B$ should be kept
secret from $C$.
Here $A$, $B$ and $C$ draw $(3,3,1)$ 
cards, respectively, from a deck of $7$ cards. 
  First $A$ makes an announcement
that allows $B$ to identify her set of cards, while $C$ cannot deduce a single card of $A$. 
After the announcement of $A$,
 $B$ knows the cards of each player, and hence he may announce $C$'s card, 
 from which $C$ learns nothing, but allows $A$ to infer the cards of $B$. 
The problem  has received  a fair amount of attention since then\footnote{The  $\Nr=0$ case is mostly considered
here, as well as in the secret key research line.
} 
e.g.~\cite{Albert2011securComm,Albert2005SafeCF,Cordon-FrancoDF12,geomProtCDFS15,colorGRSP13,Ditmarsch2011ThreeS,Duan2010,LFcaseStudy17,swansStinson14,SSadd2014},
in its  \emph{generalized} form of signature $(\Na,\Nb,\Nc)$, 
and other variants, including multiround, multiplayer, and different security requirements.
Solutions are based either   on modular arithmetic or
on combinatorial designs.



The original solution for $(3,3,1)$ uses modular arithmetic, where $A$ announces the sum of her cards
modulo $7$, and then $B$ announces $C$'s card~\cite{makarychev2001importance}. 
For the general case when $\Nc=1$ (and $\Nr=0$), solutions exist that announce the cards sum modulo
an appropriate prime number greater or equal to $\Nn$~\cite{Cordon-FrancoDF12}. These solutions use only two announcements. A solution using three announcements for $(4, 4, 2)$ is reported in~\cite{Ditmarsch2011ThreeS}, 
and a four-step protocol where $C$ holds approximately the square of the number of cards of
$A$ is presented in~\cite{colorGRSP13}. 

The relation to Steiner triple system and combinatorial designs  
goes  back  to 1847 Kirkman~\cite{kirkman}.
Using combinatorial designs Cord\'{o}n-Franco et al.~\cite{colorGRSP13} prove that  solutions exist 
when $\Na$ is a power of a prime, and present the first solutions when $\Nc > \Na$.
The solution used 4 communication steps, as opposed to the usual 2-step protocols.
Albert et al.~\cite{Albert2005SafeCF}  show that there is no 2-step solution if $c\geq a-1$.


We provide an extensive discussion of related work in Appendix~\ref{app:relatedW}. 
In addition to the papers mentioned above, 
through our new perspective on these problems, we have uncovered relations
with other areas: intersecting families of sets, coding theory, additive number theory, and distributed computability.

\paragraph{The new approach.}

Given a  publicly known {signature} $(\Na,\Nb,\Nc)$, for a deck $\cards$ of
 $\Nn=\Na+\Nb+\Nc+\Nr$ cards, 
the  basic problem underlying the situations described above, is 
to design a safe protocol $P_A$, so that $A$ 
makes a public announcement, $P_A(a)$, based on her hand, $a$. 
From the announcement $P_A(a)$,   and using his own hand, $b$, $B$
should learn something about $A$'s hand. 
 The announcement  $P_A(a)$ is deterministically determined
by the input of $A$, and the knowledge of the signature. No randomized solutions are considered in this paper.

In the language of e.g.~\cite{Cordon-FrancoDF12,colorGRSP13,Ditmarsch2011ThreeS}, 
a protocol $P_A$ should be \emph{informative} for $B$ and \emph{safe} from $C$. 
A protocol is \emph{safe} if  $C$ 
 does not learn any of the cards of $A$.
It is informative, if $B$  learns the hand of $A$.

We define
 the notion of a protocol being  \emph{minimally informative,} where the goal is
that  $B$ learns \emph{something} about the hand of $A$. We prove that
the minimal information problem is a kind of oblivious transfer problem,
in the sense that, when $\Nc+\Nr=1$, $B$ learns one card of $A$,  but $A$ does not know which one.
If $\Nc+\Nr>1$ then $B$ learns even less; he
learns that $A$ has  one of the cards of a set $s$, $|s|=\Nc+\Nr$.

In Section~\ref{sec:basicProb} we formalize this setting  
based on distributed computability~\cite{HerlihyKR:2013}, and more specifically 
when  the least amount of communication is studied~\cite{DFRbits2020}.

In Section~\ref{sec:mainProt},
using this formalization, we show  that a protocol can be viewed
as a  coloring of the set of vertices
$\powersetS{\Na}{\cards}$, all subsets of size $\Na$ of  $\cards$,
$$
P_A: \powersetS{\Na}{\cards} \rightarrow \MA,
$$
for the set of messages $\MA$ that $A$ may send. Thus, $\powersetS{\Na}{\cards}$ is the set of vertices of a Johnson graph $J(\Nn,\Na)$,
where $\Nn=|\cards|$.
We are interested in the question of how small can $\MA$ be, i.e.,
the  number of bits, $\log_2 | \MA |$, that $A$ needs to  transmit to implement either and
informative
or a minimally informative safe protocol.

We show in Theorem~\ref{def:prColoring} that $P_A$ is informative if and only if $P_A$ is a proper coloring
of  the 
\emph{$d$-distance Johnson graph} $J^d(\Nn,\Na)$, $d=\Nc+\Nr$.
Vertices $a,a'$ of  $J^d(\Nn,\Na)$ are adjacent whenever $\Na-d\leq |a\cap a'|$.
In particular, we have a Johnson graph when $d=1$.

It is well-known that there is a family of maximal clicks of  $J(\Nn,\Na)$ of size $\Na+1$, e.g.~\cite{godsil_meagher_2015}. It turns out, that the inputs of $A$ that $B$ with input $b$
considers possible,  form a maximal click of $J^{\Nc+\Nr}(\Nn,\Na)$, denoted $K_p(\bar{b})$.
 The click $K_p(\bar{b})$ consists  of all hands $a\subset \bar{b}$,  $|a|=\Na$, and hence
$p=  \binom{\Na+\Nc+\Nr}{\Na} $.
 Similarly, the hands that  $C$ considers possible with input $c$ form a click $K_p(\bar{c})$
of  $J^{\Nb+\Nr}(\Nn,\Na)$, and such clicks are of size $p=  \binom{\Na+\Nb+\Nr}{\Na}$.

We show also in Theorem~\ref{def:prColoring} that $P_A$ is 
minimally informative if and only if $P_A$ colors
at least one edge of each  click $K_p(\bar{b})$ with two different colors.
In contrast, informative requires that $P_A$ colors
every edge of  $K_p(\bar{b})$ 
with two different colors.


Thus, the chromatic number of $J^d(\Nn,\Na)$ determines the number
of messages needed for a protocol $P_A$ to be informative.
There are many interesting open questions concerning the chromatic number of Johnson graphs~\cite{godsil_meagher_2015}.
 Upper bounds have been thoroughly studied for special cases,
because they imply lower bounds on codes e.g.~\cite{BE11,EBitan96}. 
In addition to some special cases, only the trivial lower bound implied by the maximal clicks is known. 
Briefly, it is known that $\Nn/2 \leq \chi(J(\Nn,\Na))\leq \Nn$, often the chromatic number is a little bit smaller\footnote{Apparently there is no $\Nn,\Na$ where it is known that $\chi(J(\Nn,\Na))< \Nn-2$. In some special cases the exact number has been determined, Figure~\ref{fig-tableChrJohnsonG}.},  
more specifics 
 are in Appendix~\ref{sec:Johnson}.
Indeed, using coding theory techniques we show the easy result that there is an informative
protocol when $\Nc+\Nr=1$ with $\Nn$ messages (Lemma~\ref{lem:properColBasic}), and the more difficult new result for the general case,  $\Nc+\Nr\geq 1$,
that   $(2\Nn)^{\Nc+\Nr}$ messages suffice, i.e.,
to properly
color $J^{\Nc+\Nr}(\Nn,\Na)$ (Lemma~\ref{lem:properColBasicG}).
It follows that $\Theta( (\Nc+\Nr)\log \Nn)$ bits are needed and sufficient for an informative protocol; the lower bound is implied by the size of the maximal clicks of $J^{\Nc+\Nr}(\Nn,\Na)$,
more details in Section~\ref{sec:modInfo-gen-c}.

 Remarkably, only 1 bit suffices for minimal information transmission, when $\Nb<\floor{\Nn/2}$.
We study the minimal information problem in  Section~\ref{sec:sBitTransm}, where we present this
and other results. We show that if additionally $\Nc\leq \floor{\Nn/2}-2$ the $1$-bit protocol is also safe.
Also, we present a reduction from an informative protocol, showing that
when $\Nc+\Nr=1$, as $\Na$ grows from $3$
up to roughly  $\Nn/2$, the number of different messages 
goes down from $\Nn/3$ to $2$, for a safe and minimally informative protocol. 
We find it surprising that there is a safe minimally informative
protocol for the classic Russian cards case $(3,3,1)$ using $2$ messages  ($\Nn=7)$.
Namely, with a message consisting of only one bit, $A$ can transfer one of her cards to $B$, privately. 

We study the classic Russian cards problem in Section~\ref{sec:classRusCrds}, determined by colorings of $J(7,3)$,
as a   concrete example of the previous ideas. There is an informative and safe solution with $7$
messages (known since~\cite{makarychev2001importance}), and one with $6$ messages~\cite{swansStinson14}.
Namely,    using $6$ messages is optimal, since the chromatic number of $J(7,3)$ is known
to be 6.
There is also  a safe informative solution using $6$ messages 
for the \emph{weak Russian cards} problem, i.e. when $\Nc=0$ and $\Nr=1$.

While the informative property requires that all vertices of each maximal click $K_p(\bar{b})$
are colored differently by $P_A$,  the safety property requires the opposite, that not all vertices 
of each maximal click $K_p(\bar{c})$ are colored differently.
Thus, a protocol $P_A$ can be informative and safe only
if $\Nb>\Nc$. In this case, while $K_p(\bar{c})$ induces a click in $J^{\Nb+\Nr}(\Nn,\Na)$,
it does not induce a click in  $J^{\Nc+\Nr}(\Nn,\Na)$.
\begin{wrapfigure}[8]{r}[-55pt]{0.2\textwidth}
 \centering
 \vspace{-22pt}
 \includegraphics[scale=0.25]{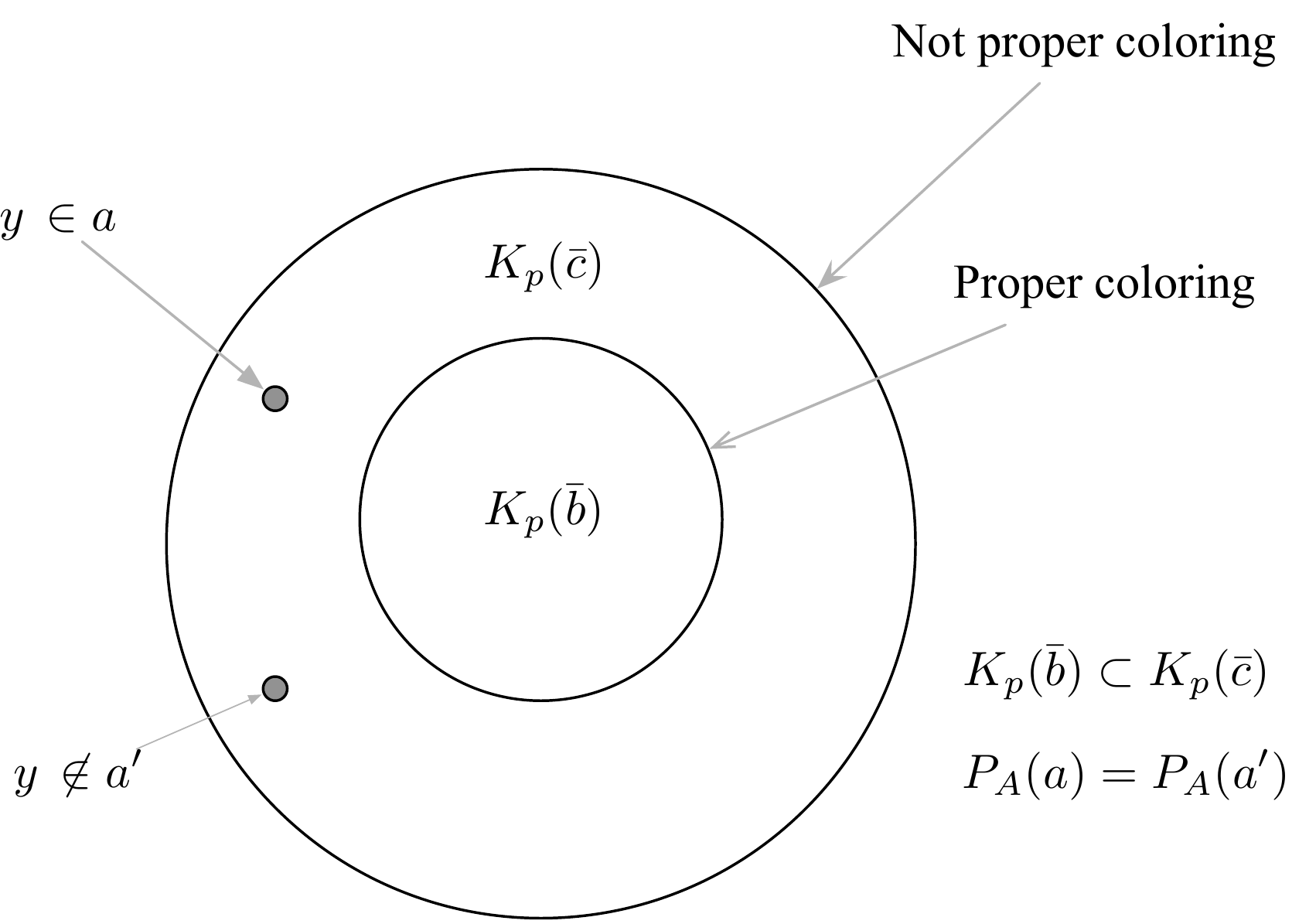}
\end{wrapfigure}
Safety requires that for each card $y$, there is a hand of $A$ that includes $y$,
and another that does not include it, both equally colored, in the complement of the hand of $C$.

We consider the protocol $\chi_{modn}$  in Section~\ref{sec:modularAlgo}, that
sends the sum of the cards modulo $\Nn$, for $\Nc+\Nr=1$,
and show that it is informative and safe, for $\Na,\Nb\geq 3$, $\Nn\geq 7$. 
Indeed, while informative is a
 coding theory property, safety is an additive number theory property.
We prove safety using simple shifting techniques~\cite{godsil_meagher_2015},
getting a    generalization  and simplification of results of~\cite{Cordon-FrancoDF12}.\footnote{
Cord{\'{o}}n{-}Franco et al.~\cite{Cordon-FrancoDF12} show that
$\chi_{modn}$ is safe
when $\Nn$ is prime 
using~\cite[Theorem 4.1]{DaSilvaDH94}, analogous to
the Cauchy-Davenport theorem, except for $(4,3,1)$, $(3,4,1)$.
}
Thus, only two additional messages are needed to make an informative protocol,
also safe (w.r.t. the best known solutions).
We present an informative protocol for the general case $\Nc+\Nr\geq 1$ based
on more involved coding theory ideas and
 discuss safety, in Section~\ref{sec:general-c}, but a detailed treatment is beyond the
 scope of this paper.

\subsubsection{Organization.}
In Section~\ref{sec:basicProb} we present the problems of secure information transmission that we study in this paper.
In Section~\ref{sec:mainProt} we review some basic facts about Johnson graphs, and  rephrase
in such terms the secure information transmission problems.
In Section~\ref{sec:genRussCards} we discuss the relation with the generalized Russian
cards problem, and some basic consequences of our formalization, e.g. there is
a safe proper coloring of $J(\Nn,\Na)$ iff there is a 
safe proper coloring of $J(\Nn,\Nn-\Na)$.
In Section~\ref{sec:classRusCrds} we present the results about six-message solutions for the weak and the classic Russian cards problem, $\Nn=7$.
In Section~\ref{sec:sBitTransm} we present the minimal information transmission results.
In Section~\ref{sec:modularAlgo} we show that $\Nn$ messages are sufficient for
safe, informative information transmission, when $\Nc+\Nr=1$, and the general
case is discussed in Section~\ref{sec:general-c}.
The conclusions are in Section~\ref{sec:concl}.
Additional details are at the end: further related work discussion 
in Appendix~\ref{sec:relatedWork}, Johnson graphs background in
Appendix~\ref{sec:Johnson}, 
additional proofs and figures are 
 in Appendix~\ref{app:lowerBoundI} and~\ref{app:symmModn}.


\section{Secure information transmission}
\label{sec:basicProb}

The model and the problem are defined here,
adapting the distributed computing formalization of~\cite{HerlihyKR:2013} to the case of an eavesdropper.
In Section~\ref{sec:inputC} we present the representation of the inputs to $A,B,C$ as a simplicial complex,
which  determines the Johnson graphs  that will play a central role in this paper.
In Section~\ref{sec:infoSafeP} the notions of protocol, and of a protocol being (minimally) informative and safe are defined.

\subsection{The input complex}
\label{sec:inputC}

Let $\cards=\set{0,\ldots,\Nn-1}$, $\Nn>1$,
 be the \emph{deck} of $\Nn$ distinct cards. An element in the deck is a \emph{card.}
 A subset $x$ of cards is a  \emph{hand,}  $x\in\powerset{\cards}$.
 We may say for short  that $x$, $|x|=m$, 
 is an $m$-set or $m$-hand, namely, if $x\in \powersetS{m}{\cards}$, the subsets of $\cards$ of size $m$.
 A  $\deal=(a,b,c)$ consists of  three disjoint  hands, meaning that
  cards in $a$ are dealt to $A$, cards in $b$ to $B$, and cards in $c$ to $C$.
  We say that the hand is the \emph{input} of the process.
 We call $\signature =(\Na,\Nb,\Nc)$ the \emph{signature} of the $\deal=(a,b,c)$
 if $|a|=\Na$, $|b|=\Nb$ and $|c|=\Nc$, following the notation introduced by Fischer and Wright~\cite{FischerW92crypto}.
We assume that $A$, $B$ and $C$ are aware of the deck and the signature.

It has been often assumed  that $\Nn=\Na + \Nb + \Nc$, but as we shall see,
it is natural to consider the case where nobody gets $\Nr$ cards, $\Nn=\Na + \Nb + \Nc+\Nr$.
While $A$ and $B$ get at least one card, $\Na,\Nb\geq 1$, $C$ may get none $\Nc\geq 0$.

All possible deals  for a given signature over $\cards$ are represented by a simplicial complex.
The vertices are of the form $(Y,y)$, $Y\in\set{A,B,C}$, and $y$ a hand. Such a vertex is called a $Y$-vertex.
 The \emph{input complex} $\cI(\Na,\Nb,\Nc)$, or $\cI$ for short,
 for signature  $\signature =(\Na,\Nb,\Nc)$ is defined as follows.
 The facets of $\cI$ are all the sets $\set{ (A,a),(B,b),(C,c)}$, where $a,b,c$ 
 is a deal of signature $\gamma$.
The input complex $\cI$ consists of all such facets, together with all their subsets.

   Notice that the $A$-vertices of $\cI$ are in a one-to-one correspondence with  
all subsets of size $\Na$ of $\cards$,   $\powersetS{\Na}{\cards}$,
 the $B$-vertices with $\powersetS{\Nb}{\cards}$,
  the $C$-vertices with $\powersetS{\Nc}{\cards}$.
  Indeed, when  $\Nc=0$, there is a single vertex for $C$ in $\cI$.
  

 The left part of  Figure~\ref{fig:basicEx} illustrates  the four $A$-neighbors of vertex $(B,\set{4,5,6})$, in $\cI$,
for signature $(3,3,1)$.
For short, we omit the commas and parenthesis from the set notation, and write $(B,\set{456})$.

\begin{figure}[h]
\centering
\includegraphics[scale=0.4]{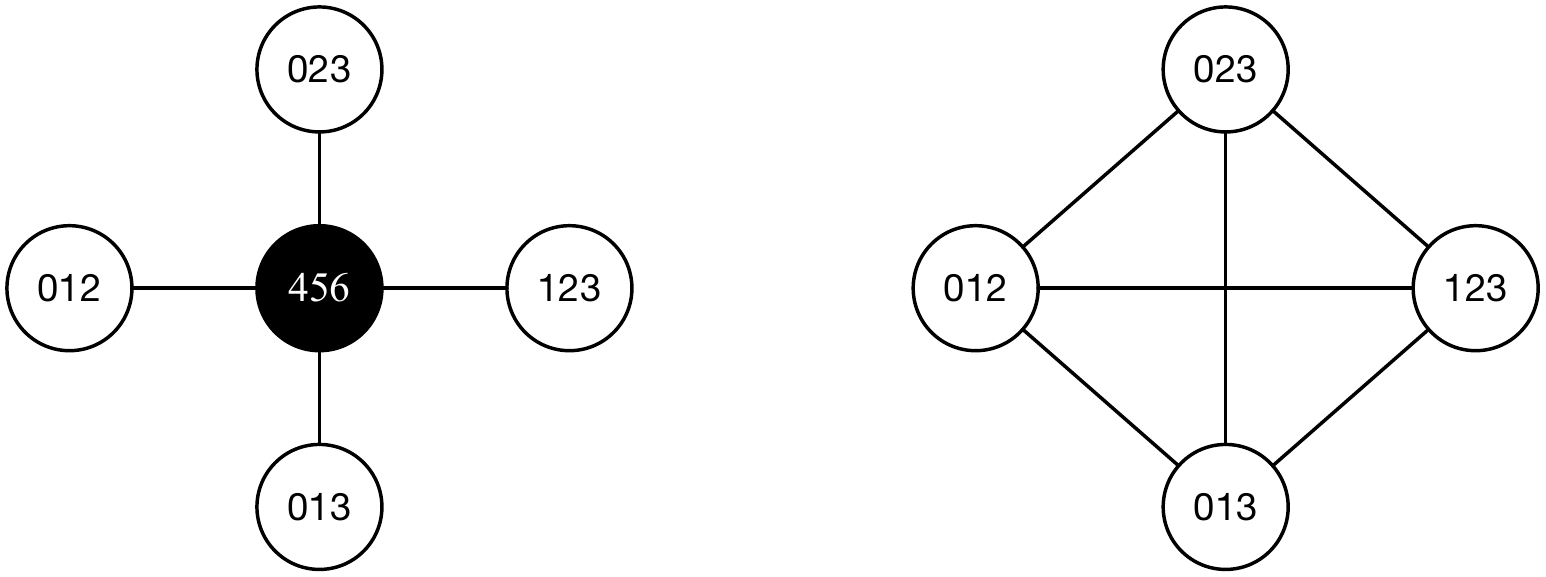}
\caption{White vertices correspond to $A$, and the black vertex correspond to $B$.
The four $A$-neighbours of  $(B,456)$ for signature $(3,3,1)$ form a click on the right, the corresponding
part of $J(7,3)=\cG_B$, defined in Section~\ref{sec:strGA}.}
\label{fig:basicEx}
\end{figure}

\begin{example}
In distributed computing the input complex with a signature $\signature =(1,1,1)$ for three processes
has been considered,
representing that processes get distinct input names from a set of $\Nn$ names~\cite{AttiyaBDPR90}.
The figure from~\cite{HerlihyS99} 
shows that in the case of $\Nn=4$, the complex is a  torus subdivided into triangles.
The vertices of each triangle are colored black, gray, and white
to represent the three different processes. Inside the vertex is the card dealt to the corresponding process.
\begin{figure}[h]
\centering
\includegraphics[scale=0.3]{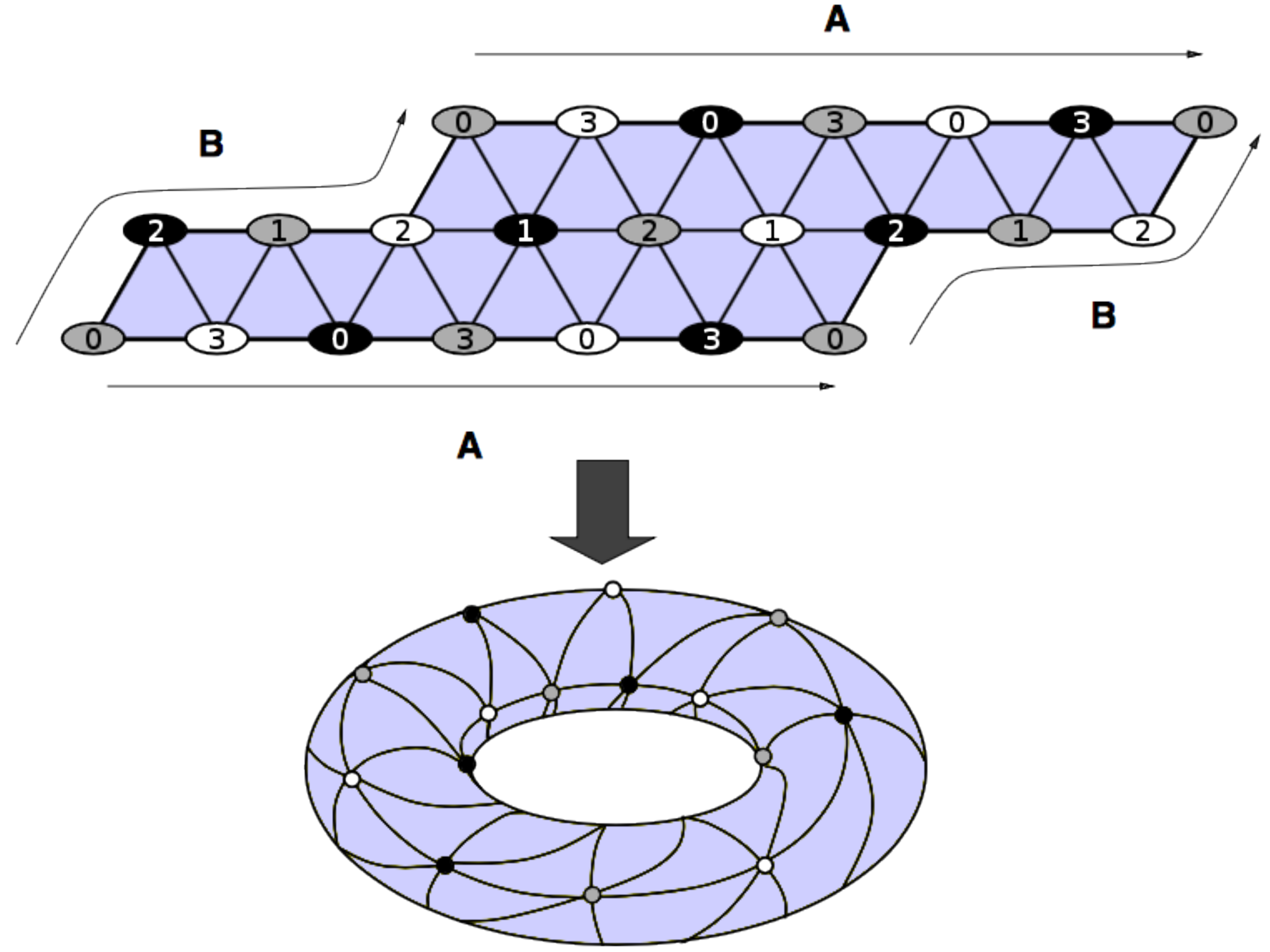}
\label{renaming-torus}
\caption{Input complex for signature $\signature =(1,1,1)$ with  $\Nn=4$ (from~\cite{HerlihyS99}).}
\end{figure}
\end{example}



\subsection{Informative and safe protocols}
\label{sec:infoSafeP}

Fix an {input complex} $\cI$
over $\cards$, $n=\Na+\Nb+\Nc+\Nr$.
In the language of e.g.~\cite{Cordon-FrancoDF12,colorGRSP13,Ditmarsch2011ThreeS}, 
a protocol should be ``informative'' for $B$ and ``safe" from $C$. 
In the case of the Russian cards problem, $B$ should learn the hand of $A$.
We define also the notion of  ``minimally informative."

  The goal is that $B$ learns something about the hand of $A$, after listening
 to an announcement $m$ made by $A$.
The announcement of $A$ is defined by a deterministic function $P_A(a)=M$, for each input vertex $(A,a)\in\cI$,
where $M$ belongs to $\MA$, the domain of possible messages that $A$ may send.
We say that $P_A$ is the \emph{protocol} of $A$.
For $B$, there is a \emph{decision function}  $\delta_B(b,M)$ 
 that produces a set of cards in $\powerset{\cards}$, based on the input $b$
of $B$, and the message $M$ received\footnote{
Since we have fixed  $\cards$ and the input complex $\cI$, implicitly $P_A(a)$ and $\delta_{B}(b,M)$
depend
on these parameters, in addition to the specific input $a$, resp. $(b,M)$. This is what we mean when we say that the players know the
input complex.}.

The minimally informative notion (consider in Section~\ref{sec:sBitTransm})
 requires only that $B$ learns \emph{something} about the hand of $A$.
As we shall see, the least one can expect is that $B$ 
 learns that $A$ has  one of the cards of a set $s$, $|s|=\Nn-\Na-\Nb$.
 Thus, if $n=\Na+\Nb+\Nc+\Nr$, with $\Nc+\Nr=1$, then $B$ should learn one of the cards in
 the hand of $A$. 
 When $\Nb=1$,  $B$ should learn that $A$ has one of the cards in a set $s$, $|s|=\Nn-\Na-1$,
 more than the trivial guess $s$, $s=\cards\setminus b$, where $b$ is $B$'s input card.
 When $\Nc+\Nr=0$ without any communication $B$ knows the hand
 of $A$, so it does not make sense to define a protocol where $B$ learns
 less information. Notice that when $\Nc+\Nr\geq 1$, we have that $\Nn-\Na-\Nb\geq 1$, 
 and the following minimally informative definition makes sense.

\begin{definition}[Informative and minimally informative]
\label{def:inform}
Let  $P_A$ be a protocol.
If  there exists $\delta_B$, such that for any given  input edge  $\set{ (A,a),(B,b)}\in\cI$, with  $M=P_A(a)$,
\begin{itemize}
\item
 $\delta_B(b,M)=a$, the protocol is \emph{informative},
\item 
for $\Nc+\Nr\geq 1$,   $\delta_B(b,M)=s\in\powersetS{\Nc+\Nr}{\cards}$, such that $a\cap s\neq \emptyset$, the protocol is \emph{minimally informative}.
\end{itemize}
\end{definition}
 
The previous definition does not talk about $C$. Indeed, it is based only on the graph which is
the subcomplex of $\cI$ induced by the  $A$-vertices and the $B$-vertices. 
A protocol is {safe} if $C$ cannot tell who holds even a single card (that she does not hold). 
Consider a deal $I=\set{(A,a),(B,b),(C,c)}\in\cI$.
Let  $P_A(a)=M$ be the announcement  sent by $A$, and denote it also by $P_A(I)$.
Two deals $I,I'\in\cI$  are \emph{initially indistinguishable} \cite{Indistinguishability}
to $C$ with input $c$ if  $(C,c)\in I,I'$. And they are \emph{indistinguishable after the protocol},
if additionally $P_A(I)=P_A(I')$.
We require then that for $C$ there are always two   indistinguishable inputs of $A$, $a,a'$,
after the protocol,  such that
 $x\in a$ and $x\not\in a'$ or else $x\not\in a$ and $x\in a'$.
 More precisely, for a vertex $(C,c)$, let $M$ be a \emph{possible} message, namely, 
 such that there exists  $I=\set{(A,a),(B,b),(C,c)}\in\cI$,
 and  $P_A(I)=M$. 
For a hand $c$, let $\bar{c} = \cards\setminus c$, and $\triangle$ the symmetric difference operator. 

\begin{definition}[Safety]
\label{def:safProt}
A protocol $P_A$ is \emph{safe,} if for any $(C,c)$, any $x\in \bar{c}$, and any possible message $M$ for $(C,c)$,
there are edges $I=\set{(A,a),(C,c)}\in\cI$, and
$I'=\set{(A,a'),(C,c)}\in\cI$, with  $P_A(I)=P_A(I')=M$
such that
  $x\in a\triangle a'$.
\end{definition}

Notice that while  $\Na,\Nb\geq 1$, the previous definition applies even when $\Nc=0$.

\begin{remark}[The $\Nc+\Nr\geq 1$ assumption]
\label{rm:NrNc0}
If  $\Nc=\Nr=0$, 
there is a single vertex for $C$ in $\cI$, and each vertex of $A$ and of $B$ belong to a single
triangle; without any communication they know each other hands. Even when $\Na=\Nb=1$,
 the protocol $P_A$ that always sends the same message, is  informative and safe.
\end{remark}

\section{Protocol as vertex coloring}
\label{sec:mainProt}

We represent subcomplexes of $\cI$ as Johnson graphs in Section~\ref{sec:strGA1}, and
some basic facts about these graphs are recalled in Section~\ref{sec:strGA}.
We reformulate the information transmission problem as  properties about vertex colorings
 of Johnson graphs in Section~\ref{sec:code}, and discuss corresponding chromatic number
 notions in Section~\ref{sec:chromNumb}.
 
 \subsection{Representing indistinguishability by Johnson graphs}$\,$
\label{sec:strGA1}

 The situation when $B$ has input $b$ is represented by a vertex $(B,b)\in\cI$.
 The  $A$-vertices that $B$ considers possible with input $b$, are the $A$-neighbors of $(B,b)$ in $\cI$. 
Thus, we define (following~\cite{DFRbits2020}) the graph  $\cG_B$   in terms of $\cI$, as follows.
The vertices of $\cG_B$ consist of all the $A$-vertices of $\cI$.
There is an edge joining two vertices $(A,a),(A,a')$ if and only if there are edges in $\cI$ connecting them
with the same vertex $(B,b)$.  
To analyze $\cG_B$, we omit the id $A$ from the vertices, 
and  let $V(\cG_B)=\powersetS{\Na}{\cards}$.
Thus, for two distinct  $a,a'\in \powersetS{\Na}{\cards}$,  $\{a,a'\} \in E(\cG_B)$ iff 
$\exists b\in \powersetS{\Nb}{\cards}$ such that
$a,a' \subseteq \bar{b} = \cards-b$. 
See Figure~\ref{fig:basicEx}.
 If $\Nr=\Nc=0$, and $\Nn=\Na+\Nb$, there are no two such distinct deals $a,a'$, and the graph has no edges (which is why it makes sense to assume $\Nc+\Nr\geq 1$,  Remark~\ref{rm:NrNc0}).

The graph $\cG_C$ is defined analogously, on the same set of vertices,
$V(\cG_C)=\powersetS{\Na}{\cards}$.
 When $C$ has input $c$ there is a vertex $(C,c)\in\cI$.
 The  $A$-vertices that $C$ considers possible with input $c$, are the $A$-neighbors of $(C,c)$ in $\cI$. 
Thus, for two distinct  $a,a'\in\powersetS{\Na}{\cards}$,  $\{a,a'\} \in E(\cG_C)$ iff 
$\exists c \in \powersetS{\Nc}{\cards}$ such that 
$a,a' \subseteq \bar{c} = \cards-c$.

%
%


\begin{lemma}
\label{lem:basicAdjacency1}
For $a,a'\in V(\cG_B)$, $\Nn= \Na+\Nb+\Nc+\Nr$, $\Nr\geq 0$, we have that
$\{a,a'\} \in E(\cG_B)$ iff $\Na-(\Nc +\Nr) \leq |a \cap a'|$.
Similarly, $\{a,a'\} \in E(\cG_C)$ iff $\Na-(\Nb +\Nr) \leq |a \cap a'|$.
\end{lemma}

\begin{proof}
Recall that   $\{a,a'\} \in E(\cG_B)$ iff $\exists b \subseteq \cards$ 
such that $|b| = \Nb$ and $a,a' \subseteq \bar{b} = \cards-b$. 

Thus,   $\Nb \leq |\cards - (a \cup a')|$.
Now, $|\cards - (a \cup a')| = (\Na + \Nb + \Nc+\Nr) - | a \cup a' |$.
Also, $| a \cup a' | = 2\Na - |a \cap a'|$.
It follows that $\Nb \leq  \Nb+\Nc+\Nr-\Na + |a \cap a'|$.
Finally, $\Na-\Nc -\Nr \leq |a \cap a'|$. 

The argument for $\cG_C$ is similar.
\end{proof}

 \begin{definition}[Distance $d$ Johnson graph]
 \label{def:jg}
For a set of $n$ elements,  the graph $J^d(n,m)$, $0\leq d\leq m$, has as vertices all  
 $m$-subsets. 
 Two vertices $a,a'$  are adjacent whenever $m-d \leq |a \cap a'| $.
 When $d=1$, we have a \emph{Johnson graph,} denoted $J(n,m)$.
\end{definition}

We have our basic theorem, for $\Na,\Nb\geq 1$, $\Nc,\Nr \geq 0$, and $\Nn= \Na+\Nb+\Nc+\Nr$.
The basic, most studied case, is when $\Nc=1,\Nr=0$, or $\Nc=0,\Nr=1$.
\begin{theorem}
\label{th:basicJohns}
The graph $\cG_B$ for signature $(\Na,\Nb,\Nc)$ is equal to the graph  $J^{\Nc+\Nr}(\Nn,\Na)$. 
In particular, $\cG_B$ is  a Johnson graph,
   $J(\Nn,\Na)$, exactly when   $\Nc+\Nr = 1$. 
   Similarly,  $\cG_C$ is equal to   $J^{\Nb+\Nr}(\Nn,\Na)$.
 \end{theorem}

  Notice that  when $d=0$ the graph $J^d(n,m)$  has no edges. 
  Thus, when $\Nc+\Nr=0$ the graph $\cG_B$ has no edges.

 The vertices of $A$ that $B$ considers possible with input $b$, are the $A$-neighbors of $(B,b)$ in $\cI$.  They are denoted $K_{p} (\bar{b})$, where $\bar{b}=\cards- b$.
They  induce a
 click in $\cG_B$ (overloading notation the click itself is also sometimes denoted  by $K_{p} (\bar{b})$).
 The vertices in $K_{p} (\bar{b})$ are all $a\subseteq \bar{b}$ with $|a|=\Na$.  
Thus, when $B$ has input $b$, $B$ considers possible that $A$ has any input  $a$, $a\in K_{p} (\bar{b})$.
Notice that if $\Nc+\Nr=0$ and $\Nn=\Na+\Nb$, then $B$ with input $b$ considers possible only one input for $A$, namely, $\bar{b}$. In this case, $E(\cG_B)=\emptyset$.

\begin{lemma}
\label{lem:neighbrs}
For each hand $b$ of $B$, the possible inputs of $A$ induce a   click $K_{p} (\bar{b})$ in $\cG_B$,
 $p= \binom{\Nn-\Nb}{\Na} $, consisting of all $a\in\powersetS{\Na}{\cards}$, such that 
 $a\subset \bar{b}$.
 Similarly, for $\cG_C$,  the vertices $K_{p} (\bar{c})$ consisting of all $a\in\powersetS{\Na}{\cards}$ such that $a\subset \bar{c}$,  induce a click in $\cG_C$.
\end{lemma}

We have illustrated the following in the figure of the Introduction.

 \begin{remark}[Subgraphs]
 \label{rem:Clicks}
If $\Nb\leq\Nc$ then   $J^{\Nb+\Nr}(\Nn,\Na)$ is a subgraph of $J^{\Nc+\Nr}(\Nn,\Na)$
on the same set of vertices. Hence,
for each $b\in\powersetS{\Nb}{\cards}$, $c\in\powersetS{\Nc}{\cards}$,
both $K_{p} (\bar{b})$ and $K_{p} (\bar{c})$ induce clicks in $J^{\Nc+\Nr}(\Nn,\Na)$.
Furthermore, if  $b\subseteq c$,  then $K_{p} (\bar{c})\subseteq K_{p} (\bar{b})$.
 \end{remark}
 
\subsection{Johnson graphs}
\label{sec:strGA}

Johnson graphs have been thoroughly studied, see Appendix~\ref{sec:Johnson}. 
We recall some basic notions here,
which are especially relevant to this paper. 


The vertices of a \emph{Johnson graph} $J(n,m)$ consist of 
  the $m$-element subsets of an $n$-element set; two vertices are adjacent when the intersection of the two vertices consists of $(m-1)$-elements.
We need the  distance $d$ version, $J^d(n,m)$ of Definition~\ref{def:jg}. When $d=1$,  $J^d(n,m)= J(n,m)$.

Let $\delta(a,a')$ denote the distance between vertices $a,a'$ in $J(n,m)$. Then, $\delta(a,a')=k$
iff $|a\cap a'|=m-k$. Or, in terms of symmetric difference, $\delta(a,a')=k$
iff $|a\triangle a'|=2k$.
One can show by induction that $J(n,m)$ has 
diameter $\min\set{m, n-m}$. Thus, for all  $d\geq \min\set{m, n-m}$,
$J^d(n,m)$ is the complete graph on ${{n}\choose{m}}$ vertices.

It is easy to see and well-known that   $J(n,m)$ is isomorphic to $J(n,n-m)$.
The  same holds for the distance $d$ version.

\begin{lemma}
\label{lem:isomorphGenJ}
The following are isomorphic graphs $J^d(n,m)\cong  J^d(n,n-m)$.
\end{lemma}
\begin{proof}
Consider vertices  $a,b$ of $J^d(n,m)$, and their complements $\bar{a},\bar{b}$.
Thus, $|a|=|b|=m$, and $|\bar{a}|=|\bar{b}|=n-m$. The isomorphism $f$ is $f(a)=\bar{a}$
and  $f(b)=\bar{b}$. By definition,  $m-d \leq |a \cap b| \leq m -1$ iff $(a,b)\in E(J^d(n,m))$.
Let $k= |a \cap b|$. Then, $|\bar{a} \cap \bar{b}| = n-m-k$,  
hence,  $n-m-d \leq n-m-k \leq n-m -1$, so   $(\bar{a},\bar{b})\in E(J^d(n,n-m))$, and the lemma follows.
\end{proof}

\begin{remark}[Maximal clicks]
\label{rm:maximalityClicks}
There are two  families of maximal cliques in  $J(n,m)$.
 For the first, take all $n-m+1$ of the $m$-subsets that contain a fixed $(m-1)$-subset; 
 for the second, take the $m$-subsets of a fixed set of size $m+1$. 
 When $n=2m$ the cliques in these two families have the same size. 
 Maximality of the cliques is implied by Erd\"os--Ko--Rado Theorem~\cite[Chapter 6]{godsil_meagher_2015}.
 In the case of $J^d(n,m)$, we have already encountered one family in Lemma~\ref{lem:neighbrs}.
 For each $(m+d)$-subset
 $\bar{b}$, there is a 
 click in  $J^d(n,m)$, denoted  $K_{p} (\bar{b})$.
 The vertices of $K_{p} (\bar{b})$ are all $m$-subsets of $\bar{b}$.
 We will encounter the other family as well, $K'_p(b)$.
A click $K'_p(b)$ is obtained by taking
  the $m$-subsets that contain a fixed $(m-d)$-subset $b$.
\end{remark}

We  recall a simple but useful \emph{shifting}
technique in Johnson graphs, and even more generally in intersecting set 
families~\cite{godsil_meagher_2015}, we use the following version.
For a hand $a$, and cards $i,j$, with  $i\not\in a$, $j\in a$, 
$$
a_{ij}=(a\setminus j)\cup \set{i},
$$
denoted by an arc 
$
a \stackrel{ij}{\longrightarrow} a_{ij}.
$
Notice that, $\set{a,a_{ij}} \in E(J(n,m))$, and if $a'$ is reachable from $a$ by $d$ arcs, then 
$\set{a,a'} \in E(J^d(n,m))$. 

For a hand $s$, we say that $a'$ is $s$-\emph{reachable} from $a$ if there
is a directed path from $a$ to $a'$ defined by a (possibly empty) sequence of arcs $\stackrel{ij}{\longrightarrow}$, all of them with $i\in s$. (For the following cf.~\cite[Lemma 1]{RamrasD2011}).

\begin{lemma}
\label{binInterG0}
Let  $a\in V(K_{p} (\bar{b}))$.
Let  $s=\bar{b}\setminus a$. Thus, $|s|=d$.
Then,  $V(K_{p} (\bar{b}))$ is the set of $s$-reachable vertices from $a$.
\end{lemma}
\begin{proof}
First, notice that  $a$ is $s$-reachable from itself.
Now, 
let $a'$ be any other vertex of $K_{p} (\bar{b})$.
If $2d'= |a\triangle a'|$, $d'\leq d$, order the cards in $a\setminus a'$ as $x_1,\ldots, x_{d'}$
and those in  $a'\setminus a$ as $x'_1,\ldots, x'_{d'}$.
Then, $a'$ is reachable from $a$ by the path
$$
a=a_0 \stackrel{x'_1 x_1}{\longrightarrow} a_1 \stackrel{x'_2 x_2}{\longrightarrow} a_2 \cdots\stackrel{x'_{d'} x_{d'}}{\longrightarrow}
a_{d'}=a'.
$$
\end{proof}

We will need the following claims.

\begin{lemma}
\label{binInterG}
Let $K_{p} (\bar{b})$  be a click of $J^{d}(n,m)$.
For any  set of $k$ vertices, $1\leq k <  p $, $\set{a_1,\ldots,a_k} \subset K_{p} (\bar{b})$,   
there exists a set $s \subset \bar{b}$, $ |s|= d$,
such that
  for any $a_i$,  $a_i \cap s \neq \emptyset$. 
\end{lemma}

\begin{proof}
Pick $a\in K_{p} (\bar{b})$  not in $\set{a_i}$.
Let  $s=\bar{b}\setminus a$, $|s|=d$.
Since  $K_{p} (\bar{b})$ is the set of $s$-reachable vertices from $a$ (Lemma~\ref{binInterG0}),
 all  other vertices  in 
 $K_{p} (\bar{b})$ are $s$-reachable from $a$, $s=\bar{b}\setminus a$.
 And hence, for the subset $\set{a_i}$ of those vertices, we have that 
   for any $a_i$,  $a_i \cap s \neq \emptyset$. 
%
\end{proof}

In particular, when $d=1$, the following holds.

\begin{lemma}
\label{cor:d1min}
Consider $J(n,m)$ and any $K_{m+1} (\bar{b})$.
For any  set of $k$ vertices, $1\leq k \leq  m+1 $, $\set{a_1,\ldots,a_k} \subseteq K_{m+1} (\bar{b})$,
it holds  that  $|\cap a_i |=m+1 -k$. 
\end{lemma}
\begin{proof}
Consider the $a_i$ vertices in order $a_1,\ldots, a_{k}$, and the shiftings
$$
a_1 \stackrel{x'_1 x_1}{\longrightarrow} a_2 \stackrel{x'_2 x_2}{\longrightarrow} a_3 \cdots a_{k-1}\stackrel{x'_{k-1} x_{x-1}}{\longrightarrow}
a_{k},
$$
where 
$a_{i+1}\setminus a_i = x'_i$ 
and   $a_i\setminus a_{i+1}= x_i$.
Thus, by induction on $i$, for each $i\geq 1$, $| a_1\cap a_2\cap  \ldots \cap a_i| = m+1-i$. 
\end{proof}


\subsection{Protocol as vertex coloring of a Johnson graph}
\label{sec:code}

Consider a protocol $P_A$ for signature $(\Na,\Nb,\Nc)$, with $\Nn=\Na+\Nb+\Nc+\Nr$.
In light of Theorem~\ref{th:basicJohns}, 
we take the view of 
 $P_A$  as a vertex coloring, 
$P_A: \powersetS{\Na}{\cards} \rightarrow \MA$.
For vertex $(A,a)\in\cI$,   $P_A(a)$ is the message $M\in \MA$, sent by $A$ when she has input $a$.
We assume that $P_A$ is surjective.
The set of $A$-vertices colored $M$
is  $ P_A^{-1}(M)$.\footnote{Thus,  $ P_A^{-1}(M)$ is equivalent to an ``announcement'' by $A$ in the terminology
of~\cite{Albert2005SafeCF}, or the ``alternative hands" for $A$, in the notation of~\cite[Proposition 24]{DitmarschRC03}.}

Recall that a vertex coloring of a graph is \emph{proper} if each pair of adjacent vertices have different colors.
The following theorems  reformulate  the informative and safety
notions of Definitions~\ref{def:inform} and~\ref{def:safProt}.  

\begin{theorem}[Informative characterization]
\label{def:prColoring}
Let $P_A: \powersetS{\Na}{\cards} \rightarrow \MA$
be a protocol.
\begin{itemize}
\item
\label{color:proper}
$P_A$ is informative if and only if $P_A$ is a {proper vertex coloring} of $J^{\Nc+\Nr}(\Nn,\Na)$.
\item
\label{color:properMin}
When $\Nc+\Nr\geq 1$, $P_A$ is minimally informative if and only if 
for each $b\in\powersetS{\Nb}{\cards}$
there is  some edge  $\set{a,a'}$ in the click $K_{p} (\bar{b})$ of  $J^{\Nc+\Nr}(\Nn,\Na)$,
such that $ P_A(a)\neq P_A(a')$.

\end{itemize}
\end{theorem}

\begin{proof}
The first condition (informative)
is clearly necessary for the protocol to be informative; if
there is a vertex $(B,b)$ such that two neighbours $(A,a),(A,a')$ have the same color, $M$,
then $B$ cannot distinguish them, produces the same output, $\delta_B(b,M)$. 
Conversely, if all vertices in $K_p(\bar{b})$ have
different colors, then $B$ with hand $b$ will learn  the hand of $A$.
More formally, there is a function $\pi_b$ of the colors of the $A$-neighbour of $(B,b)$, for each 
 $(B,b)\in\cI$,  known a priori to $B$, such that $\pi_b(M)=a$ when $P_A(a)=M$. 
The decision function for $B$ is $\delta_B(b,M)=\pi_b(M)$.

The second condition (minimally informative) is defined only when $\Nc+\Nr\geq 1$,
and hence $K_{p} (\bar{b})$ has at least two vertices.
The condition is clearly necessary, otherwise, when $B$ has input $b$, he 
 will output the same value on all of $A$ possible hands (and there are at least two), independently of what the hand of $A$ is. 
 If $B$'s output is a set $s$, $|s|=\Nn-\Na-\Nb$, then it could be that the input of $A$
 was actually $a\subset \cards \setminus s$, $|a|=\Na$. 
Conversely, let $V_M\subset K_{p} (\bar{b})$ be the subset of vertices $a_i$ such that $P_A(a_i)=M$.
Notice that $0< |V_M|< p$, since there is 
 an edge $\set{a,a'}\in E(K_{p} (\bar{b}))$ with
$\chi(a)\neq\chi(a')$. 
By Lemma~\ref{binInterG} 
there exists a set $s\subset \bar{b}$, $ |s|= \Nc+\Nr=\Nn-\Na-\Nb$, 
such that
  for any $a_i$,  $a_i \cap s \neq \emptyset$. 
Thus, we may define $\delta_B(b,M)=s$.
 \end{proof}

\begin{remark}[Informative]
\label{rm:informaBasic}
Some observations of the informative reformulation.
\begin{itemize}
\item Each edge of $J^{\Nc+\Nr}(\Nn,\Na)$ is in some click
$K_{p} (\bar{b})$. Thus, $ P_A$ being a \emph{proper} vertex coloring is equivalent
to the property that for
  all edges  $\set{a,a'}\in E(K_{p} (\bar{b}))$, it holds that $P_A(a)\neq P_A(a')$, for any such click.
  In contrast, the minimally informative property requires only that not all edges of each click
  have both endpoints colored equally.

\item By Lemma~\ref{lem:isomorphGenJ},  
  $J^{\Nc+\Nr}(\Nn,\Na) \cong  J^{\Nc+\Nr}(\Nn,\Nn-\Na)$, thus  
 there is an informative protocol   for one if and only if there is an informative
 protocol for the other. 
 This equivalence does not generally hold for minimally informative protocols,
 e.g.  the protocol $\chi_2$ of  Section~\ref{sec:minInfTransf}.
 
 The reason is that a click $K_p(\bar{b})$  in $J^{\Nc+\Nr}(\Nn,\Nn-\Na)$
 translates into a click $K'_p(b)$ in  $J^{\Nc+\Nr}(\Nn,\Na)$ (see Remark~\ref{rm:maximalityClicks}).
When $\Nc+\Nr=1$ safety is  preserved, see Theorem~\ref{th:dual}.
  \item
If $d=\Nc+\Nr$, $d'=\Nc'+x'$, and
  $d\leq d'$, then  $J^{d}(\Nn,\Na)$ is a subgraph of $J^{d'}(\Nn,\Na)$ (Remark~\ref{rem:Clicks}). Thus,  
if $P_A$ is a proper vertex coloring of  $J^{d'}(\Nn,\Na)$ then it is also a proper vertex coloring of 
$J^{d}(\Nn,\Na)$ (similarly, for $\Nn'>\Nn$). 
\end{itemize}
\end{remark}

We have the following  special case.
 By Lemma~\ref{cor:d1min}, 
for the case when $\Nc+\Nr=1$ (recall Theorem~\ref{th:basicJohns}), we have
that $B$ learns at least one card of $A$.

\begin{lemma}
\label{lemm:bit}
Let $\Nc+\Nr=1$.
For a minimally informative  protocol  $P_A$,
there exists a decision function for $B$,  $\delta_B$, such that when the hand of $A$ is $a$ and $P_A(a)=M$, then
$\delta_B(b,M)=x$, for some $x\in a$.
\end{lemma}

Recall from Section~\ref{sec:strGA1}  the  graph $\cG_C$.
The vertices of $\cG_C$ consist of all the $A$-vertices of $\cI$.
There is an edge joining two vertices $(A,a),(A,a')$ if and only if there are edges in $\cI$ connecting them
with the same vertex $(C,c)$.  
Then,  $V(\cG_C)=V(\cG_B)=\powersetS{\Na}{\cards}$, and
for two distinct hands $a,a'$ of size $\Na$,  $\{a,a'\} \in E(\cG_C)$ iff 
$\exists c \in \powersetS{\Nc}{\cards}$ such that 
$a,a' \subseteq \bar{c} = \cards-c$. Namely, we have the graph $J^{\Nb+\Nr}(\Nn,\Na)$,
where $K_p(\bar{c})$ induces a click, for every  $c\in\powersetS{\Nc}{\cards}$.
In the following the set of colors of vertices of a click is denoted,
 $P_A(K_p(\bar{c}))=\set{M ~|~ P(a)=M, a\in K_p(\bar{c})}$.

\begin{theorem}[Safety characterization]
\label{th:safetyMain}
Let $P_A: \powersetS{\Na}{\cards} \rightarrow \MA$.
The following  conditions are equivalent.
\begin{enumerate}
\item \label{th:safetyMain1}
$P_A$ is {safe}.
\item  \label{th:safetyMain2}
Consider any  $c\in\powersetS{\Nc}{\cards}$, and any
  $y\in\bar{c}$. For each $M\in P_A(K_p(\bar{c}))$, 
there exist $a,a'\in  K_p(\bar{c})$ with $P_A(a)=P_A(a')=M$ such that $y\in a\triangle a'$.
\end{enumerate}
\end{theorem}

\begin{proof}
The   equivalence is straightforward, recalling the one-to-one correspondence between hands $c$, $|c|=\Nc$
and $C$-vertices of $\cI$, and observing that $c\cap a$ for an $A$-hand $a$ is equivalent
to the existence of a deal $I\in\cI$ including $(A,a),(C,c)$. Indeed, for any  $a$  in $P_A^{-1}(M)$, there exists
one $c$, $|c|=\Nc$ with $c\cap  a=\emptyset$.

%
%
\end{proof}

\begin{remark}[Safety]
\label{rm:subgraphsM}
\begin{itemize}
\item
Informative requires $P_A$ to be a proper vertex coloring of $J^{\Nc+\Nr}(\Nn,\Na)$,
while safety requires that $P_A$ is not a proper vertex coloring of $J^{\Nb+\Nr}(\Nn,\Na)$.
\item
Thus, by Remark~\ref{rem:Clicks}, a protocol can be informative and safe only
if $\Nb>\Nc$. In this case, while $K_p(\bar{c})$ induces a click in $J^{\Nb+\Nr}(\Nn,\Na)$,
it does not induce a click in  $J^{\Nc+\Nr}(\Nn,\Na)$, by Remark~\ref{rm:maximalityClicks}.
(cf.~\cite[Lemma 2]{Albert2005SafeCF}).
\item
Joining color classes $P_A^{-1}[M]\cup P_A^{-1}[M']$ of a protocol preserves
safety, but not necessarily informative properties
(see Section~\ref{sec:minInfRussCards}).
\end{itemize} 
\end{remark}

Notice that it could be that there is a hand $c$ for $C$, 
for which some message $M$ is never sent by $P_A$.
But as was observed in~\cite[Proposition 6]{Cordon-FrancoDF12},  
 with protocols that send the sum of the cards modulo $\Nn$ 
this is not the case, see  Section~\ref{sec:sBitTransm} and~\ref{sec:modularAlgo}.


 The following argument is similar to~\cite[Proposition 29]{DitmarschRC03}.

\begin{lemma}
\label{lm:doubleCsafety}
Let  $\Na\geq 2$, $\Nc\geq 1$,  $P_A$ be a safe protocol.
 Consider any $M$.
For any vertex $a\in P_A^{-1}(M)$,     any  $z\in a$, 
and any card $y$,
there must be another vertex $a'\in P_A^{-1}(M)$ that also includes  card $z$,
and $y\in a\triangle a'$.
\end{lemma}

\begin{proof}
Suppose that $y\in a$ (the other case is similar).
Since $P_A$ is safe, there must be another vertex in $P_A^{-1}(M)$ that does not include $y$.
Consider all vertices in $P_A^{-1}(M)$ that do not include $y$, denoted  $V_{\bar y}$ .
If one of them also includes $z$ we are done.

Thus, suppose that none of them contains $z$. Let $a'\in V_{\bar y}$ be such that $y\not\in a'$.
Thus, $z$ is also not in $a'$ (else we are done).  

Consider a $\Nc$-hand $c$
 that contains $y$ in the complement of $a'$. Thus,  $C$ with hand $c$ may hear $M$, but if so she knows
that $A$ does not have card $z$, a contradiction to the safety of $P_A$.

%
%
%
%
\end{proof}

%
 Notice that $\Nc\geq 1$ is necessary,
  otherwise,  Lemma~\ref{lm:doubleCsafety} may not hold; an
  example is  protocol $\chi_1$ of Theorem~\ref{th:russCardPrWeak}. And clearly, $\Na\geq 2$ is also necessary for the lemma to hold.
  
\begin{remark}[The assumption $\Na\geq 2$]
\label{lem:simpleSafe1}
A simple consequence of Theorem~\ref{th:safetyMain} is that we should concentrate on the case that $\Na\geq 2$.
If $\Na=1$ then a safe protocol $P_A$ must always send the same message $M$.
Otherwise, if $P_A(y)\neq P_A(y')$ for $y,y'\in\cards$, then
when $C$ has a hand $c$, such that $y,y'\in \bar{c}$, then when $C$  hears $P_A(y)$
she knows that $A$ does not have card $y'$. Thus a safe protocol $P_A$ cannot be minimally informative,
and thus cannot be informative either.
\end{remark}


\subsection{Chromatic numbers}
\label{sec:chromNumb}

For an informative, not necessarily  safe protocol, 
the minimum number of bits to communicate her full hand
is $\log_2 \chi$, where $\chi$ is the chromatic number of $J^{\Nc+\Nr}(\Nn,\Na)$. 
In the
case of $\Nc+\Nr = 1$, namely a Johnson graph, computing the chromatic number  is an 
important open question e.g.~\cite[Chapter 16]{godsil_meagher_2015}.
It is however known that 
$\Nn/2 \leq \chi(J(\Nn,\Na))\leq \Nn$ and hence, when $\Nc+\Nr=1$, 
the number of bits necessary and sufficient for an informative protocol is $\Theta(\log\Nn)$.
We show in
Section~\ref{sec:modInfo-gen-c} that in general, 
the number of bit is $\Theta( (\Nc+\Nr)\log \Nn)$.

The \emph{safe chromatic number} of $J^{d}(n,m)$, $d=\Nc+\Nr$, denoted $\chi^{sf}$, is the 
cardinality of the smallest color set $\MA$ for which the graph has a safe proper coloring,
or $\infty$ if no such coloring exists.  We will see cases where it is $\infty$ in
Theorem~\ref{th:mainBoundCases}.
Recall that the safety property depends on $\Nc$, which is why we
have to specify that $d=\Nc+\Nr$. For the same $\Nc$, we have that
$ \chi\leq \chi^{sf}$. As we shall see in Section~\ref{sec:classRusCrds}, there are cases where 
$ \chi< \chi^{sf}$, namely, $\chi(J(7,3))=6$ and $\chi(J(7,3))=7$.

Similarly, $\chi_{min}$ is the  cardinality of the smallest color set $\MA$ for which the graph has a minimal informative
coloring, and if we require additionally safety, then it is denoted $\chi_{min}^{sf}$.
Thus, $\chi^{min}\leq \chi_{min}^{sf}\leq \chi^{sf}$.
We will see that $ \chi_{min}^{sf}$ can be much smaller than $\chi^{sf}$.
In an extreme case, for $\Nn$ even, we have that
$\chi_{min}^{sf} (J(\Nn,\Nn/2))=2$  (Corollary~\ref{cor:minInfo1}),
while $\chi^{sf} (J(\Nn,\Nn/2))\geq \chi (J(\Nn,\Nn/2)) >\Nn/2$  
(since $\chi(J(n,m)) \geq\max{ \set{n-m+1,m+1}}$, see Appendix~\ref{sec:Johnson}).

\section{Russian cards problems}
\label{sec:genRussCards}

In Section~\ref{sec:basicsRusCard} we
present the generalized Russian cards problem and discuss its relation with
our information transmission problem. 
Some general bounds that will be useful later on are in Section~\ref{sec:boundsRusCard}.

\subsection{The problem statement}
\label{sec:basicsRusCard}

The Russian cards problem has signature $(3,3,1)$, and the 
generalized Russian cards problem  has
signature $(\Na,\Nb,\Nc)$.
The players $A$, $B$ and $C$ each draw $\Na$, $\Nb$ and $\Nc$ cards, respectively, from the deck $\cards$ of 
$\Nn=\Na + \Nb + \Nc+\Nr$ cards.
In this context,  \emph{two-step protocols} have been thoroughly studied, usually when $\Nr=0$. 
First $A$ and then $B$ makes an announcement, both
 heard by $C$. 
If a protocol $P_A$ is informative and safe, and $\Nr=0$,  one may  assume that $P_B$,
the protocol of $B$, is simply to announce $C$'s set of cards.
First, since the protocol $P_A$ is informative, $B$ knows the cards of $A$ after $A$'s announcement, 
and hence he can deduce the cards of $C$. After the announcement  $P_B$, $A$ can deduce the cards
of $B$. The announcement made
by $B$ is the set of cards of $C$, and hence does not give any new information to $C$.

We consider also the case where $\Nr >0$.  Then on input $b$, once $B$ learns the hand $a$ of $A$, 
he announces $\cards\setminus (a\cup b)$, a superset of $C$'s hand. 
We work under this security assumption for $\Nr>0$.
Namely, that we allow  $C$ to learn  only cards that are not held by either A or B.
Among the cards held jointly by $A$ and $B$, she does not learn who holds which card.
With this clarification, we continue  to focus only in $P_A$, the protocol of $A$.
 We have the following consequence of 
Theorem~\ref{def:prColoring}.

\begin{theorem}
\label{th:mainRussC}
There is a $2$-step solution for the Russian problem $(\Na,\Nb,\Nc)$,  $\Nn=\Na + \Nb + \Nc+\Nr$ with $A$ making
the first announcement,
 if and only if
 there is a  safe proper coloring of $J^{\Nc+\Nr}(\Nn,\Na)$.
\end{theorem}

\subsection{General bounds}
\label{sec:boundsRusCard}

In light of Theorem~\ref{th:mainRussC}, we keep on
presenting our results in terms of safe proper colorings of $J^{\Nc+\Nr}(\Nn,\Na)$,
but one should keep in mind that they are all bounds on when there is a
 $2$-step solution for the Russian problem $(\Na,\Nb,\Nc)$,  $\Nn=\Na + \Nb + \Nc+\Nr$ with $A$ making
the first announcement.


Recall that when  $\Nc+\Nr=1$ there are two cases: $\Nc=1,\Nr=0$, and $\Nc=0,\Nr=1$.
Thus, the fallowing has two instantiations.
 The cases $(\Na,\Nb,1)$ and  $(\Nb+1,\Na-1,1)$,
and the cases $(\Na,\Nb,0)$ and $(\Nb+1,\Na-1,0)$.
Remarkably, the following result does not hold for minimally informative protocols
(see Corollary~\ref{cor:minInfo1}).

For a protocol 
$P_A: \powersetS{\Na}{\cards} \rightarrow \MA$, 
the protocol
$\bar{P}_A: \powersetS{\Nn-\Na}{\cards} \rightarrow \MA$
 is defined by
$$
\bar{P}_A(a)= P_A (\bar{a}),
$$
where as usual, $\bar{a}=\cards\setminus a$.

The following shows that there is a safe proper coloring of $J(\Nn,\Na)$ iff there is a 
safe proper coloring of $J(\Nn,\Nn-\Na)$.

\begin{theorem}[duality]
\label{th:dual}
Assume $\Nc+\Nr=1$,
so  $\Nn=\Na + \Nb + 1$.
A protocol $P_A$ is  informative and safe  for $(\Na,\Nb,\Nc)$ 
if and only if the protocol $\bar{P}_A$ is  informative and safe   for  
$(\Nb+1,\Na-1,\Nc)$.
\end{theorem}

\begin{proof}
There are two cases: $\Nc=1,\Nr=0$, and $\Nc=0,\Nr=1$.
First we show the equivalence for the informative property, in both cases.

Notice that $\Nn-\Na=\Nb+1$. 
By Lemma~\ref{lem:isomorphGenJ}, we have that  $J(\Nn,\Na)\cong  J(\Nn,\Nn-\Na)$,
under the isomorphism $f(a)=\bar{a}$.
Thus, if $P_A$ is an  informative, i.e.,   proper
vertex coloring of $J(\Nn,\Na)$, then $\bar{P}_A(a)= P_A (f(a))$ is 
a proper vertex coloring of $ J(\Nn,\Nn-\Na)$.

Now, consider the case $\Nc=1,\Nr=0$, and assume that $P_A$ is  safe for $(\Na,\Nb,1)$. 
That is,
for every  card $c\in\powersetS{\Nc}{\cards}$,  $\Nc=1$, 
$y\in\bar{c}$, and $M\in P_A(K_p(\bar{c}))$,
there exists $a,a' \in K_p(\bar{c})$, $P_A(a)=P_A(a')=M$ such that $y\in a\triangle a'$.

To prove that $\bar{P}_A$ is safe, we need to consider a card  $c\in\powersetS{\Nc}{\cards}$, and the vertices of   $K'_p(\bar{c})$
in $J(\Nn,\Nn-\Na)$, which are $\bar{a} \in\powersetS{\Nn-\Na}{\cards}$, such that $\bar{a}\subseteq \bar{c}$.

Let $y\in\bar{c}$, $\bar{a}\in K'_p(\bar{c})$ with
 $\bar{P}_A(\bar{a})=M$. Suppose $y\in a$ (the case when $y\not\in a$ is similar).

 Thus, $P_A(a)=M$ and $c\in a$.
By Lemma~\ref{lm:doubleCsafety} 
there exists $a' \in\powersetS{\Na}{\cards}$, $y\not\in a'$, $P_A(a')=M$,
such that 
$c\in  a'$.

Now, let   $a' \in\powersetS{\Na}{\cards}$, $y\not\in a'$, $P_A(a')=M$, with $c\in a'$.
Then, $c$ is in both $a$ and $a'$, and hence $c$ is in neither $\bar{a}$ nor $\bar{a'}$.
Namely, $\bar{a},\bar{a'}\in K'_p(\bar{c})$. But $\bar{P}_A(\bar{a})=P_A(\bar{a'})=M$.
And we are done, because $y\in \bar{a}\triangle \bar{a'}$.

For the converse, assume $P_A$ is safe for $(\Nb+1,\Na-1,1)= (\Nn-\Na,\Na-1,1)$,
and consider 
 $c\in\powersetS{\Nc}{\cards}$, and the vertices of   $K_p(\bar{c})$
in $J(\Nn,\Na)$, which are ${a} \in\powersetS{\Na}{\cards}$, such that ${a}\subseteq \bar{c}$.

Let $y\in\bar{c}$, ${a}\in K_p(\bar{c})$ with
 $P_A({a})=M$. Suppose $y\in a$ (the case when $y\not\in a$ is similar).

Consider $\bar{a}$, and hence $\bar{P}_A(\bar{a})=P_A(a)$.
Thus, $c\in\bar{a}$.
By Lemma~\ref{lm:doubleCsafety} 
there exists $\bar{a}' \in\powersetS{\Nn-\Na}{\cards}$, $y\not\in \bar{a}'$, $\bar{P}_A(\bar{a}')=M$,
such that 
$c\in \bar{a}'$.

Then, $c$ is in both $\bar{a}$ and $\bar{a}'$, and hence $c$ is in neither ${a}$ nor ${a'}$.
Namely, ${a},{a'}\in K_p(\bar{c})$. But $P_A({a})=P_A({a'})=M$.
And we are done, because $y\in {a}\triangle {a'}$.

Finally, we prove the safety equivalence, for the second case, where $\Nc=0,\Nr=1$.
Solving the weak Russian cards problem for the case $(\Na,\Nb,0)$ is equivalent to solving it for the case 
$(\Nb+1,\Na-1,0)$.
This case is easier, it does not need Lemma~\ref{lm:doubleCsafety}.
If $P_A$ is  safe for $(\Na,\Nb,0)$, then we take $c$ and $\bar{c}$ as the empty set.
Then, for any $y\in\cards$, and $M$,
there exists $a,a'$ such that $P_A(a)=P_A(a')$, 
such that $y\in a\triangle a'$.
Then, $y\in \bar{a}\triangle \bar{a}'$, which is what is needed for $\bar{P}_A$ to be safe, since 
$\bar{P}_A(\bar{a}) = \bar{P}_A( \bar{a}')$.

And the converse is the same. 
If $\bar{P}_A$ is safe, then for every $M$, and any $y$,
it holds $y\in \bar{a}\triangle \bar{a}'$ for some $\bar{a}, \bar{a}'$ of size $\Nn-\Na$ such that $\bar{P}_A(\bar{a}) = \bar{P}_A( \bar{a}')$.
And thus, $y\in a\triangle a'$, with $P_A(a)=P_A(a')$.
%
\end{proof}

For instance, there is  solution for the  $(4,2,1)$ case, 
because it is  equivalent to a solution to $(3,3,1)$, the classic Russian cards case\footnote{
This is the example of~\cite{Albert2005SafeCF},
``we get a 7-line good announcement for (4, 2, 1). It may further be observed that this is the complement of a 7-line good announcement for (3,3,1) as found above (for no apparent reason related to designs)''.
}.
However,  there is no solution
for the  $(2,4,1)$ case, as we show in the next theorem (and was observed in~\cite{Albert2005SafeCF}).
The reason is that in this case we get the graph $J(7,2)$, which has no safe proper coloring.
Thus, while we
 assume that $A$ makes the first announcement; to analyze the other case, one may 
exchange values of  $\Na$ and $\Nb$.
It may be more convenient
that $A$ makes the first announcement, or that $B$ makes it, in terms of both solvability and communication complexity.
  For the first case, a coloring has to be found for $J(\Nn,\Na)$,
and for the second case, one for $J(\Nn,\Nb)$.



\begin{theorem}
\label{th:mainBoundCases}
If $\Nc+\Nr\geq \min\set{\Na,\Nn-\Na}-1$, $\Nc\geq 1$, then
there is no safe proper coloring of $J^{\Nc+\Nr}(\Nn,\Na)$.
\end{theorem}

\begin{proof}
Recall that the diameter of $J(\Nn,\Na)$ is $\min\set{\Na,\Nn-\Na}$. 
 If  $\Nc+\Nr=d\geq  \min\set{\Na,\Nn-\Na}$ then $J^d(\Nn,\Na)$ is a complete graph, and each vertex must have
 a different color, so if $A$ announces $M$ then $C$ learns that her hand is the single hand in $\chi^{-1}(M)$.

Assume therefore that 
$d= \min\set{\Na,\Nn-\Na}-1$, $\Na\geq 2$.
We have that  $\set{a,a'}\in E(J^d(\Nn,\Na))$ iff $|a\cap a'|\geq 1$.
If $\chi$ is a safe proper coloring,  consider a vertex $a\in\chi^{-1}(M)$ that includes some card $x$, for some color $M$.
The safety requirement 
  implies that there must be another vertex $a'\in\chi^{-1}(M)$ that also includes  card $x$,
  by Lemma~\ref{lm:doubleCsafety}, since $\Nc\geq 1$.
  A contradiction to the claim that $\chi$ is a proper coloring, because then  $|a\cap a'|\geq 1$.
\end{proof}

The requirement that $\Nc\geq 1$ is needed.
Suppose $\Nc=0$. For the case $\Nn=4,\Na=2,\Nb=1,\Nr=1$,
the following $P_A$ a safe proper coloring of $J(4,2)$, with three messages $0,1,2$.
\phantom{.}\\
$P_A^{-1}[0] = \{01, 23\}$\\
$P_A^{-1}[1] = \{02, 13\}$\\
$P_A^{-1}[2] = \{03, 12\}$


Recall that a protocol can be informative and safe only
if $\Nb>\Nc$ (Remark~\ref{rm:subgraphsM}).  Thus, combining this fact with 
Theorem~\ref{th:mainBoundCases}, we get the following.

\begin{corollary}
\label{cor:mainBoundInfoSafe}
There is no informative and safe protocol if  $\Nc\geq \Nb$ or if $\Nc+\Nr \geq \min\set{\Na,\Nn-\Na}-1$, $\Nc\geq 1$.
\end{corollary}
 
%
%


There are several particular cases of interest, some previously observed\footnote{Using two different
proof techniques, it was shown that if $\Na\leq \Nc+1$, there is no informative and safe solution ($\Nr=0$), in~\cite[Corollary 2]{Albert2005SafeCF}) and~\cite[Theorem 6]{swansStinson14}.}.

\section{Six messages solutions to the  Russian cards problem}
\label{sec:classRusCrds}

We study the classic Russian cards problem~\cite{DitmarschRC03}, with signature  $(3,3,1)$, and also we
consider  the \emph{weak} variant where $C$ gets no cards at all, $(3,3,0)$, both with  $\Nn=7$.
By Theorem~\ref{th:basicJohns} in both the classic or the weak variant,
we need to consider colorings of the Johnson graph $J(7,3)$.  
These provide  concrete examples of the previous ideas.

The classic Russian cards problem with signature $(3,3,1)$ and $\Nn=7$ 
has been thoroughly studied,
an exhaustive analysis can be found in~\cite{DitmarschRC03}\footnote{Notice that for a given deal, there are 102 ``direct exchanges'' for the Russian cards problem~\cite[Corollary 41]{DitmarschRC03}.
The direct exchanges are characterized in this paper, and the characterization can
be naturally rephrased in our framework.
A direct exchange corresponds in our notation to  a color class, $\chi^{-1}(M)$, the set of hands of $A$ 
on which the protocol sends message $M$. 
}. 
It is well-known that there is a uniform solution  with seven messages (Theorem~\ref{th:safetyRussCards}),
each announcement of the same size. Only one solution is known that we are aware of with six messages~\cite{swansStinson14}, non-uniform. 
We show there is no uniform safe solution with  $6$ announcements (where all except one announcement
are of the same size).

\subsection{Upper bounds: information transmission with 6 messages}
\label{sec:rusCards:upper}

The chromatic number of Johnson graphs has been well studied e.g.~\cite{EBitan96},
but in general, determining the chromatic number of a Johnson graph is an open 
problem~\cite[Chapter 16]{godsil_meagher_2015}.
 It is   known that $\chi(J(7,3))=6$, and hence there is an informative protocol with 6 messages, and no less\footnote{
 The same lower bound is~\cite[Theorem 4]{swansStinson14},  proved by reduction to a combinatorial design theorem.}, 
  by Theorem~\ref{def:prColoring}.  
  We also present an 
 explicit solution below,  which is informative, but not  safe for the weak version,
 and then a solution  that is informative and safe for the weak version, but not for the classic version. 
 At the end we prove there is no uniform solution with $6$ messages for the classic version.

 \begin{theorem}
 \label{th:informaRussCards}
 There is an informative (non-safe) protocol for the Russian cards problem sending $6$ different messages,
 and this is optimal.
 \end{theorem}

 While $\cG_B=J(7,3)$  is the same for the Russian cards problem and for its weak version, 
 for the protocol to be safe one needs to consider the possible inputs of $C$.
In the Russian cards problem, the $C$-vertices are
$\powersetS{\Nc}{\cards}$, $\Nc=1$,
while in the weak version,  there is a single $C$-vertex,
 $(C,\emptyset)$.
 We will show,  that the graph $J(7,3)$
 has a safe proper coloring in either version.
Thus, there are  solutions to the weak problems $(3,3,0)$ and $(4,2,0)$,
and similarly, for the  classic problems,  $(3,3,1)$  and
 the $(4,2,1)$,
  by the duality Theorem~\ref{th:dual}.\footnote{
 This explains the issue raised in~\cite{Albert2005SafeCF} (Example p.12):
``   Applying this construction, we get a 7-line good announcement for $(4,2,1)$. 
It may further be observed that this is the complement of a 7-line good announcement for $(3,3,1)$ 
as found above (for no apparent reason related to designs).''
}

In the case of the {weak Russian cards problem},
 there is a solution using $6$ different messages. Namely,
for $J(7,3)$, the informative and safe chromatic number w.r.t. $V(\cG_C)=\emptyset$ 
is equal to the chromatic number $\chi^p=\chi=6$.
Complementary protocols (Theorem~\ref{th:dual}) solve the cases $(3,3,1)$
and $(4,2,1)$, 
using $6$ different messages, and this is optimal in terms of the number of messages (colors).

\begin{theorem}
\label{th:russCardPrWeak}
There is a solution for the weak Russian cards problem with  6 messages, and this is optimal.
\end{theorem}

First, it is not hard to design  proper $6$-colorings of $J(7,3)$, the following is an example:
\phantom{.}\\
$\chi^{-1}[0] = \{012, 034, 056, 135, 146, 236, 245\}$\\
$\chi^{-1}[1] = \{016, 024, 035, 123, 145,  256,  346\}$\\
$\chi^{-1}[2] = \{015, 023, 046, 124, , 136, 345\}$\\
$\chi^{-1}[3] = \{013, 026,  045,   125, 234,  356\}$\\
$\chi^{-1}[4] = \{014, 025, 036, 126, 456\}$\\
$\chi^{-1}[5] = \{134, 156, 235, 246\}$\\
For each $i\in\set{0,1,2,3,4}$ the coloring is safe, because there are vertices $a,a'\in \chi^{-1}[i]$, 
with $x\in a$ and $x\not\in a'$, for each $x\in\set{0,\ldots,6}$. 
However,  this coloring  is not {safe} w.r.t. $\Nc=0$, i.e. for the weak version, because
there is no $a\in  \chi^{-1}[5]$, with $0\in a$. That is, if $C$ listens to announcement $5$ she learns
that $A$ does not have card $0$.

 To obtain a safe coloring w.r.t. $\Nc=0$, we may  fix $\chi^{-1}[5]$ by adding a vertex $a$
 that contains card $0$, but taking care that $a$ is not adjacent to any vertex already
 there.
 We construct a safe variant $\chi_1$ of $\chi$, 
by removing the vertex $a_1 = \{012\}$ from  $\chi^{-1}[0]$ and adding it to $\chi_1^{-1}[5]$.
We get the following    $6$ coloring, safe w.r.t. $\Nc=0$, because by removing
$a_1$ from  $\chi^{-1}[0]$ we have not disrupted the safety of the announcement $0$.

$\,$\\
$\chi_1^{-1}[0] = \{034, 056, 135, 146, 236, 245\}$\\
$\chi_1^{-1}[1] = \{016, 024, 035, 123, 145,  256,  346\}$\\
$\chi_1^{-1}[2] = \{015,  023, 046, 124, 136,  345\}$\\
$\chi_1^{-1}[3] = \{013, 026, 045, 125,  234,  356\}$\\
$\chi_1^{-1}[4] = \{ 014, 025, 036, 126, 456\}$\\
$\chi_1^{-1}[5] = \{012, 134, 156, 235, 246\}$\\

However, the previous coloring, is not  safe w.r.t. the Russian cards problem, where the $C$-vertices are
$\powersetS{\Nc}{\cards}$, $\Nc=1$.
For example,  if $C$ 
has card $1$ and $A$ announces color $5$, then $C$ knows that $A$ has hands $235$ or $246$,
and can deduce that $A$ has card $2$ and also that she does not have card $0$.
Also,  if $C$  has card $0$ and $A$ announces color $4$ then she knows that $A$ has hands $126$ or $456$,
and can deduce that $A$ does not have $3$ and she has $6$. 

There is an informative and safe coloring of the Russian cards problem with six messages~\cite{swansStinson14},
$\,$\\
$\chi_2^{-1}[0] = \{013, 026, 045, 124, 156, 235, 346 \}$\\
$\chi_2^{-1}[1] = \{015, 023, 046,  126, 134,  245, 356 \}$\\
$\chi_2^{-1}[2] = \{016, 024, 035, 123,  145, 256\}$\\
$\chi_2^{-1}[3] = \{012,036,   135, 234, 456  \}$\\
$\chi_2^{-1}[4] = \{ 056,  034, 125, 146,  236\}$\\
$\chi_2^{-1}[5] = \{ 014, 025,  136, 246, 345\}$\\

\subsection{Impossibility of uniform solutions}
\label{sec:rusCards:lower}

Here we discuss a new technique to study the structure of six message solutions to the Russian cards problem.
The previous solution with six messages, partitions the $35$ hands of $A$ into color classes of size $5,5,5,6,7,7$.
 We prove now that this is optimal for a six message \emph{uniform} solution, namely, two color classes
 must be of size 7. Thus, there is no solution with six messages
 with classes of sizes $5,6,6,6,6,6$ nor $5,5,6,6,6,7$. 
 
\begin{theorem}
\label{th:lowerB6}
There is no uniform solution to the Russian cards problem with six messages.
\end{theorem}

\begin{proof}
Assume for contradiction that there is such a protocol $P_A$, which partitions all the 35 possible hands of $A$ into 6 color classes. 
One  class must have 5 hands, by a counting argument, not all can have at least 6,
and it is not hard to check that a color class cannot have only 4 hands.
Also, a color class cannot have more than 7 hands (as observed in~\cite{DitmarschRC03}). 
Thus, the most uniform solution induces a partition of sizes  $5,6,6,6,6,6$.
And the less-uniform solutions are either  of sizes  $5,5,6,6,6,7$,
or $5,5,5,6,7,7$. 

A partition with 5 hands must have a single card, say $0$, that appears in 3 hands. All other cards appear twice.
There are  15  hands containing 0. Consider all remaining 12 hands containing 0 in the other color classes,
say 2 through 6.

In the remaining 5 classes 
there must be 3 with  two hands containing 0,
and 2 classes with three hands containing 0. 
Recall that each card must appear at least twice in a color class, Lemma~\ref{lm:doubleCsafety}.
Also, no color class can have 4 hands containing 0, because then two hands would
have an intersection of 2 cards (and share an edge of $J(7,3)$, violating the properness of the coloring).

\begin{figure}[h]
\centering
\includegraphics[scale=0.18]{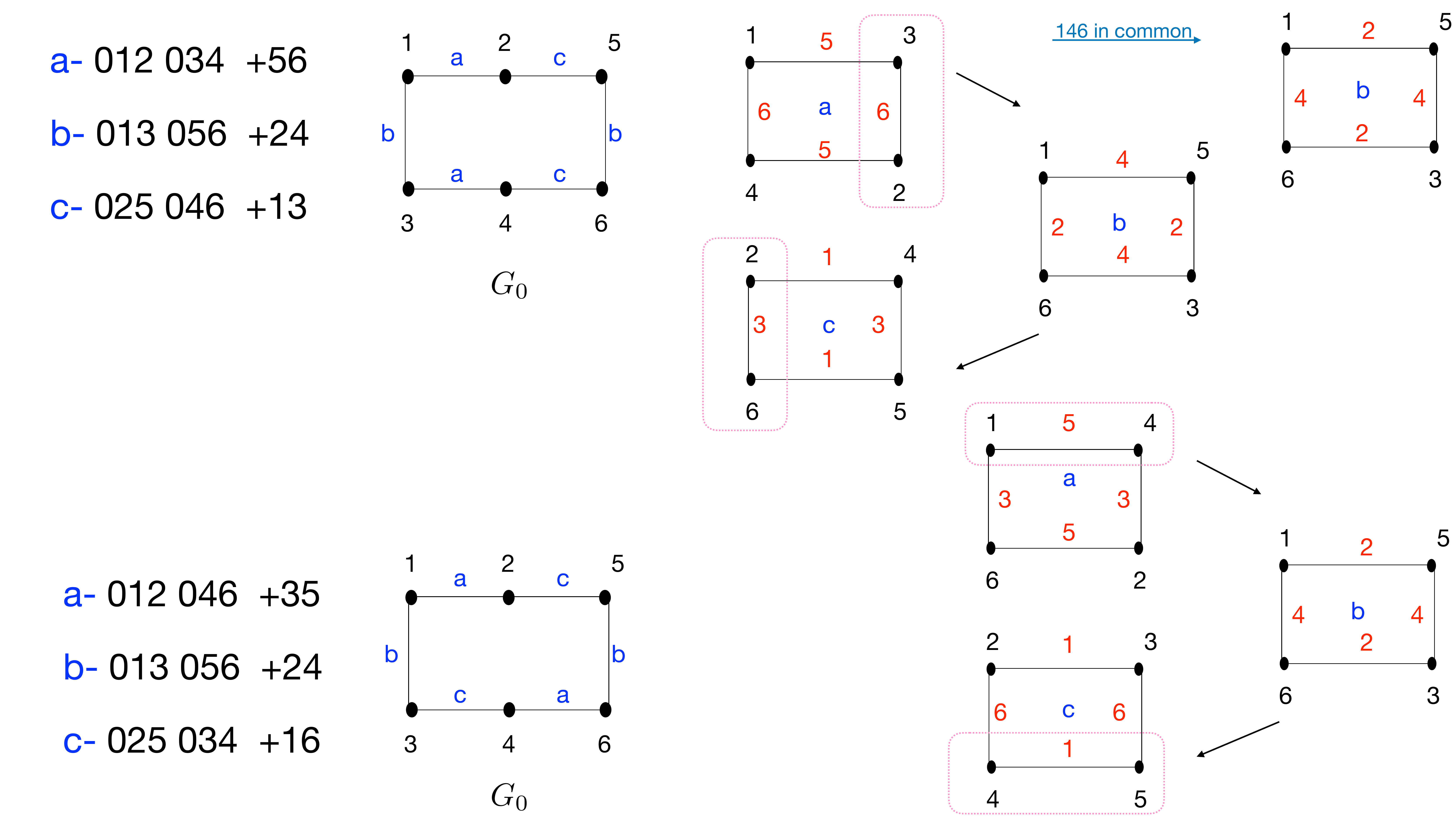}
\caption{First configuration on top $12,34$; $13,56$; $25,46$. Second configuration on bottom 
$12,46$; $13,56$; $25,34$. On the right part of the trees of possible ways of completing them.}
\label{fig-6coloringRC}
\end{figure}

Consider  three color classes of size  6, denoted $a,b,c$, 
each one has exactly two hands containing 0. 
The case where one of these classes is of size 7, and hence it has three hands containing 0, is similar; it will be 
discussed at the end.

The 3 color classes $a,b,c$ with two  hands containing $0$ define a graph 
$G_0$ on the vertices $\cards\setminus 0=\set{1,2,3,4,5,6}$, each
vertex representing a card. An edge of this graph is colored with an element from $\set{a,b,c}$,
meaning that if an edge $x,y$ is colored  $i$, then the hand $0xy$ is in class $i\in\set{a,b,c}$.

 Since two hands in a  class cannot have an intersection of more than one card, it follows that the 
 edges of the same color are independent in $G_0$. 

Now, assume for contradiction that a vertex, say $1$, has degree 3.
The three edges $\set{1,v_1},\set{1,v_2},\set{1,v_3}$ are colored with different elements from $\set{a,b,c}$.
As we shall see, this implies that $1$ appears  in three hands of each class, $a,b,c$.
Therefore, it appears in two hands, of each of the remaining classes, $d,e,f$.
We can thus consider the graph $G_1$ on the vertices $\cards\setminus 1$, with edges colored
with  elements from $\set{d,e,f}$, meaning that if an edge $x,y$ is colored  $i$, then the hand $1xy$ is in class $i$.
The vertex $0$ of $G_1$ must then have degree $3$, because as we shall see, this is needed for $0$
to appear three times in each class $d,e,f$. But this implies that $0$ is incident to one of $v_1,v_2,v_3$, say $v_i$,
since the graph has only $6$ vertices. Namely, $\set{0,v_i}$ is and edge of $G_1$, and
$\set{1,v_i}$ is and edge of $G_0$, so the hand $01v_i$ appears twice, in a class of $\set{a,b,c}$
and a class of $\set{d,e,f}$, a contradiction to the assumption that a vertex has degree three in $G_0$.

Thus, the edges of $G_0$ either they form a cycle or two triangles.
There are two types of cyclic configurations for the three classes $a,b,c$ with two hands containing 0:
either for each $i \in\set{a,b,c}$, the edges colored $i$ are opposite in the cycle or not. For instance,
12,34; 13,56; 25,46  (all plus 0)
or else 12,46;13,56;25,34  (all plus 0). See Figure~\ref{fig-6coloringRC} for these two cyclic configurations, and
 Figure~\ref{fig-6coloringRCtriangles} for the triangles case. These figures illustrate
 the case where 0 appears in exactly two hands, and the color classes are of size 6.
 
\begin{figure}[h]
\centering
\includegraphics[scale=0.18]{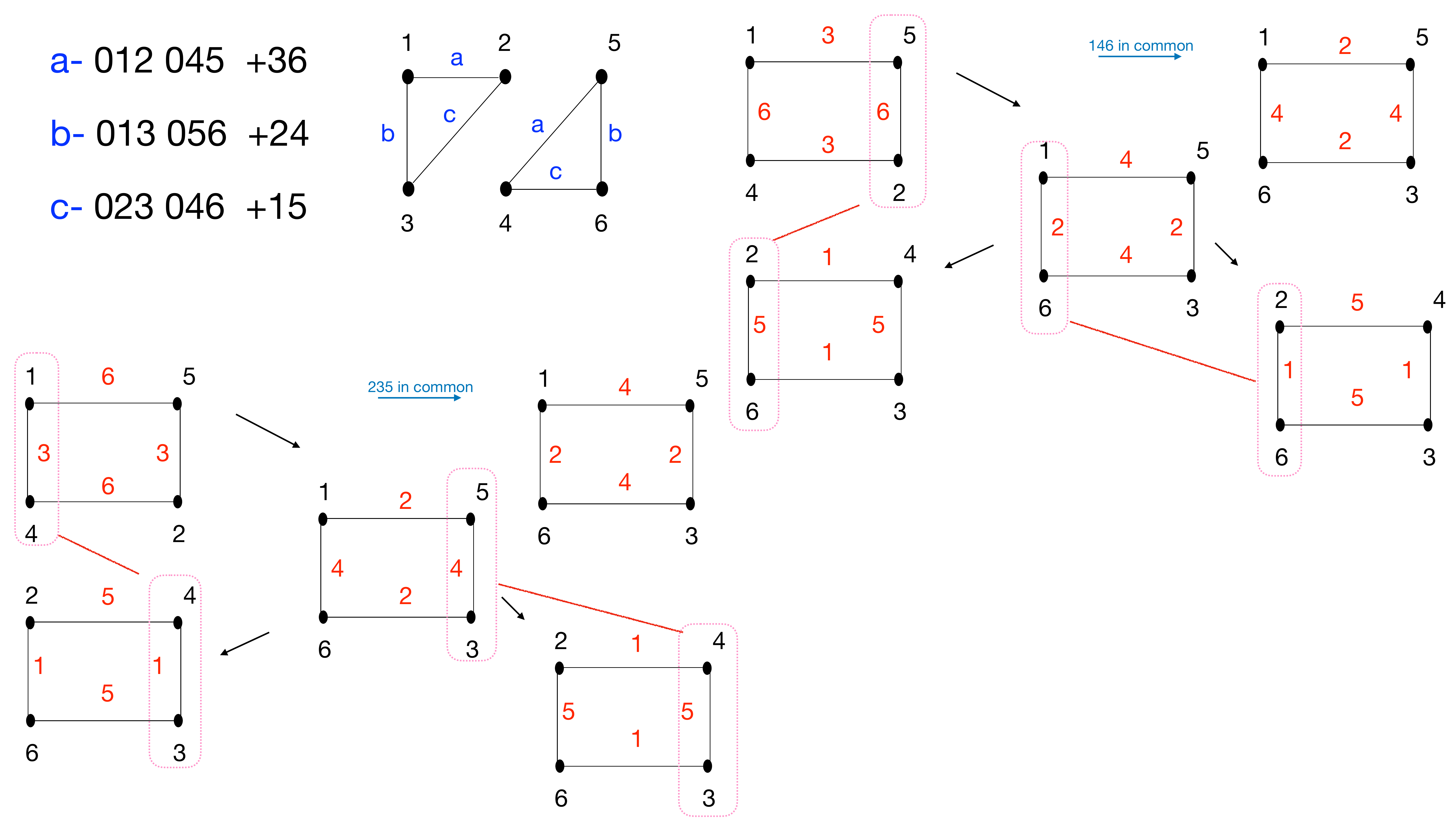}
\caption{The $a,b,c$ classes define two triangles. 
On the right part are the trees of possible ways of completing first $a$, then $b$ and then $c$, to have each 6 hands.
Each hand is represented by an edge.}
\label{fig-6coloringRCtriangles}
\end{figure}

We need to complete each set of two hands to form a color class of 6 hands, by adding  4 more hands.
These 4 more hands do not contain 0. The  process to do it, is 
 represented by three graphs, $G_a,G_b,G_c$. Now the vertices of the graph $G_i$, $i\in\set{a,b,c}$ are the four cards spanned by the two
independent edges of the class $G_i$.
There are four edges on these four vertices forming a cycle in each $G_i$; each edge corresponds
to a combination that \emph{does not}
appear in one of the two independent edges of $G_i$ 
(because two cards that already appeared in a hand, cannot occur in
another hand). The goal is to color these four edges, with the two remaining colors (0
is no longer available, because it already  appears in two hands).

Notice that a loop on a vertex $x$ could in principle be used, coloring it with the  two remaining colors,
giving the hand $xyz$, if the two remaining colors are $yz$. However, at most one such loop can be
used (using two such loops, would give hands with intersection $yz$, with violets the requirement that the color is proper).
And using a loop prevents using the two adjacent edges, leaving only the other two, non-adjacent edges to be used,
ie, coloring only 3 edges.
It follows that no such loop can be used, because we need to color 4 edges, to obtain together with the 2 hands
containing $0$, the
total number of hands which is 6 in the color class. 

Consider all 4 combinations of taking one card from each pair (of 2 values different from 0).
Then  add each of the two remaining cards
to complementary pairs, as illustrated in the figures.
For example, in Figure~\ref{fig-6coloringRC}, for the pairs (a) 12,34 one most add values 56. And there are only two options of getting independent edges.
Add 5 to 13 and to 24; add 6 to 14 and to 23, as in the figure. Or else 
add 6 to 13 and to 24; add 5 to 14 and to 23.

Once  5 is added to 13 and to 24, and  6 to 14 and to 23, the next move is determined, to complete class (b).
In the figure a blue arrow shows that 146 would be in common to the next class,
if we added 2 to 15 and 36; and 4 to 16 and 35. Thus, the only option is the complementary choice.
But then, either way, it is not possible to add 1 and 3 to class (c).  In the figure one choice
is shown, where 236 is repeated in classes (a) and (c).
The reader can verify that in either of the two types of configurations, this process cannot be completed.
The full tree  for  for the first configuration is in Figure~\ref{fig-6coloringRC-tree}.

To complete the proof, we describe how to deal with a class of size 7, where $0$ occurs in 3 hands.
Actually, exactly the same argument is used, considering two hands that contain $0$.
This is illustrated in Figure~\ref{fig-6coloringRCtriangles-partSize7}, where the hand $015$ of class $c$
is underlined, to indicate that it does not play a role on the right side of the figure (in fact, this
prevents it from using $15$ to label a loop). 
Namely, in the figure,  the two hands of class (c) are selected, $023,046$, to complete them
with 4 hands not containing 0, into a color class of size 7 (together with $015$).
Thus, we have the vertices 1245 on the graph for color class (c) on the right, and the possible combinations
represented by four edges forming a cycle. Each edge must be colored with $1$ or $5$, forming two
independent edges colored  $1$ and $5$.
The tree of possibilities is therefore the same as before. 
\begin{figure}[h]
\centering
\includegraphics[scale=0.5]{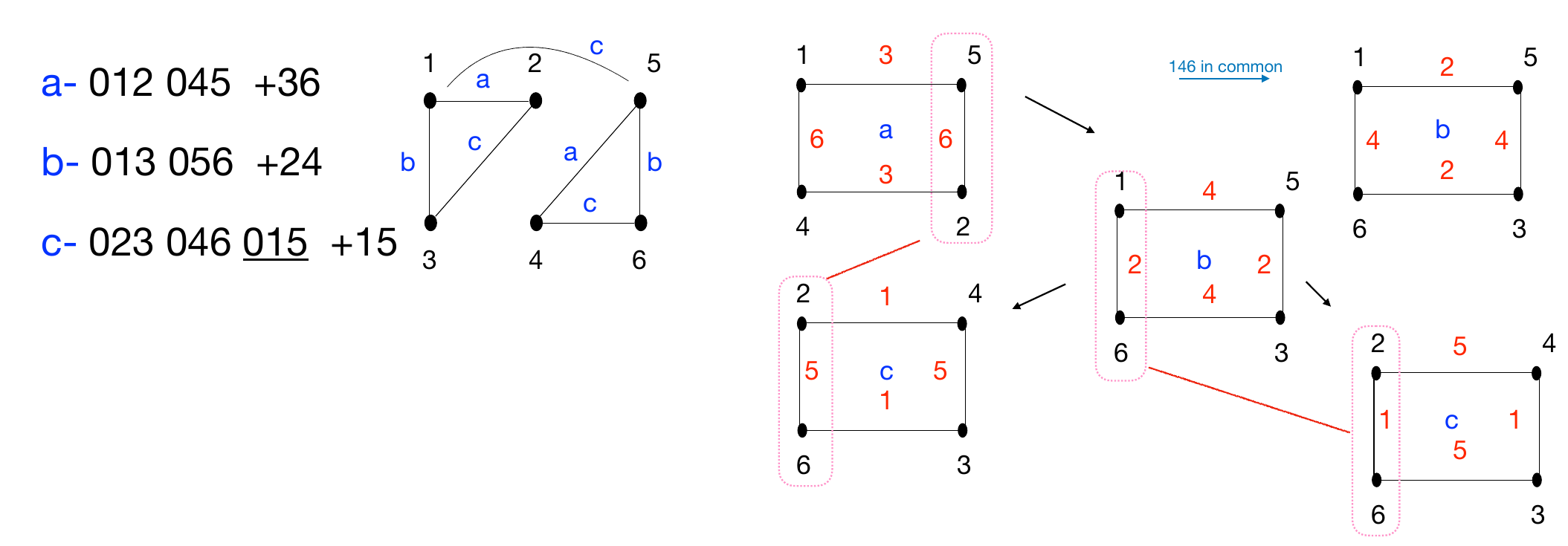}
\caption{Color class $c$ is of size 7, where 0 is in three hands, $023,046,015$. The first two
hands $023,046$ play the same role, as in the other figures. They have to be completed
with four more hands, to make a total of 7 hands.}
\label{fig-6coloringRCtriangles-partSize7}
\end{figure}

\end{proof}

\section{Minimal information transmission}
\label{sec:sBitTransm}

We study  first the  protocol, $\chi_2$,  that sends the sum of the cards modulo $2$.
The techniques are simple, but
serve as an introduction to the more complicated case of $\chi_{modn}$,
the  mod $\Nn$ version of this protocol, studied in Section~\ref{sec:modularAlgo}.

We show in Section~\ref{sec:minInfo2} that $\chi_{2}$ is minimally informative  
only if $\Nb<\floor{\Nn/2}$. Thus, $\chi_{2}$ is not minimally informative for
the classic Russian cards case $(3,3,1)$.

In Section~\ref{sec:minInfRussCards} we describe  how to transform an informative
protocol into a minimally informative protocol.   Applying the reduction to $\chi_{modn}$,
when $\Nc+\Nr=1$, as $\Na$ grows from $3$
up to roughly  $\Nn/2$, the number of different messages 
goes down from $\Nn/3$ to $2$.

This reduction shows that there is a safe minimally informative
protocol for the Russian cards case $(3,3,1)$ using $3$ messages.
Finally, we present a solution to the Russian cards case  using only $2$ messages, in Section~\ref{sec:minInf2RC}.
Given that there is no uniform safe informative protocol using $6$ messages  (Theorem~\ref{th:lowerB6}),
indeed this  $2$-message protocol  splits color classes of an informative protocol.

\subsection{Minimal information with 2 messages}
\label{sec:minInfo2}

For signature $(\Na,\Nb,\Nc)$, with $\Nn=\Na+\Nb+\Nc+\Nr$, consider
a protocol  
 $\chi_2:  \powersetS{\Na}{\cards} \rightarrow \set{0,1}$, defined by
$$
\chi_{2}(a)=\sum x\in a \pmod{2}.
$$


\subsubsection{The protocol $\chi_2$ is minimally informative}
\label{sec:minInfTransf}

Recall Lemma~\ref{lem:neighbrs}.
For each input vertex $(B,b)$ denoting that $B$ gets hand $b$,
there are  $m=  \binom{\Nn-\Nb}{\Na} $ possible hands $a_i$ for $A$, corresponding to vertices $(A,a_i)$.
In $J^{\Nc+\Nr}(\Nn,\Na)$ these vertices form a maximal click
 $K_{p} (\bar{b})$ of $\cG_B$,
 $p=  \binom{\Nn-\Nb}{\Na} $, consisting of all hands $a\subset \bar{b}$,  $|a|=\Na$.
If $\Nb\geq \floor{\Nn/2}$ then for $b$ of size $\Nb$, $\bar{b}$ may consist of cards of the same parity,
and thus all $a\subset \bar{b}$,  $|a|=\Na$ have the same parity, and $\chi_{2}$ is not minimally informative.


\begin{lemma}
\label{lem:mainMinCol}
Assume that $\Nc+\Nr\geq 1$, $\Na\geq 1$, 
 $\Nb<\floor{\Nn/2}$.
Then $\chi_{2}$ is a minimally informative  protocol.
\end{lemma}

\begin{proof}
We use two facts.
Since $\Nb<\floor{\Nn/2}$ then $|\bar{b}|> \Nn-\floor{\Nn/2}$, for any $b$ with $|b|=\Nb$,
and  $\bar{b}$  must contain both even and odd cards.
Since $\Nc+\Nr\geq 1$ (as required by the minimally informative definition), then 
 $\Na< |\bar{b}|$.

To show that $\chi_{2}$ is minimally informative, consider any click $K_{p} (\bar{b})$.
Let $a\subset \bar{b}$,  $|a|=\Na$, be a vertex of  $K_{p} (\bar{b})$  with the largest  number of odd cards. 
Since   there are both even and odd cards in $\bar{b}$, $a$ contains at  least one odd card, $y$.
Since $a$ contains the largest possible number of odd cards, it contains the minimum number of even cards.
Thus, there is at least one even card $y'\in \bar{b}\setminus a$, given that $|a|< |\bar{b}|$. 
 Let $a'=(a\setminus y)\cup y'$. Thus, $a'$ is also a vertex of $K_{p} (\bar{b})$, and
 $\chi_{2}(a)\neq \chi_{2}(a')$.
\end{proof}

%
%
%
%
%
%
%
%
%
%

\subsubsection{The protocol $\chi_2$ is safe}
\label{sec:safeChi2}

Lemma~\ref{lem:mainMinCol} implies that $\chi_{2}$ is  minimally informative when
 $\Nn=7, \Na=3,\Nb=2,\Nc=2, \Nr=0$, namely, for  $J^2(7,3)$.
But it is not safe, because if $C$ has hand $\set{1,3}$ and the announcement is $0$ she knows
that $A$ does not have card $5$. Or if the announcement is $1$, she knows that $A$ has card $5$.
More generally, the number of odd cards in $\cards$ is $\floor{\Nn/2}$. If $\Nc = \floor{\Nn/2}-1$
then when $C$ holds $\Nc$ odd cards she can deduce from the announcement whether $A$
holds the remaining odd card.
Thus, assume that $\Nc \leq  \floor{\Nn/2}-2$, and  additionally,   $\Na\geq 2$ (Remark~\ref{lem:simpleSafe1}).

In Section~\ref{sec:modSafe} we  discuss the modulo $\Nn$ case and the relation of proving safety with 
additive number theory.  The  proof here for the modulo 2  case provides a simple illustration of the ideas.
 
The safety characterization of Theorem~\ref{th:safetyMain}(\ref{th:safetyMain2}) instantiated for
 protocol $\chi_{2}$, says that (cf.~\cite[Proposition 6]{Cordon-FrancoDF12})
$\chi_{2}$ is {safe} (with respect to $\Nc$) if and only if
for each  $\Nc$-set $c$,  $y\in\bar{c}$, and $M\in \set{0,1}$, 
there exists two $\Na$-sets $a,a' \in \bar{c}$, $\chi_{2}(a)=\chi_{2}(a')=M$ such that $y\in a\triangle a'$.


\begin{lemma}
\label{lem:mainMinColSafe}
Assume that ~$\Na,\Nb\geq 2$ and $\Nc\leq \floor{\Nn/2}-2$.
Then $\chi_{2}$ is a  safe protocol.
\end{lemma}

\begin{proof}
Consider any  $\Nc$-set $c$, and  $y\in\bar{c}$.
Let $z,z'\in \cards\setminus (c\cup y)$ be cards of different parity, which they exist because  $\Nc\leq \floor{\Nn/2}-2$.
First, let $a_1$ be any $\Na$-set in $\bar{c}$ that does not include $y$, and which  includes $z$ but not $z'$,
which exists because $\Nb\geq 2$.
Let $a_2=(a_1\setminus z)\cup z'$.
Thus, $\chi_2(a_1)\neq\chi_2(a_2)$.
Similarly, let $a'_1$ be any $\Na$-set in $\bar{c}$ which includes $y$, and which  includes $z$ but not $z'$.
And let $a'_2=(a'_1\setminus z)\cup z'$.
Thus, $\chi_2(a'_1)\neq\chi_2(a'_2)$.

We are done, because  for each $M\in\set{0,1}$, there is one $i\in\set{1,2}$ such that $\chi_2(a_i)= M$ and does not include $y$,
and there is one $i\in\set{1,2}$ such that $\chi_2(a'_i)= M$ and does  include $y$.
\end{proof}

Combining Lemma~\ref{lem:mainMinCol} and Lemma~\ref{lem:mainMinColSafe} we get 
the following theorem.

\begin{theorem}
\label{th:mainMinInfo}
Let $\Nn=\Na+\Nb+\Nc+\Nr$.
If
 $\Na,\Nb\geq 2$, $\Nc\leq \floor{\Nn/2}-2$,
 $\Nc+\Nr\geq 1$, and
 $\Nb<\floor{\Nn/2}$,
 then $\chi_2$ is minimally informative and safe. 
\end{theorem}

Thus, for example, when $\Nn=7$, $\Na=3$, $\Nb=2$, $\Nc=1,\Nr=1$, namely, $J^2(7,3)$,
then $\chi_2$ is both minimally informative and safe.
Similarly for $\Nn=7$, $\Na=4$, $\Nb=2$, $\Nc=1,\Nr=0$, namely, $J(7,4)$.
Which is interesting, because it shows that the  duality Theorem~\ref{th:dual} does not
hold for minimally informative protocols; 
notice that $J(7,4)\cong J(7,3)$, but $\bar{\chi}_2$ is not minimally informative
for $J(7,3)$ (neither is $\chi_2$).
More generally, for the  Russian cards case, we get the following.

\begin{corollary}
\label{cor:minInfo1}
Assume $\Nc+\Nr=1$. Then, $\chi_2$ is minimally informative and safe,
whenever $\Na>\ceil{\Nn/2}-1$ and $\Nb<\floor{\Nn/2}$.
\end{corollary}

%

\subsection{Reducing  informative to minimally informative protocols}
\label{sec:minInfRussCards}


As observed in Section~\ref{sec:minInfTransf}, the protocol $\chi_2$ is not
minimally informative when  $\Na\leq \ceil{\Nn/2}-1$ or $\Nb\geq \floor{\Nn/2}$, and thus, in particular, 
 for the Russian cards problem $(3,3,1)$, $\Nr=0$.
We present here a protocol for this case, 
based on the $\chi_{modn}$  protocol studied in Section~\ref{sec:modularAlgo}.
Notice that the protocol $\chi_{modn}$ is safe and informative when $\Nc+\Nr=1$.

The protocol  uses the idea that, merging two color
classes of a protocol $P_A$, 
$P_A^{-1}[M]\cup P_A^{-1}[M']$, leads to a new protocol that 
preserves safety (but possibly not informative properties). 
Actually, the idea works for any safe and informative protocol 
 $P_A: \powersetS{\Na}{\cards} \rightarrow \MA$.
 If $|\MA|=m$, let us denote  $\MA= \mathbb{Z}_m$.
 
If  $ P_A: \powersetS{\Na}{\cards} \rightarrow \mathbb{Z}_{m}$ is a safe proper coloring of $J^{\Nc+\Nr}(\Nn,\Na)$, $\Nc+\Nr\geq 1$,
define the protocol,
 $
 P_A^{[p]}: \powersetS{\Na}{\cards} \rightarrow \mathbb{Z}_{\ceil{m/(p-1)}},
$
where 
\begin{equation*}
P_A^{[p]}(a) = P_A(a)  \pmod{\ceil{m/(p-1)}},
 \end{equation*}
$p= {{\Na+\Nc+\Nr}\choose{\Na}} = {{\Nn-\Nb}\choose{\Na}}$.

\begin{theorem}[Information reduction]
\label{th:minInfoReduc}
If $P_A$ is a safe and informative protocol 
then $P_A^{[p]}$ is a safe and minimally informative protocol.
Thus, if $m$ is the different number of messages used by $P_A$,
then $\ceil{m/(p-1)}$ is the number of messages used by  $P_A^{[p]}$.
\end{theorem}

\begin{proof}
%
Notice that each color class of $P_A^{[p]}$ consists
of a union of at most $p-1$
 color classes of $P_A$.
Since the protocol $P_A^{[p]}$  is defined in terms of merging color classes
of $P_A$, if $P_A$ is safe
then $P_A^{[p]}$   is safe (follows directly from Theorem~\ref{th:safetyMain}).

Furthermore, $P_A^{[p]}$  is minimally informative, because 
the number of vertices in a click $K_{p} (\bar{b})$ is 
$p= {{\Na+\Nc+\Nr}\choose{\Na}} = {{\Nn-\Nb}\choose{\Na}}$.
Since $P_A$ is informative, any two vertices of $K_{p} (\bar{b})$
belong to different color classes of $P_A$. 
Since each color class of  $P_A^{[p]}$  consists of at most  $p-1$
color classes of $P_A$, then not all such
 $p$ vertices can be assigned the same color by  $P_A^{[p]}$.
%
\end{proof}

In the case of $\Nc+\Nr=1$, the protocol $\chi_{modn}$
studied in Section~\ref{sec:modularAlgo} is a safe and informative protocol
(Theorem~\ref{th:safetyRussCards}), using $\Nn$ different messages.
In this case, $p=\Na+1$. Thus we have the following.

\begin{corollary}
\label{cor:minInfoRusCard}
The protocol   $ \chi_{modn}^{[\Na+1]}$ 
is minimally informative and safe for $\Na,\Nb\geq 3$, $\Nc+\Nr=1$,
using $ \ceil{\Nn/\Na}$ different messages.
In particular, it uses $3$ messages for the case $(3,3,1)$, $\Nr=0$.
\end{corollary}

Notice that not every minimally informative safe protocol can be obtained
by reduction from an informative protocol.
Theorem~\ref{th:mainMinInfo} states that  $\chi_2$ is minimally informative and safe in
some cases where
\begin{equation}
\label{eq:noSafeInfo}
\Nc\geq \Nb \,\text{ or }\,  \Nc+\Nr \geq \min\set{\Na,\Nn-\Na}-1.
\end{equation}
For instance, the case of signature $(6,6,8)$, $\Nr=0$, satisfies the
hypothesis of the theorem and hence $\chi_2$ is minimally informative and safe.
But recall that in such cases (\ref{eq:noSafeInfo}) there  is no informative and safe protocol (Corollary~\ref{cor:mainBoundInfoSafe}).

\subsection{A solution to the Russian Cards problem with two messages}
\label{sec:minInf2RC}

In this section we present a solution found by Zoe Leyva-Acosta and Eduardo Pascual-Aseff, using a computer program.
The following protocol $\chi$ is a minimally informative 2-coloring of $J(7,3)$.
\vspace{3mm}
\phantom{.}\\
$\chi^{-1}(0) = \{012, 013, 014, 015, 016, 023, 024, 025, 036, 046, 056, 126, 134, 135,$\\
\hspace*{4.5em} $234, 236, 245, 246, 345, 356, 456\}$\\
$\chi^{-1}(1) = \{026, 034, 035, 045, 123, 124, 125, 136, 145, 146, 156, 235, 256, 346\}$
\vspace{3mm}

In Table~\ref{tab:mininf_table} 
we show for each $3$-set $b$, how $\chi$ partitions
the 3-set vertexes in $K_p(\bar{b})$ into two color classes, so that the reader  can verify
that this is in fact a minimally informative coloring for $J(7,3)$. 
To verify that $\chi$ is also a safe coloring, in Table~\ref{tab:safety_table} we show how $\chi$ partitions $K_p(\bar{c})$ for each
card $c$ into two color classes. The reader can check that in all such partitions and for any 
card other than $c$, there is a hand which contains it and another that doesn't. 

\begin{table}[h!]
\centering
\begin{tabular}{| m{0.8cm} | m{2.5cm} | m{2.5cm} | m{0.8cm} | m{2.5cm} | m{2.5cm} |}
\hline

\vspace{2mm}\hspace{2mm}$b$ & \hspace{2mm}$\chi^{-1}(0) \cap K_p(\bar{b})$ & \hspace{2mm}$\chi^{-1}(1) \cap K_p(\bar{b})$ & \vspace{2mm}\hspace{2mm}$b$ & \hspace{2mm}$\chi^{-1}(0) \cap K_p(\bar{b})$ & \hspace{2mm}$\chi^{-1}(1) \cap K_p(\bar{b})$ \\ [0.5ex] 
\hline
\hspace{1.5mm}$012$ & \hspace{1.5mm}\{345, 356, 456\} & \hspace{1.5mm}\{346\} & \hspace{1.5mm}$126$ & \hspace{1.5mm}\{345\} & \hspace{1.5mm}\{034, 035, 045\} \\

\hspace{1.5mm}$013$ & \hspace{1.5mm}\{245, 246, 456\} & \hspace{1.5mm}\{256\} & \hspace{1.5mm}$134$ & \hspace{1.5mm}\{025, 056\} & \hspace{1.5mm}\{026, 256\} \\

\hspace{1.5mm}$014$ & \hspace{1.5mm}\{236, 356\} & \hspace{1.5mm}\{235, 256\} & \hspace{1.5mm}$135$ & \hspace{1.5mm}\{024, 046, 246\} & \hspace{1.5mm}\{026\} \\

\hspace{1.5mm}$015$ & \hspace{1.5mm}\{234, 236, 246\} & \hspace{1.5mm}\{346\} & \hspace{1.5mm}$136$ & \hspace{1.5mm}\{024, 025, 245\} & \hspace{1.5mm}\{045\} \\

\hspace{1.5mm}$016$ & \hspace{1.5mm}\{234, 245, 345\} & \hspace{1.5mm}\{235\} & \hspace{1.5mm}$145$ & \hspace{1.5mm}\{023, 036, 236\} & \hspace{1.5mm}\{026\} \\

\hspace{1.5mm}$023$ & \hspace{1.5mm}\{456\} & \hspace{1.5mm}\{145, 146, 156\} & \hspace{1.5mm}$146$ & \hspace{1.5mm}\{023, 025\} & \hspace{1.5mm}\{035, 235\} \\

\hspace{1.5mm}$024$ & \hspace{1.5mm}\{135, 356\} & \hspace{1.5mm}\{136, 156\} & \hspace{1.5mm}$156$ & \hspace{1.5mm}\{023, 024, 234\} & \hspace{1.5mm}\{034\} \\

\hspace{1.5mm}$025$ & \hspace{1.5mm}\{134\} & \hspace{1.5mm}\{136, 146, 346\} & \hspace{1.5mm}$234$ & \hspace{1.5mm}\{015, 016, 056\} & \hspace{1.5mm}\{156\} \\

\hspace{1.5mm}$026$ & \hspace{1.5mm}\{134, 135, 345\} & \hspace{1.5mm}\{145\} & \hspace{1.5mm}$235$ & \hspace{1.5mm}\{014, 016, 046\} & \hspace{1.5mm}\{146\} \\

\hspace{1.5mm}$034$ & \hspace{1.5mm}\{126\} & \hspace{1.5mm}\{125, 156, 256\} & \hspace{1.5mm}$236$ & \hspace{1.5mm}\{014, 015\} & \hspace{1.5mm}\{045, 145\} \\

\hspace{1.5mm}$035$ & \hspace{1.5mm}\{126, 246\} & \hspace{1.5mm}\{124, 146\} & \hspace{1.5mm}$245$ & \hspace{1.5mm}\{013, 016, 036\} & \hspace{1.5mm}\{136\} \\

\hspace{1.5mm}$036$ & \hspace{1.5mm}\{245\} & \hspace{1.5mm}\{124, 125, 145\} & \hspace{1.5mm}$246$ & \hspace{1.5mm}\{013, 015, 135\} & \hspace{1.5mm}\{035\} \\

\hspace{1.5mm}$045$ & \hspace{1.5mm}\{126, 236\} & \hspace{1.5mm}\{123, 136\} & \hspace{1.5mm}$256$ & \hspace{1.5mm}\{013, 014, 134\} & \hspace{1.5mm}\{034\} \\

\hspace{1.5mm}$046$ & \hspace{1.5mm}\{135\} & \hspace{1.5mm}\{123, 125, 235\} & \hspace{1.5mm}$345$ & \hspace{1.5mm}\{012, 016, 126\} & \hspace{1.5mm}\{026\} \\

\hspace{1.5mm}$056$ & \hspace{1.5mm}\{134, 234\} & \hspace{1.5mm}\{123, 124\} & \hspace{1.5mm}$346$ & \hspace{1.5mm}\{012, 015, 025\} & \hspace{1.5mm}\{125\} \\

\hspace{1.5mm}$123$ & \hspace{1.5mm}\{046, 056, 456\} & \hspace{1.5mm}\{045\} & \hspace{1.5mm}$356$ & \hspace{1.5mm}\{012, 014, 024\} & \hspace{1.5mm}\{124\} \\

\hspace{1.5mm}$124$ & \hspace{1.5mm}\{036, 056, 356\} & \hspace{1.5mm}\{035\} & \hspace{1.5mm}$456$ & \hspace{1.5mm}\{012, 013, 023\} & \hspace{1.5mm}\{123\} \\

\hspace{1.5mm}$125$ & \hspace{1.5mm}\{036, 046\} & \hspace{1.5mm}\{034, 346\} & \hspace{1.5mm}      & \hspace{1.5mm}                  & \hspace{1.5mm}        \\

\hline
\end{tabular}
\vspace{2mm}
\caption{Color partitions of $K_p(\bar{b})$ for each $b$, according to $\chi$}
\label{tab:mininf_table}
\end{table}

\begin{table}[h!]
\centering
\begin{tabular}{| m{0.5cm} | m{4cm} | m{4cm} |}
\hline
\vspace{2mm}
\hspace{1.5mm}$c$ & $\chi^{-1}(0) \cap K_p(\bar{c})$ & $\chi^{-1}(1) \cap K_p(\bar{c})$ \\ [0.5ex] 
\hline
\hspace{1.5mm}$0$ & \{126, 134, 135, 234, 236, 245, 246, 345, 356, 456\} & \{123, 124, 125, 136, 145, 146, 156, 235, 256, 346\} \\
\hline
\hspace{1.5mm}$1$ & \{023, 024, 025, 036, 046, 056, 234, 236, 245, 246, 345, 356, 456\} & \{026, 034, 035, 045, 235, 256, 346\} \\
\hline
\hspace{1.5mm}$2$ & \{013, 014, 015, 016, 036, 046, 056, 134, 135, 345, 356, 456\} & \{034, 035, 045, 136, 145, 146, 156, 346\} \\
\hline
\hspace{1.5mm}$3$ & \{012, 014, 015, 016, 024, 025, 046, 056, 126, 245, 246, 456\} & \{026, 045, 124, 125, 145, 146, 156, 256\} \\
\hline
\hspace{1.5mm}$4$ & \{012, 013, 015, 016, 023, 025, 036, 056, 126, 135, 236, 356\} & \{026, 035, 123, 125, 136, 156, 235, 256\} \\
\hline
\hspace{1.5mm}$5$ & \{012, 013, 014, 016, 023, 024, 036, 046, 126, 134, 234, 236, 246\} & \{026, 034, 123, 124, 136, 146, 346\} \\
\hline
\hspace{1.5mm}$6$ & \{012, 013, 014, 015, 023, 024, 025, 134, 135, 234, 245, 345\} & \{034, 035, 045, 123, 124, 125, 145, 235\} \\
\hline
\end{tabular}
\vspace{2mm}
\caption{Color partitions of $K_p(\bar{c})$ for each $c$, according to $\chi$}
\label{tab:safety_table}
\end{table}

\section{The modular protocol  $\chi_{modn}$ for $\Nc+\Nr=1$}
\label{sec:modularAlgo}


For signature $(\Na,\Nb,\Nc)$, with $\Nn=\Na+\Nb+\Nc+\Nr$, consider
the protocol 
 $\chi_{modn}: \powersetS{\Na}{\cards} \rightarrow \mathbb{Z}_\Nn$, defined by
$$
\chi_{modn}(a)=\sum x\in a \pmod{\Nn}.
$$
All operations in this section are modulo $\Nn$, working in $\mathbb{Z}_\Nn$, even when not explicitly stated.
We show that  $\chi_{modn}$ is informative and safe when $\Nc+\Nr=1$.  It is easy to see that
$\chi_{modn}$ is not informative when $\Nc+\Nr>1$, and more complicated techniques
are needed, discussed in Section~\ref{sec:general-c}.

\subsection{$\chi_{modn}$ is informative and coding theory}
\label{sec:modInfo}

 The result that $\chi_{modn}$ is informative when $\Nc+\Nr=1$ is known and easy~\cite{Cordon-FrancoDF12}.
But our perspective that this is equivalent to being a proper vertex coloring
of $J(\Nn,\Na)$ exposes the connection with coding theory.
It is actually the  argument 
(generalized in Section~\ref{sec:general-c} to $\Nc+\Nr>1$)  behind  a classic 
 coding theory proof that shows a lower bound  on $A(n,4,w)$,
 the maximum number of codewords in any binary code of length $n$,
 constant weight $w$, and Hamming distance 4~\cite[Theorem 1]{GrahamS1980}.

\begin{lemma}
\label{lem:properColBasic}
For $\Nc+\Nr=1$, $\chi_{modn}$ is a proper vertex coloring of $J(\Nn,\Na)$, for $1\leq \Na<\Nn$.
\end{lemma}
\begin{proof}
Let
 $a=\set{x_1,x_2,\ldots,x_\Na} $ and
$a'=\set{x'_1,x_2\ldots,x_\Na}$ be adjacent vertices of $J(\Nn,\Na)$, $x_1\neq x'_1$.
Thus, 
 $a\cap a' = \set{x_2,\ldots,x_m}$ and $a\triangle a'=\set{x_1,x'_1}$.
If $\Na\geq 2$, let $  k = \sum x\in  a\cap a' \pmod{\Nn}$, else remove $k$ from the following equation.
 Then
 $\sum x\in a \equiv x_1+k\pmod{\Nn}$,
and
$\sum x\in a' \equiv x'_1+k\pmod{\Nn}$. 
Thus,  $\chi_{modn}(a)\neq\chi_{modn}(a')$, since $x_1\neq x'_1$ and $0\leq x_1,x'_1\leq n-1$.
\end{proof}

Notice that taking the sum modulo a number smaller than $\Nn$ may not give
a proper coloring. For example, for $J(7,3)$, $a=\set{012}$, $a'=\set{126}$, 
$\sum x\in a \pmod{6}=\sum x\in a' \pmod{6}=3$. Yet, we know from
Theorem~\ref{th:informaRussCards} that there is a proper coloring of $J(7,3)$ with $6$ colors.

\subsection{Additive number theory for safety}
\label{sec:3prop}

We have already hinted in Section~\ref{sec:safeChi2} that
proving that  the modular protocol is safe  translates into a question about additive number theory.
For each $M\in\mathbb{Z}_\Nn$, 
we look for solutions to the following linear congruence in $\mathbb{Z}_\Nn$,  
\begin{equation}
\label{eq:mainDistinct}
x_1+x_2+\cdots x_\Na \equiv M\pmod\Nn
\end{equation}
with distinct $x_i\in\mathbb{Z}_\Nn$.  
Additionally, for any given  $\Nc$-subset $c$ of $\mathbb{Z}_\Nn$,
we want that no $x_i\in c$.
Such a solution is denoted $a$, since it corresponds to an $\Na$-set, a vertex
 $a \in K_p(\bar{c})$, and it is said to \emph{avoid} $c$. 
 For $y\in \mathbb{Z}_\Nn$ and a solution $a$ to the linear congruence,
we say that $y\in a$, if $y=x_i$ for some $x_i$ in the solution.
 Finally, we need to show that for any $y\in \mathbb{Z}_\Nn$,  $y\not\in c$
there are two $c$-avoiding solutions, $a,a'$, such that $y\in a$ and $y\not\in a$.


The safety proofs are based on simple  properties about solutions to equation~(\ref{eq:mainDistinct}), stated for the general case of $\Nc+\Nr\geq 1$.
 And Lemma~\ref{lem:shiftCons}, which does not talk about $\Nc+\Nr$ at all.

We  already used the  {shifting} technique in Section~\ref{sec:strGA}.
For a vertex $a$, and cards $i,j$, with  $i\not\in a$, $j\in a$, 
$
a_{ij}=(a\setminus j)\cup \set{i},
$
denoted by an arc 
$
a \stackrel{ij}{\longrightarrow} a_{ij}.
$
For set $c$, we say that $a'$ is $c$-\emph{avoiding-reachable} from $a$ if there
is a directed path defined by a (possibly empty) sequence of arcs $\stackrel{ij}{\longrightarrow}$, all of them with
$i\not\in c$. 
The \emph{weight} of arc  $a \stackrel{ij}{\longrightarrow} a_{ij}$ is $i-j$,
and the weight of a sequence of arcs is the sum of their weights.
We are interested in zero-sum paths, because if there is a zero-sum path from $a$ to $a'$, then
$\chi_{modn} (a)=\chi_{modn}(a')$.
We use the following simple idea repeatedly, illustrated in Figure~\ref{fig-thPrivate-case3}.

\begin{lemma}
\label{lem:basicTool}
Let $c$ be a $\Nc$-set and $a$ an $\Na$-set, $a\subseteq \bar{c}$, with $\Na\geq 2$.
Consider two cards  $z_1,z_2\in a$.
If there exists an integer  $i$, $1\leq i \leq \floor{(z_2-z_1-1)/2}$  such that both 
$z_1+i\not\in a\cup c$ and
$z_2-i\not\in a\cup c$, then let  $y_1=z_1+i$ and $y_2=z_2-i$.
 The following is a zero-sum, $c$-avoiding path  from $a$ to  $a'$ 
$$
a \stackrel{y_1 z_1}{\longrightarrow} a_{1} \stackrel{y_2 z_2}{\longrightarrow} a' .
$$
Thus, $ \chi_{modn} (a)=\chi_{modn}(a')$, and
$a'\cap (c\cup \set{z_1,z_2}) =\emptyset$ and $a\cap (c\cup\set{y_1,y_2})=\emptyset$.
\end{lemma}

An immediate application of Lemma~\ref{lem:basicTool} is the following, illustrated
in Figure~\ref{fig-thPrivate-case3}.

\begin{lemma}
\label{lem:mainJd1Shft}
Assume $\Na\geq 2$, $\Nc\geq 0$, and $\Na+\Nc < \Nn/2$.
Let $c$ be a $\Nc$-set and $a$ an $\Na$-set, $a\subseteq\bar{c}$. 
For any two $z_1,z_2\in a$  there exist $y_1,y_2$ such that 
the following is a zero-sum, $c$-avoiding path  from $a$ to  $a'$ 
$$
a \stackrel{y_1 z_1}{\longrightarrow} a_{1} \stackrel{y_2 z_2}{\longrightarrow} a' .
$$
Thus, $ \chi_{modn} (a)=\chi_{modn}(a')$, and
$a'\cap (c\cup \set{z_1,z_2}) =\emptyset$ and $a\cap (c\cup\set{y_1,y_2})=\emptyset$.
\end{lemma}

\begin{proof}
Let $\ell_1=z_2-z_1-1$ and $\ell_2=z_1-z_2-1$.
Thus, $\ell_1$ is the number of cards in the interval $(z_{1},z_2)$ and $\ell_2$ is the number of cards in the interval $(z_2,z_1)$.

Let $i$, $1\leq i \leq \floor{\ell_1/2}$ be the smallest positive integer such that both 
$z_1+i\not\in a\cup c$ and
$z_2-i\not\in a\cup c$.
 If there exists such an integer, we are done, taking $y_1=z_1+i$ and $y_2=z_2-i$,
 noticing that $y_1\neq y_2$, since  $ i \leq \floor{\ell_1/2}$.
  Figure~\ref{fig-thPrivate-case3} illustrates three cases.
 \begin{figure}[h]
\centering
\includegraphics[scale=0.4]{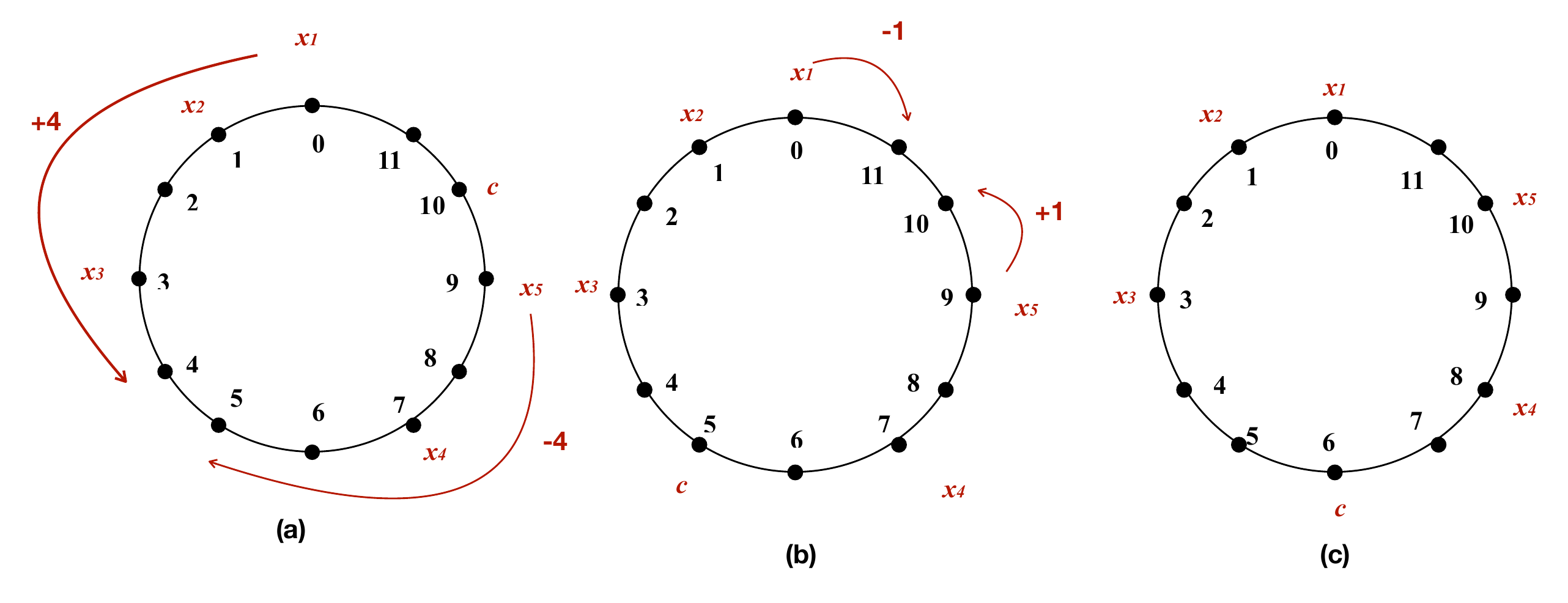} 
\caption{Case $\Nn=12,\Na=5,\Nc=1,\Nr=0$ of Lemma~\ref{lem:mainJd1Shft},
where $a=\set{x_1,x_2,x_3,x_4,x_5}$. In case (c) there is no two-step $c$-avoiding path for $x_1,x_5$}
\label{fig-thPrivate-case3}
\end{figure}
Otherwise, repeat the same argument on the other side, and we are done if there exists
 $i$, $1\leq i \leq \floor{\ell_2 /2}$  such that both 
$z_1-i\not\in a\cup c$ and
$z_2+i\not\in a\cup c$.
Thus (if we are not done),  there is a subset  $a_1$ of $\set{z_1+1,z_1+2,\ldots, z_2-1}$ such that $a_1\subseteq a\cup c$,
and $|a_1|\geq \floor{\ell_1 /2}$, and similarly, 
  a subset  $a_2$ of $\set{z_1-1,z_1-2,\ldots, z_2+1}$ such that $a_2\subseteq a\cup c$,
and $|a_2|\geq \floor{\ell_2/2}$,
and such that $|a_1|+|a_2|+|\set{z_1,z_2}| =\Na+\Nc$.

Hence, $\Na+\Nc \geq \floor{\ell_1/2}+\floor{\ell_2/2}+2$. But recall that  $\ell_1+\ell_2=\Nn-2$,
and thus, a simple case analysis about the parity of $\ell_1$ and $\ell_2$ shows that
$ \floor{\ell_1/2}+\floor{\ell_2/2}+2\geq \Nn/2$, a contradiction to the assumption that
$\Na+\Nc <  \Nn/2$.
\end{proof}

The previous Lemma~\ref{lem:mainJd1Shft} 
does not apply for $J(7,3)$, because in this case $\Na=3,\Nc=1$ and $\Na+\Nc>n/2$.
Indeed, the claim of the lemma is false in this case.
For example, taking $a=\set{0,1,4}$, and selecting $x_1=0,x_3=4$, the 
only possible $a'$ is $a'=\set{5,6}$, so in this case neither $5$ nor $6$ can 
take the value for $c$, they cannot be avoided. 
To deal with the symmetric case, where
 $\Na+\Nc=\floor{\Nn/2}$, the following lemma will be useful.

Notice the effect of shifting by one a vertex $a=\set{x_1,x_2,\ldots,x_{\Na}}$.
Namely,  $\chi_{modn}(\set{x_1+1,x_2+1,\ldots,x_{\Na}+1})=\chi_{modn}(\set{x_1,x_2,\ldots,x_{\Na}})+\Na$. Thus, 

\begin{remark}[Relatively prime]
\label{rm:prime}
If
 $\Na,\Nn$ are relatively prime, then for each $m\in\mathbb{Z}_\Nn$,
there exists an $x_m$, such that $a_m=\set{x_m,x_m+1,\ldots,x_m+\Na-1}$,
 $\chi_{modn}(a_m)=m$.  
\end{remark}

When we are satisfied that $a$ has only $\Na-1$ consecutive
 cards, we can use the following stronger claim.\footnote{
Lemma~\ref{lem:shiftCons} 
is similar to~\cite[Lemma 5]{Cordon-FrancoDF12}, except that this one
gives additional structure to the $\Na$-sets $a$, for $\Na\leq \Nn/2$.
}

\begin{lemma}
\label{lem:shiftCons}
Let $2\leq \Na\leq \Nn/2$.
For each $M\in\mathbb{Z}_\Nn$, and each $x_1\in\cards$, there is an $\Na$-set $a \in\chi_{modn}^{-1}(M)$, consisting of
at least $\Na-1$ consecutive cards, starting in either $x_1$ or $x_1+1$.
\end{lemma}

\begin{proof}
For the general case where $\Na,\Nn$ may not be relatively prime, assume w.l.o.g. that $x_1=0$.
We  prove that there are $\Nn$ distinct $\Na$-sets $a_r$,
such that for each $M$, one of them is in $\chi_{modn}^{-1}(M)$.
Each vertex  $a_r$ consists of $\Na-1$ consecutive values starting at either $0$ or $1$,
plus one additional value. 
For $\Na-1\leq r\leq 2\Na-2$, let $a_r=\{0,1,\ldots, \Na-2,r\}$. Thus,
\begin{align*}
a_{\Na-1}&=\set{0,1,\ldots, \Na-2, \Na-1},\\
a_{\Na}   &=\set{0,1,\ldots, \Na-2, \Na},\\
a_{\Na+1} &=\set{0,1,\ldots, \Na-2, \Na+1},\\ 
\, & \vdots \\
a_{2\Na-2} & =\set{0,1,\ldots, \Na-2, 2\Na-2}.
\end{align*}
Notice that each $a_r$ consists of a set of $\Na$ distinct values, since  we are assuming $\Na\leq n/2$.
Now, for $2\Na-1\leq r\leq n+\Na-2$,
 let $a_{r} = \{1,2,\ldots, \Na-1,r-\Na+1\}$. Thus,
 \begin{align*}
 a_{2\Na-1} & =\set{1,\ldots, \Na-1, \Na},\\ 
a_{2\Na} & =\set{1,\ldots, \Na-1, \Na+1},\\
a_{2\Na+1} &=\set{1,\ldots, \Na-1, \Na+2},\\
\, & \vdots \\
a_{n+\Na-2} &=\set{1,\ldots, \Na-1,n-1}.
\end{align*}
Again,  each $a_r$ consists of a set of $\Na$ distinct values, since  we are assuming $\Na\leq n/2$.
  Notice that  $\chi_{modn}(a_{2\Na-1})=\chi_{modn}(a_{2\Na-2})+1$ (mod $n$).
And in general,    $\chi_{modn}(a_{r+1})=\chi_{modn}(a_{r})+1$ (mod $n$).
In total, we have that $a_{\Na-1},\ldots, a_{n+\Na-2}$ are $n$ distinct values
(mod $n$).
Thus, for each $M$, there is one $a_r \in\chi_{modn}^{-1}(M)$.
\end{proof}

\subsection{If $\Nc+\Nr= 1$ then $\chi_{modn}$  is  safe}
\label{sec:modSafe}

Now we show that when $\Nc+\Nr= 1$, the  well-known codes
described in Section~\ref{sec:modInfo}, defined by $\chi_{modn}^{-1}$, 
are  safe. We prove it 
 using  the elementary additive number theory properties (the theorem  generalizes and simplifies results of~\cite{Cordon-FrancoDF12})\footnote{
In~\cite[Corollary 9]{Cordon-FrancoDF12}  it is shown that the protocol is safe
when $\Nn$  prime, 
with a proof based on a non-trivial theorem 
by Dias da Silva and Hamidoune~\cite[Theorem 4.1]{DaSilvaDH94}. Which is analogous to
the Cauchy–Davenport theorem, the first theorem in additive group theory~\cite{mann1973}.
Then, this result was extended to~\cite[Theorem 13]{Cordon-FrancoDF12}, proving that 
a protocol that announces
the sum of the cards modulo $p$ is safe, except for $(4,3,1)$, $(3,4,1)$,
where $p$ is the least prime greater than or equal to $\Na+\Nb+1$.
For this, Bertrand's postulate, as well as a theorem of Nagura~\cite{nagura1952} was used
(stating that one can always find a prime number relatively close to a given integer).
}
 of Section~\ref{sec:3prop}.

Recall the safety characterization of Theorem~\ref{th:safetyMain}.
 Instantiated for
 protocol $\chi_{modn}$
 it says that\footnote{
 It is similar to~\cite[Proposition 6]{Cordon-FrancoDF12}, except that this proposition
 also says that if $\chi_{modn}$ is safe, then for each value $M$ of $\cards$, there is an $a\subseteq \bar{c}$  for which $\chi_{modn}(a)=M$. 
  }
$\chi_{modn}$ is {safe} 
if and only if for each  $\Nc$-set $c$, 
$y\in\bar{c}$, and $M\in \chi_{modn}(K_p(\bar{c}))$, 
there exist $\Na$-sets $a,a' \subseteq \bar{c}$, $\chi_{modn}(a)=\chi_{modn}(a')=M$ such that $y\in a\triangle a'$.
Thus, we can assume that $\Nc=1,\Nr=0$, because proving that  the protocol is safe in this case,
implies that it is  safe when  $\Nc=0,\Nr=1$.

%

The conditions that  $\Na,\Nb\geq 3$ are necessary,
by Corollary~\ref{cor:mainBoundInfoSafe}.
Also, $\Nn\geq 7$, because  if $\Na=3$ and $\Nn=6$, then the protocol is not safe.
For instance, if $C$ has hand $5$ and hears announcement $4$
(because $A$ has hand $\set{0,1,3}$),
then she can deduce that $A$ does not have card $4$.

In the proof  we will assume that $\Na\leq \floor{\Nn/2}$,
by the duality Theorem~\ref{th:dual}.
Furthermore, to make the proof more elegant, we prove the 
  (almost) symmetric cases where $\Na+\Nc \geq  \Nn/2$ separately,
  in Appendix~\ref{app:symmModn}. Then, we can  use Lemma~\ref{lem:mainJd1Shft} directly.

\begin{theorem}
\label{th:safetyRussCards}
The protocol $\chi_{modn}$ is informative and safe when $\Nc+\Nr=1$, $\Na,\Nb\geq 3$, $\Nn\geq 7$.
\end{theorem}

\begin{proof}
We already saw that  $\chi_{modn}$ is informative, in Lemma~\ref{lem:properColBasic}.
To prove safety,
as explained above, we may assume that $\Nc=1,\Nr=0$.
Also, we may assume that $\Na\leq \floor{\Nn/2}$,
by the duality
Theorem~\ref{th:dual}.

We have considered the cases where:
$\Nn=2\Na+1$ in Lemma~\ref{lm:safetyRussCprime},
 $\Na=\Nn/2-1$ with both $\Nn$ and $\Na$ even in Lemma~\ref{lem:symmetricRCevenA},
 and $2\Na=\Nn$ in Lemma~\ref{lem:symmetricRC}.
Thus, we may assume that 
 $\Na+1<  {\Nn/2}$, and we can use Lemma~\ref{lem:mainJd1Shft} directly.

Consider an $M\in\mathbb{Z}_\Nn$ and $c\in\cards$. Let $y\in\bar{c}$.
First we show that there is an $\Na$-set $a_1\subseteq\bar{c}$, such that $y\in a_1$ and $\chi_{modn}(a)=M$.

For  $M\in\mathbb{Z}_\Nn$ and $x_1=y-1$, let  $a_1=\set{x_1,x_2,\ldots,x_{\Na}}$ be the set defined by Lemma~\ref{lem:shiftCons}.
Thus, $\chi_{modn}(a_1)=M$, and $a_1$
consists of at least $\Na-1$ consecutive
cards starting in either $x_1$ or $x_1+1$, thus, $y\in\set{x_1,x_2}$, and in both cases, $y\in a_1$.
If $a_1\subseteq\bar{c}$ we are done. 

Thus, assume  $c\in a_1$. 
Then, we use   Lemma~\ref{lem:mainJd1Shft}, to remove $c$ from $a_1$,
without touching $y$.
Namely, we apply the lemma with
 $c$ and any other card of $a$ different from $y$.
We have shown that there is an $\Na$-set $a_1\subseteq\bar{c}$, such that $y\in a_1$ and $\chi_{modn}(a_1)=M$.
 Figure~\ref{fig-safetyRussCards1} illustrates three cases, that can be dealt with,
 even when $\Nn/2=\Na+1$.
 \begin{figure}[h]
\centering
\includegraphics[scale=0.4]{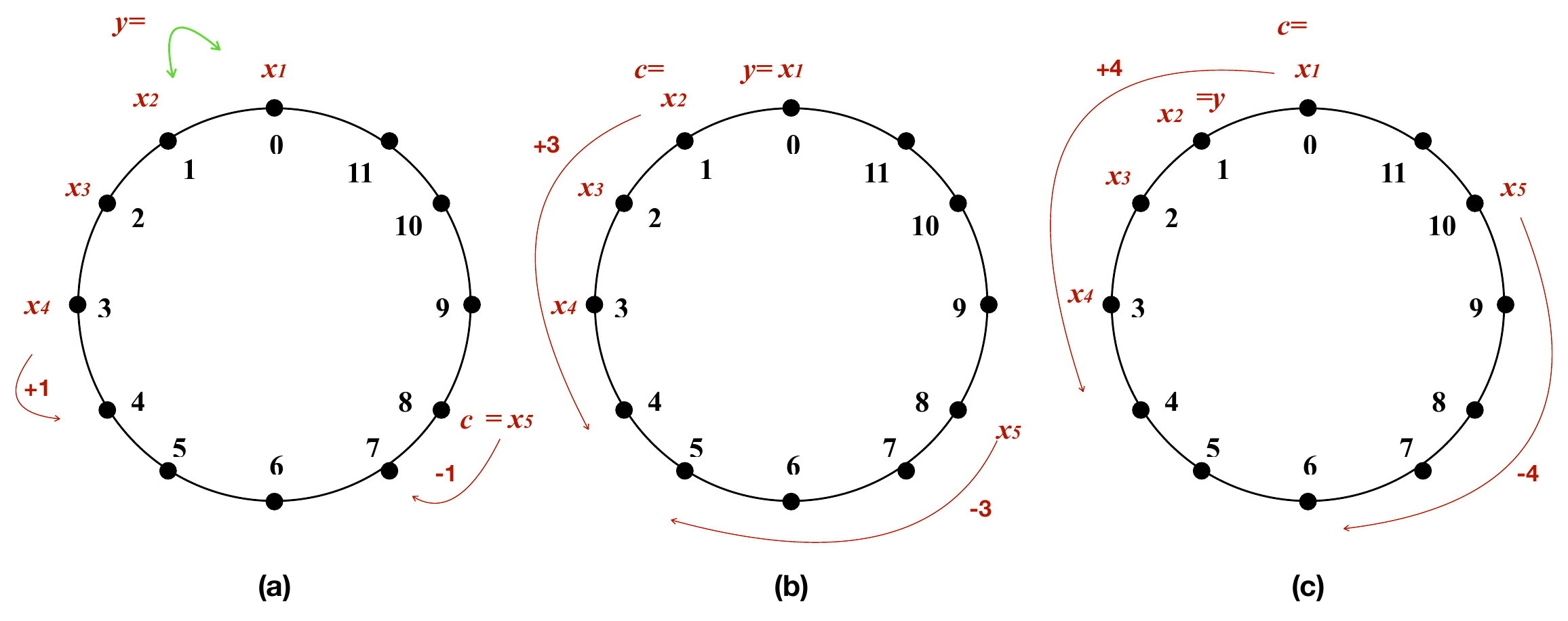} 
\caption{Case  $\Nn=12,\Na=5,\Nc=1,\Nr=0$}. 
\label{fig-safetyRussCards1}
\end{figure}

To complete the proof, we need to show that  there
is an $\Na$-set $a_2\subseteq\bar{c}$, such that $y\not\in a_2$ and $\chi_{modn}(a_2)=M$.
This is done again by a direct application of 
Lemma~\ref{lem:mainJd1Shft}, removing from $a_1$ any two cards that include $y$,
without including $c$. 

%
%
\end{proof}

\section{Informative transmission: the general case $\Nc+\Nr\geq 1$}
\label{sec:general-c}
\label{sec:modInfo-gen-c}

In this section we briefly discuss an informative solution when $\Nc+\Nr\geq 1$.
As far as we know, this is the first informative protocol,
and there is no safe and informative general solution known.
Swanson et al.~\cite{SSadd2014} discuss informative protocols and
their relation to combinatorial designs. They  explain the combinatorial difficulty
of the case  $\Nc+\Nr\geq 1$.

 
We have seen in Section~\ref{sec:modInfo} that $\chi_{modn}$ is informative when  $\Nc+\Nr=1$,
but not when $\Nc+\Nr>1$, namely,
  $\chi_{modn}$ is not a  proper vertex coloring of $J^{\Nc+\Nr}(\Nn,\Na)$.
We are now   behind  the classic 
 coding theory proof that shows a lower bound  on $A(n,2\delta,w)$,
 the maximum number of codewords in any binary code of length $n$,
 constant weight $w$, and Hamming distance $2\delta$.
Namely, the proof that shows that the vertices in $\chi_{modn}^{-i}$ in this case
define a binary code of length $n$, constant weight $w$, and Hamming distance $2\delta$.

We rephrase the coding theory argument from~\cite[Theorem 4]{GrahamS1980} in our notation.
Let $q$ be a primer power (positive integer power of a single prime number), $q\geq n$. Let the elements of the Galois field $\text{GF}(q)$ 
be $w_0,w_1,\ldots, w_{q-1}$. For a vertex $a$ of $J^d(\Nn,\Na)$, let $a_i=1$ if $i\in a$, and else $a_i=0$.
Namely, for the following lemma we view $a$ as a vector  
${a}=(a_0,\ldots,a_{n-1})\in \mathbb{F}_{\Na}^{n}$.
Define $\bar{\chi}({a})$ to be the vector  $(\chi_1(a), \chi_2(a), \ldots, \chi_{d}(a))$,
\begin{equation} \label{eq1}
\begin{split}
\chi_1({a}) & = \sum_{a_i=1} w_i, \\
\chi_2({a})  & = \sum_{\substack{ {i<j}\\ {a_i=a_j=1}}} w_i w_j, \\
\chi_3({a})  & = \sum_{\substack{ {i<j<k}\\ {a_i=a_j=a_k=1}}} w_i w_j w_k, \\
\cdots
\end{split}
\end{equation}
Then, for $\vec{v} \in \text{GF}(q)^{d}$, the set of vertices colored $\vec{v}$ is $\bar{\chi}^{-1}(\vec{v})$.

Recall that if $d\geq \min\{\Na,\Nn-\Na \}$  then $J^d(\Nn,\Na)$ it is a complete graph.

\begin{lemma}
\label{lem:properColBasicG}
$\bar{\chi}$
 is a proper vertex coloring of $J^d(\Nn,\Na)$, $d\geq 1$,
 and $d<\min\{\Na,\Nn-\Na \}$. 
\end{lemma}
\begin{proof}
Consider two vertices $a,b$ of  $J^d(\Nn,\Na)$ 
  viewed as vectors of $\mathbb{F}_{\Na}^{\Nn}$,
and such that $\bar{\chi} (a)= \bar{\chi} (b)$.
Assume for contradiction that $a$ and $b$ are adjacent.
Thus, 
 there are $2\gamma$  distinct coordinates $r_1,\ldots, r_\gamma$, $s_1,\ldots,s_\gamma$,
$\gamma\leq d$,
where $a$ and $b$ disagree, and on all other coordinates they agree.
Say, $a_{r_i}=1$ while  $b_{r_i}=0$, and conversely, $a_{s_i}=0$ while  $b_{s_i}=1$  ($1\leq i\leq \gamma$).
Write $\alpha_i=w_{r_i}$, $\beta_i=w_{s_i}$ ($1\leq i\leq \gamma$).
Since  $\bar{\chi} (a)= \bar{\chi} (b)$ we have
\begin{align*}
\sigma_1 &= \sum_i \alpha_i    = \sum_i \beta_i        \\
\sigma_2 &= \sum_{i<j} \alpha_i\alpha_j    = \sum_{i<j} \beta_i\beta_j        \\
\, & \cdots  \,  \\
\sigma_{d}&=\sum_{i_1<\cdots < i_{d}} \alpha_{i_1} \cdots  \alpha_{i_{d}}   = \sum_{i_1<\cdots < i_{d}} 
\beta_{i_1} \cdots  \beta_{i_{d}}
\end{align*}
Therefore, $\alpha_1,\ldots, \alpha_\gamma, \beta_1,\ldots,\beta_\gamma$ are $2\gamma$
distinct zeros of the polynomial 
$$
x^\gamma - \sigma_1 x^{\gamma-1}+\sigma_2 x^{\gamma-2}-\cdots \pm \sigma_\gamma.
$$
But a polynomial of degree $\gamma$ over a field has at most $\gamma$ zeros.
\end{proof}

Thus, the set of colors needed is of size at most $q^d$.
Which implies that there is always a set of size at most $(2\Nn)^d$ to properly
color $J^d(\Nn,\Na)$,
 because 
  Bertrand's postulate states that  there is a prime $p$ such that $\Nn <p\leq 2\Nn$.

On the other hand, there is a corresponding (asymptotically in terms of $\Nn$, for $\Nc+\Nr$ constant) lower bound\footnote{
Recall that $\binom{z+k}{k}={{k^z}\over{\Gamma(z+1)}}(1+{ z(z+1)\over 2k} + O(k^{-2}))$,
as $k\rightarrow \infty$.
Thus, in more detail, the lower bound in the number of bits is $\Theta( (\Nc+\Nr)\log \Nn -(\Nc+\Nr)\log(\Nc+\Nr))$. 
}.
Namely, by Lemma~\ref{lem:neighbrs}, the clicks
$K_{p} (\bar{b})$ in $\cG_B$ have size
 $p= \binom{\Na+\Nc+\Nr}{\Na}$, and
 by Lemma~\ref{lem:isomorphGenJ},  $J^d(n,m)\cong  J^d(n,n-m)$.
 Thus,
 

\begin{theorem}
\label{th:mainInfoGeneral}
 $\Theta( (\Nc+\Nr)\log \Nn)$ bits are needed and sufficient for an informative protocol.
 \end{theorem}

\section{Conclusions}
\label{sec:concl}

We have presented a new perspective that  brings closer
together previous research on secret sharing and Russian card problems,
by defining the underlying basic problem of safe information transmission
from $A$ to $B$ in the face of an eavesdropper $C$.
The new perspective inspired by distributed computing is based on a formalization in terms of Johnson graphs,
which facilitates using known results about these graphs,
closely related to coding theory, and motivates developing
new additive number theory proofs. We are able thus to prove
new results, as well as explaining and unifying previously known results.

We have assumed, following these previous research lines, that
the inputs are correlated, using a deck of cards.
Considering a deterministic protocol $P_A$ for $A$, we
stayed with the common definition of safety, requiring that $C$ does not
learn any of the cards of $A$ after listening to her announcement.
Also,  we considered the standard definition of informative, requiring
the $B$ learns her whole hand.
We defined a new requirement, of minimal information transfer,
requiring that always $B$ learns \emph{something} about $A$'s hand.

Many interesting avenues remain for future work.
Some problems would imply solutions in coding theory, where 
much research has beed done; the smallest number of messages
needed for informative information transmission is equivalent
to finding the chromatic number of a Johnson graph, 
a question of wide interest which is open even in the case of $J(\Nn,\Na)$, not to mention
the general case of $J^d(\Nn,\Na)$.
For  $d=\Nc+\Nr=1$, we have described solutions which show that no more than 2 additional messages
are needed to go from the known proper coloring solutions with $\Nn-2$ messages (or more), 
to a proper coloring that is additionally safe, with $\Nn$ messages. 
A  thorough study of the general case $d\geq 1$ is beyond the scope of this paper.

The colorings for minimal information transmission do not seem to have been studied
before. Even very concrete cases remain open.
We showed that there is a minimally informative safe protocol for the Russian cards problem $(3,3,1)$ 
with only two messages, but the solution was found using a computer program.
The modular algorithm $\chi_2$ works only in the
cases described by Theorem~\ref{th:mainMinInfo}, and this solution  uses 3 messages.

It would be of course interesting to consider randomized solutions, and the relation to the
%
Fischer, Paterson and Rackoff~\cite{FischerPR89secBitTr} approach.
They consider
the problem of $A$ and $B$ agreeing on a bit that is  secret from $C$, using 
randomized  protocols.
They then mention that it is not clear how to get rid of randomization, because
the protocol itself is known also to $C$, and illustrate the difficulty with the following
example, using the notion of \emph{key set}  $\set{x,y}$. This notion plays a crucial role
in the algorithms of this paper, and subsequent ones.  A key set  consists  of one card of $A$ and one of $B$,
 equally likely, given the information available to $C$, that $A$ holds $x$ and $B$ holds $y$,
or the opposite. 
Then $A$ and $B$ can obtain a secret bit $r$ from the key set, say $r=0$ if $A$ holds
the smaller card.
If $A$ announces a key set $\set{x,y}$ by picking the smallest card in her hand for $x$
and the smallest card not in her hand for $y$, then she may be revealing her entire hand
by announcing $\set{x,y}$.
Nevertheless,  they describe   a deterministic protocol where 
$A$ and $B$ exchange  message several rounds to agree on a bit
that is secret to $C$, that works when $\Nc\leq \min (\Na,\Nb)/3$,  $\Na,\Nb\geq 1$.
Notice that the solutions
that we discuss  do hide form $C$ the location of the cards of $A,B$, while this is not the case for their protocol (but we have not shown that  our solutions satisfy their requirment
that all inputs are equally likely).
%

Our protocols that send one bit, by which $B$ learns one of the
cards of $A$, are somewhat reminiscent of the widely studied oblivious transfer problem~\cite{Schoenmakers2011}, but passive,
in the sense that the whole interaction consists
of $A$ sending a message to $B$, and $A$ does not know 
 which of her cards were learned by $B$.  Namely,
$B$ has no say as to which card he wishes to learn.

%
 
 Notice that
 a solution to the Russian cards problem implies a solution to the secret key problem.
When $\Nc=1,\Nr=0$, consider the $N= {{\Nn-1}\choose{\Na}}$  possible deals to $A$ and $B$, all 
 possible from the perspective of  $C$, indexed
 from $0$ to $N-1$ in some predetermined way, and let $r$ the index of the actual deal. 
 Both $A$ and $B$ can compute $r$, while $C$ has no information about it~\cite{FischerPR89secBitTr}.
 Thus, $A$ and $B$ can share a string of $\log_2 N$ bits, without revealing any of their cards to $C$,
using the Russian cards protocol with signature $(\Na,\Nb,1)$, where $A$ sends a string of $\log_2 (\Nn)$ bits   
and $B$ answers with a $\log_2 (\Nn)$ bit string (again, not clear that they are all equally likely).
%

Many other interesting problems remain open, about the relation
with combinatorial designs that has been thoroughly studied e.g.~\cite{SSadd2014},
about stronger security
requirements e.g.~\cite{LFcaseStudy17},  about fault-tolerant solutions~\cite{HerlihyKR:2013},
and more than two parties e.g.~\cite{Duan2010}. It would be interesting to
understand the role of Johnson graphs in multi-round protocols; there exists work
both from the secret sharing side e.g.~\cite{FischerW92crypto},  and from the Russian cards side~\cite{colorGRSP13,Ditmarsch2011ThreeS}, and of course in distributed computing, although
without preserving privacy~\cite{DFRbits2020}.

%
%
%
%
%
%

\section*{Acknowledgements}
We would like to thank Zoe Leyva-Acosta and Eduardo Pascual-Aseff for their many  comments, and
for finding the 2-message minimally informative solution to the Russian cards problem.
Also, Jorge Armenta,  for his help in the early stages of this research.
This work was supported by  the UNAM-PAPIIT project IN106520.

\bibliographystyle{splncs04}
\bibliography{lit}

%
\appendix
\section{Related work.}
\label{sec:relatedWork}
\label{app:relatedW}
There is related work in several domains: Russian card problems, Johnson graphs, coding theory, 
additive number theory, unconditionally secure key exchange, distributed computability
and correlated inputs.

\subsection{Russian cards}
\label{sec:relatedWork:russianCards}

Many instances of the generalized Russian cards problem have been studied,
included  where the cards are dealt over more than three agents and work on cryptography.
The generalized Russian cards problem has close ties to the field of combinatorial designs, 
particularly for perfect security notions~\cite{LFcaseStudy17,SSadd2014}.
The traditional security requirement of the Russian cards problem,
which is the one we consider,  $C$ may not know with certainty who holds
any given card, that does not mean that she may not have a high probability of guessing this
information correctly. To this end, stronger notions of security have been studied in these papers.

Cord\'on--Franco et al.~\cite{Cordon-FrancoDF12}
investigate  conditions for when $A$ or $B$ can safely announce the sum of the cards they hold modulo 
 the smallest prime greater than or equal to $\Nn = \Na + \Nb + \Nc$. 
 They hold whenever $\Na, \Nb > 2$ and $\Nc = 1$, except for the cases $(3, 4, 1)$ and $(4, 3, 1)$.
 The exceptional cases $(3, 4, 1)$ and $(4, 3, 1)$ are treated separately using Haskell,
 and shown to work with modulo $\Nn$.
 They observe that because $C$ holds a single card, this also implies that $A$ and $B$ 
will learn the card deal from the announcement of the other player. For the general case when $\Nc\geq 1$
they give a characterization of when the protocol is safe, but  notice that the protocol is informative only when $\Nc=1$.

Albert et al.~\cite{Albert2011securComm} investigate both the problem of communicating the entire hand  and communicating a secret bit.
 The analysis includes a sum announcement protocol for the case $(k,k,1)$, where $k\geq  3$; 
  both players announce the sum of their cards modulo $2k+1$. In addition, they show that state safe implies bit safe, and pose the  open question of whether a protocol for sharing a secret bit implies the existence of a protocol for sharing states/card deals.

There are several  additional ways of restating the safety property of Definition~\ref{def:safProt},
such as  CA2 and CA3  from~\cite{Albert2005SafeCF}.

\begin{lemma}[Safety characterization]
\label{lm:safetyMain}
Let $P_A: \powersetS{\Na}{\cards} \rightarrow \MA$.
The following  conditions are equivalent.
\begin{enumerate}
\item \label{lm:safetyMain1}
$P_A$ is {safe}.
\item  \label{lm:safetyMain3}
 For each  $M\in \MA$,  $\Nc$-set $c$,  the following holds.
Let $X_{\bar{c}}$ be the subset  of  $P_A^{-1}(M)$ avoiding $c$.
If $X_{\bar{c}}\neq \emptyset$ then for any
 $y\not\in c$, 
there exist $a,a' \in X_{\bar{c}}$ 
such that $y\in a\triangle a'$.
\item \label{lm:safetyMain4}
 For each  $M\in \MA$, 
 \begin{description}
\item[CA2] for every $\Nc$-set $c$ the members of $P_A^{-1}(M)$ avoiding $c$ have empty intersection, and
\item[CA3] for every  $\Nc$-set $c$ the members of $P_A^{-1}(M)$ avoiding $c$ have
 union consisting of all cards $\cards$ except those of $c$.
\end{description}
\end{enumerate}
\end{lemma}

\subsection{Johnson graphs and algebraic graph theory}
\label{sec:relatedWork:johnsonG}
As we show here,  Johnson graphs capture the relations induced by correlated inputs
defined by a deck of cards. Furthermore,
certain vertex colorings of Johnson graphs turn out to capture 
essence behind  information transmission with such correlated inputs.
Johnson graphs, Kneser graphs and other related highly  symmetric graphs have been well studied through algebraic methods~\cite{GodsilRalgGraphTh}, and  in  spectral analysis of graphs~\cite{horaObata}.
They are related to the  Erd\"os--Ko–Rado Theorem, one of the fundamental results in combinatorics
about intersecting families of sets. 
Its proof uses a simple yet useful operation  called \emph{shifting}, that we use too.
The symmetry  and algebraic properties of Johnson graphs are well understood, yet,
although their chromatic number is important, especially in coding theory,
it remains an open problem, see~\cite[Chapter 16]{godsil_meagher_2015} where there
is a summary of known results, as well as in~\cite{BE11}.

\subsection{Coding theory}
\label{sec:relatedWork:codingT}
Vertex colorings of Johnson graphs  are closely related to coding theory.  
Coding theory  captures  necessary properties
 for  information transmission with such correlated inputs; but the properties are
 not sufficient for the safety
requirement that $C$ does not learn about the inputs, for this, additional properties about the codes
are needed.
The independence number of the Johnson graph $J(n,m)$ is the size of the largest constant weight code with word length $n$, weight $m$, and minimum distance $4$. The chromatic number is the minimum number of parts in a partition into such constant weight codes. There is a lot of literature, due to its combinatorial interest and also applications.
For instance,
Smith et al.~\cite{SmithHP06} extend known tables  of constant weight codes of length $n\leq 28$ up to 63,
motivated by the generation of frequency hopping lists for use in assignment problems
in radio networks. Large distance between codewords gives smaller overlap between lists.
This leads to fewer clashes on the same frequency and so less interference.  Similarly, a
larger number of codewords allows larger list re-use distances in the network and again
leads to lower interference.

A binary constant weight code of word length $n$ and weight $w$ and 
distance $d$  is a collection of $(0,1)$-vectors of length $n$, all having $w$ ones and
$n-w$ zeros, such that any two of these vectors differ in $d$ places.
The Johnson graph $J(n,w)$ is the graph on the binary vectors of length $n$
and weight $w$, adjacent when they have Hamming distance $2$.

The chromatic number of $J(\Nn,w)$ is the minimum number of disjoint constant weight codes of length $n$,
weight $w$, and distance $4$, for which the union is the set of all $n$-tuples with weight $w$.
It is also the minimum number of disjoint packings of $(w-1)$-subsets by $w$-subsets, for which the union
is the set of all $w$-subsets of the $n$-set.
Let $(n,d,w)$ denote a code of length $n$, constant weight $w$, and distance $d$, and let
$A(n,d,w)$ denote the maximum size of an $(n,d,w)$ code. Graham and Sloan~\cite{GrahamS1980} proved, for $d=4$,
that $\chi (J(n,w))\leq n$ for all $0\leq w\leq n$. The proof is actually by the
same algorithm of the Russian cards problem:
putting the structure of abelian group on the coordinate positions, and all words with given sum of the elements in the support form a constant weight code with minimum distance 4.
They present a generalization for all $d$, using an algorithm where a color is a vector,
giving an upper bound for the number of colors need to color $J^d(n,w)$, and
that we describe in Section~\ref{sec:modularAlgo}. This and other more complicated
methods, as well as explicit tables  are described in~\cite{BSSS90},
where the importance in combinatorics and coding of $A(n,d,w)$ is emphasized. 
Although the chromatic number of Johnson graphs have been thoroughly studied, 
there seem to be no non-trivial general lower bounds.
Apparently only a few cases are known where 
$\chi (J(n,w))< n$, and in those cases, $\chi (J(n,w))\geq n-2$, see Brouwer and Etzion~\cite{BE11}.
In general, determining the chromatic number of a Johnson graph is an open problem of wide interest~\cite{godsil_meagher_2015}.

\subsection{Combinatorial Designs}
Coding theory is an enormous topic in its own right, but  some results  are closely connected to another old and large topic: combinatorial designs.
The theory  of designs concerns itself with questions about  subsets  of a set  possessing a high degree of regularity, thus, 
the generalized Russian cards problem has close ties to the field of combinatorial designs.
The signature $(3,3,1)$ was first considered by Kirkman~\cite{kirkman}, 
who suggests a solution using a design.
The design consists of seven triples, which  are precisely the lines that form the projective geometric plane.
  Particularly for perfect security notions, designs are important,
  as demonstrated in~\cite{LFcaseStudy17,swansStinson14,SSadd2014}.
Such notions require $C$ not gaining any probabilistic advantage in guessing the fate of some set of 
$\delta$ cards,  \emph{perfect $\delta$-security}. An equivalence between
  perfectly $\delta$-secure strategies and $(c+\delta)$-designs on $n$ points with block size $\Na$, when announcements are chosen uniformly at random from the set of possible announcements
  is established. Also,  example solutions are provided, including a construction that yields
perfect $1$-security against  when $\Nc=2$, and  a construction  strategy with 
$\Na = 8, \Nb = 13$, and $\Nc = 3$ that is perfectly $2$-secure.
Notice that such stronger security notions requiere protocols that use a larger set of possible messages.

\subsection{Additive number theory}
\label{sec:relatedWork:additiveNumTh}

While coding theory properties are necessary for informative properties of the protocol,
to be safe, additional properties are needed, which define  additive number theory problems, 
at least when working with  additive protocols such as those in~\cite{Cordon-FrancoDF12} 
and those we consider in Section~\ref{sec:sBitTransm} and~\ref{sec:modularAlgo}.
Announcing the cards modulo $7$  was among the answers to a Moscow Mathematics Olympiad problem~\cite{makarychev2001importance}
that motivated subsequent work on Russian card problems.

 Although finding solutions to a linear congruence is a classic problem,  less seems to be known 
when the solution is required to be with distinct values~\cite{adamsP2010,HegdeM2016}, 
the question seems to have been studied first only fairly recently in~\cite{adamsP2010}, and a characterization of when a linear congruence
$$
\alpha_1 x_1+\alpha_2 x_2+\cdots \alpha_\Nn x_\Nn \equiv \alpha \pmod\Nn
$$
with $\alpha, \alpha_1,\ldots\alpha_n\in\mathbb{Z}$ has solutions with distinct values
has been  presented in~\cite{GrynPhiPon2013}.
The characterization implies that in our case (where the first $\Na$ coefficient $\alpha_i=1$ and the others are 
equal to $0$) the congruence has a solution, for every $\alpha\in \mathbb{Z}_\Nn$, 
a fact that can be proved directly rather easily (see~\cite[Lemma 5]{Cordon-FrancoDF12}),
but
to prove safety we need a  more detailed analysis, as explained in Section~\ref{sec:modSafe}.
The question has interesting applications and relations to weighted sub- sequence sum questions,
as described in these papers. Some work exists motivated by a 1964 conjecture
by Erd\"os, and Heilbronn\cite{erd1964addition} 
giving lower bounds of the number of distinct sums of $\Na$-subsets of $\mathbb{Z}_\Nn$,
 later proven by Dias da Silva and Hamidourne~\cite{DaSilvaDH94}, which is what is used
 to analyze the modular algorithm in~\cite{Cordon-FrancoDF12}.

\subsection{Unconditionally Secure Secret Key Exchange}
\label{sec:relatedWork:secureSKeyExch}
The idea that card games could be used to achieve perfect cryptography without further assumptions 
proposed by  Peter Winkler in 1981 in the context of the game of Bridge,
led to a sequence of papers by Fischer and Wright.
Peter Winkler~\cite{Winkler83} developed bidding conventions whereby one bridge player could send her partner 
secret information about her hand that was totally unrelated to the actual bid and completely undecipherable to the opponents, even though the protocol was known to them.
 Much work has continued to be done, especially on the randomized setting,
see e.g. for a more recent paper~\cite{KMizukiTN08}, and information theoretic~\cite{MaurerW99}.

Fischer and Wright's~\cite{FWeffic93}  motivation of considering  \emph{card games}, where $A,B,C$ draw cards from a deck of $d$ cards,
as specified by a signature $(\Na,\Nb,\Nc)$, with $\Na+\Nb+\Nc=\Nn=d$
(in~\cite{FW96} they also discuss a bit the case where there is a card which nobody gets), is as follows.
It is desired correlated random initial local variables for the players, that have a simple structure and a small amount
of initial information. By looking at her own cards, a player gains some information about
the other players' hands: a set of cards that appear in no other player's hand.
It is noted that if the initial local variables are uncorrelated, 
an eavesdropper can simulate any player over all random choices and all possible initial random values and 
learn the secret key.
Thus, Fischer and Wright ask: 
We would like to know which distributions of private initial values allow any team that forms to obtain 
an $n$-bit secret key. Although their protocols use randomization, they
require they \emph{always} work, the key to be \emph{completely} secret from 
a computational unbounded eavesdropper, and \emph{exactly}
known by all players (so standard techniques based on computational difficulty cannot be used).


Fischer and Wright~\cite{FW96} have explored in this and other papers
the problem of players  sharing a secret key using a  deal of cards while their conversation is
overhead by $E$, inspired by the work of Winkler.
They present a general model for communication among players overheard by a passive eavesdropper $E$,
in which all players including $E$ are given private inputs that may be correlated.
They study  secret key exchange in this model.
In particular, they consider the situation in which the team players are dealt hands of cards of 
prespecified sizes from a known deck of distinct cards.
They consider both the cases where $E$ gets the remaining card, and where she gets no card.
They start with an example of a deck of four cards, $A$ is given two, and $B$ one.
They explain that, if $E$  does not see the remaining card or if $A$ and $B$ can use randomization,
then $A$  and $B$ can agree on a perfectly secret bit. If $E$ sees
the remaining card or $A$ and $B$  are required to behave deterministically, 
then $A$ and $B$ cannot agree even on a weakly (strong requires equal probability) secret bit.
More generally, in $N$-\emph{valued multiparty secret key exchange} the players chose a value $v$ from
a known set of $N$ values. In the \emph{perfect} version $E$ considers all $N$ values equally likely,
while in the \emph{weak} version she considers all $N$ values possible.
They define it in terms of three requirements. \emph{Agreement} is met if all parties know
the secret key $B$; \emph{secrecy} is met if the eavesdropper's probability of guessing $B$ correctly
is the same before and after hearing the communication; \emph{uniformity} requires
that $B$ has equal probability of being any of the $2^n$ possible $n$-bit sequences.
Notice that there is no explicit requirement saying that the eavesdropper should not learn
any of the input bits of the players (and indeed in some of their protocols the eavesdropper
learns cards of the other players, e.g. the one-bit secret key exchange protocol in~\cite{FischerW92crypto}),  while this is a requirement for Russian cards games.

In particular,  they show that secret key exchange is not possible if the player's inputs are not correlated.
A \emph{signature} $(s_1,\ldots,s_k;d)$ specifies the hands size $s_i$ for each player and the deck size $d$.
The perfect (resp. weak) capacity of a signature is the largest $N$ such that $N$-valued perfrect (resp.)
weak) secret key exchange is possible when the deal is chosen randomly as specified by the given signature.
Previous work was informal, and some studied the case of $N=2$ and two players.

 Fischer and Wright~\cite{FWeffic93}
proposed a method for reducing the problem of a multi-party $n$-bit secret key exchange to the problem 
of a $2$-party $n$-bit secret key exchange. They present a simulation (that needs randomization),
and needs that the deals in the multiparty signature are large enough.
Hence, using this method, one can easily extend a  protocol from $A,B$ so that it performs a $p$-party $n$-bit
 secret key exchange with $p\geq 3$. In this paper they also describe the \emph{transformation protocol}
 for two parties. This protocol is later improved in~\cite{KMizukiTN08}, where a detail and clear
 analysis  is presenting, showing 
  that the improved transformation protocol establishes 
an $n$-bit secret key exchange for a signature $(a,b;e)$ if and only if $\Psi(a,b;e)\geq n$,
$\Psi$ a function which is approximately proportional to $d$, where $d$ is the number of distinct
cards in the deck. 
 For \emph{key set} protocols, Fischer and Wright show that $A$ and $B$ can share a bit
 if and only if $\Na+\Nb\geq \Nc+2$, this is reported in~\cite{MizukiSN02journal} 
 (journal version of~\cite{MizukiT99}), where
 a characterization for the signatures that are solvable by key set protocols is presented,
 and observe that actually the transformation protocol of Fischer and Wright~\cite{FWeffic93} can deal with a case that is not solvable
 by key set, namely $(3,2;4)$. 
 All this is  for randomized protocols,
 the only case of deterministic protocols that we are aware of is Fischer, Paterson and Rackoff~\cite{FischerPR89secBitTr}, where they give a protocol for secret bit transmission, and show it works
 if $\Nc\leq\min(a, b)/3$. Notice that this protocol is not private against the deal: $A,B$ reveal some of
 their cards in the process. 



\subsection{Distributed Computability}
\label{sec:relatedWork:distrComp}

In a distributed system a set of processes communicate with each other to solve problems. The
simplest kind is one where they start with input values, and decide on output values, once.
In a \emph{task} the domain is a set of input assignments to the processes,
 the range is a set of output assignments, and the task specification $\Delta$ is an input/output relation between them.
An input vector~$I$ specifies in its $i$-th entry the (private) input to the $i$-th process,
and an output vector $O\in\Delta(I)$ states that it is valid for each process $i$ to produce as output the $i$-th entry of~$O$, whenever the input vector is~$I$. 
In more detail, a task $\cT=(\cI,\cO,\Delta)$ is defined by an input complex $\cI$, an output complex $\cO$,
and a carrier map $\Delta$.
An important example of a task is \emph{consensus}, where each process is given an input from a set of possible input values, and the participating processes have to agree on one of their inputs.

Delporte et al.~\cite{DFRbits2020} observed that the least
amount of communication that $A$ and $B$ need to send to each other,
one has to consider a vertex coloring of the graphs $\cI_A$ and $\cI_B$.
As  explained here, when $\cI$ corresponds to a signature $(\Na,\Nb,\Nc)$,
the proper coloring needed is of $J^{\Nc+\Nr}(\Nn,\Na)$.
This is the minimum needed so that $A$ and $B$ can learn each other inputs,
otherwise there will be input vectors indistinguishable to them. Here we explore
the additional requirement that input vectors are  indistinguishable to $C$
after listening to the conversation.

Notions of \emph{indistinguishability}
are central in computer science, particularly in distributed computing.
Representing the indistinguishability structure appropriately, exposes
what can and cannot be done in a given situation~\cite{Indistinguishability}.

A \emph{distributed computing model} has to specify various details related to how the processes communicate
with each other and what type of failures may occur, e.g.~\cite{AW04,Lynch96}. 
It turns out that different models may have different power, i.e., solve different sets of tasks.
 
The theory of distributed computability has been well-developed since the early 1990's~\cite{HerlihyS99}, with origins even before~\cite{BiranMZ90,FischerLP85},
 and  overviewed in a book~\cite{HerlihyKR:2013}.
It was discovered that the reason for why a task may or may not be computable is of a topological nature.
 The input and output sets of vectors are best described as \emph{simplicial complexes},
  and a task can be specified by a relation $\Delta$ from the input complex $\mathcal I$ to the output complex $\mathcal O$.
 The main result is that a task is solvable in the layered message-passing model  if and only if 
 there is a certain subdivision of the input complex ${\mathcal I}$ and a certain simplicial map $\delta$ to the output complex ${\mathcal O}$, that respects the specification $\Delta$. 

Notice that the requirement that $A$ and $B$ decide on each others inputs is closely related to
the \emph{interactive consistency} problem (and other vector consensus variants e.g.~\cite{CNV06}),
 introduced early on~\cite{PSL80} in a system where processes may fail,
and has continued to be studied up to day due to its practical importance. 
Once a solution to interactive consistency is obtained,
a solution to consensus can be obtained, if each process decides e.g. on the majority
of the inputs it has received.

\subsection{Correlated input complex}
\label{sec:relatedWork:correlatedInputs}

In distributed computing a common situation is when the inputs are not correlated.
The input complex $\cI$
is called \emph{colorless:} any input may be assign to any process. Colorless tasks have both   
 input and  output complex colorless.
Correlated inputs make the task computability analysis much more complicated. Thus,  
the book~\cite{HerlihyKR:2013} treats first colorless tasks, and then presents more advanced topological
techniques to deal with the general setting.

In various situations related to renaming, the input complex consists of assigning
distinct input names to the processes, from some domain on input names.
This leads to a card game where each process gets a single card.
Correlated outputs have been considered for this input complex,
many encompassed by the Generalized symmetry breaking
family of tasks~\cite{CastanedaIRR16symBreak}.
In this paper the question is considered of which correlated inputs are sufficient
to solve other tasks, especially set agreement, in a wait-free read/write context.

The condition-based approach started in~\cite{MostefaouiRR03conditions}
 studies subcomplexes of tasks that have a colorless input complex
(consensus or set agreement),
that make an unsolvable task either solvable or more efficiently solvable.
Namely, how much correlation among inputs is required to solve a given task.
Relations with coding theory are investigated in~\cite{FriedmanMRR07}.

\section{Johnson graphs}
\label{sec:Johnson}

In a \emph{Johnson graph} $J(n,m)$  the vertices are  $m$-subsets of a $n$-set,
and two vertices $a,a'$ are adjacent when $a\cap a'=m-1$.
In Figure~\ref{fig-johnsonGs} some examples are depicted.
In other words, when the symmetric difference is $|a\triangle a'|=2$.
The special case of $J(7,3)$ corresponds to $\cG_B$ ($=\cG_B$) of the classic Russian cards
problem.
The  \emph{Johnson distance} $d$ of two $m$-sets is half the size of their symmetric difference.
Thus, the graph $J^d(n,m)$ describes the distance-$d$ relation, and $J^1(n,m)$
is denoted $J(n,m)$.
\begin{figure}[h]
\label{fig-johnsonGs}
\centering
\includegraphics[scale=0.3]{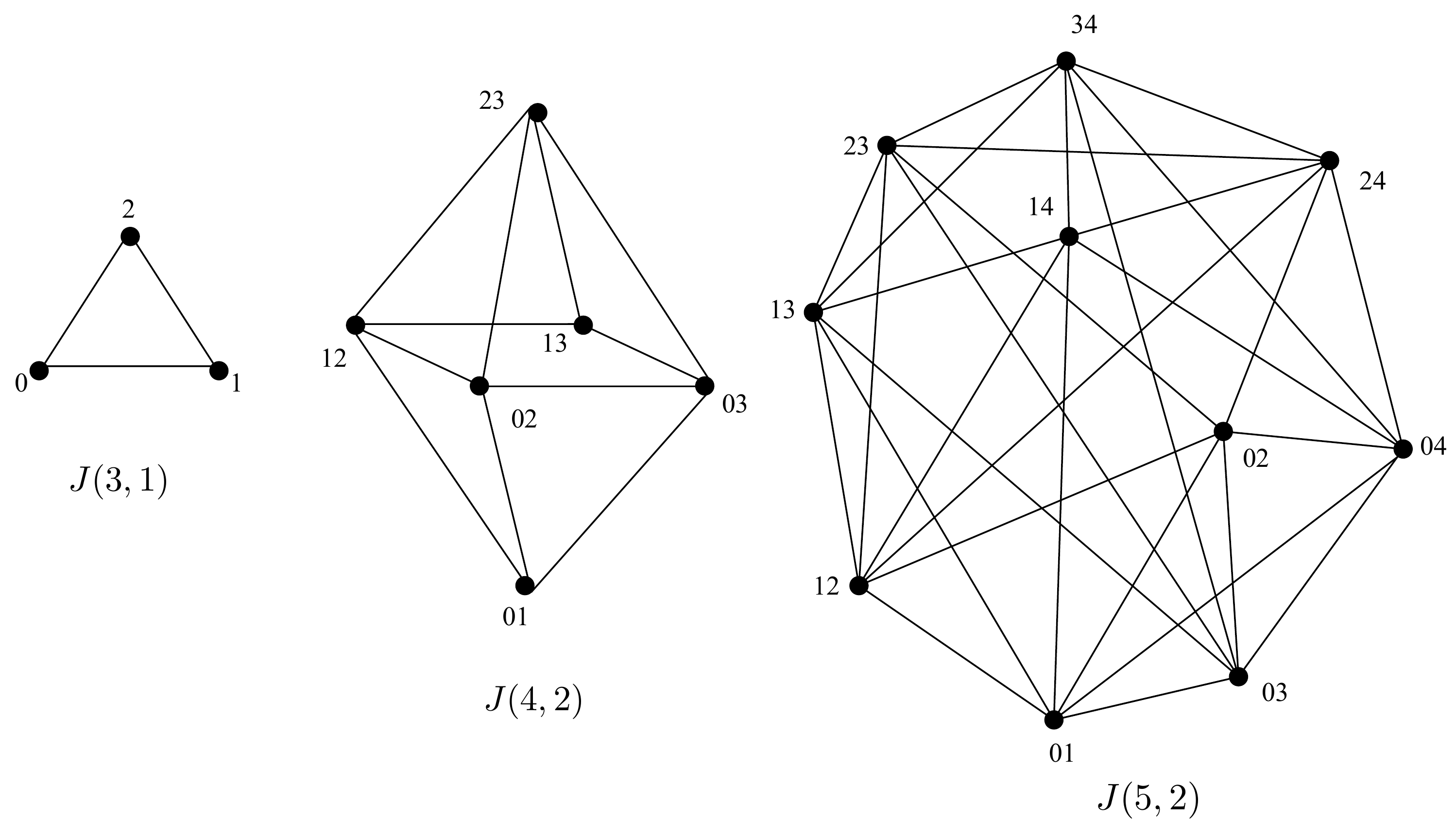}
\caption{Some Johnson graphs.}
\end{figure}

Notice that the \emph{Kneser graph,} $K(n,m)$ is the graph on the $m$-subsets of an $n$-set, adjacent when they are disjoint. And the \emph{generalized Johnson graph} $J(n,m,i)$ 
is the graph on the $m$-subsets of an $n$-set, adjacent whenever their intersection is $i$.
The graphs $J(n,m,m-1)$ are the Johnson graphs,  the graphs $J(n,m,0)$ are the Kneser graphs.
The Kneser graph $J(5,2,0)$ is the famous Petersen graph. All these are highly symmetric
graphs that recur throughout the book~\cite{GodsilRalgGraphTh}.
It is observed there that the following are isomorphic graphs $J(n,m,i)\cong  J(n,n-m,n-2m+i)$, $i\leq m\leq n$,
by the function that maps an $m$-set to its complement. 

Johnson graphs are related to coding theory, Quantum probability~\cite[Chapter 6: Johnson Graphs]{horaObata} 
and Steiner systems, and have been thoroughly studied.
Some of the more relevant   facts to our study are (we provide citations unless they are easy to prove):
\begin{enumerate}
\item 
\label{JG:isom}
The following are isomorphic graphs $J(n,m)\cong  J(n,n-m)$. Also, $J(n,1)\equiv J(n,n-1)\cong K_n$. 
$J(n,2)$ called a \emph{triangular graph},
which is the line graph of $K_n$.
\item 
\label{JG:distance}
Let $\delta(a,a')$ denote the distance between vertices $a,a'$ in $J(n,m)$. Then, $\delta(a,a')=k$
iff $a\cap a'=m-k$. 
Thus two $k$-subsets are adjacent in the Kneser graph $K(n, k)$ if and only if they are at maximum possible distance in $J(n, k)$. 
\item 
\label{JG:regularDiam}
$J(n,m)$ is distance-regular of diameter $\min\set{m, n – m}$. 
\item 
\label{JG:maxClicks}
The set of maximal cliques in $J(n,m)$ are of size $n-m+1$ and $m+1$ see~\cite{RamrasD2011}.
\item 
\label{JG:chrNumb}
The chromatic number of Johnson graphs have been well studied e.g.~\cite{EBitan96},
see Figure~\ref{fig-tableChrJohnsonG}.
But in general, determining the chromatic number of a Johnson graph is an open problem~\cite[Chapter 16]{godsil_meagher_2015}.
For the triangular graph, $\chi (J(n,2))=n$ for odd $n$, and $\chi (J(n,2))=n-1$ for even $n$.
It is known that $\chi (J(n,m))\leq n$. Often the chromatic number is a little bit smaller.
For $ n \equiv 1,3 \pmod{ 6}$, $n > 7$,
$\chi(J(n,3))= n-2$.
For the Russian cards case, notice that it is known that  $\chi(J(7,3)) =6$.
\item
\label{JG:genLB}
As far as we know the only general (for specific instances, there are others) lower bound 
on the chromatic number is $\chi(J(n,m)) \geq\max{ \set{n-m+1,m+1}}$,
implied by the maximal cliques in the Johnson graph.
\item 
\label{JG:transAuto}
The Johnson graph is vertex transitive and distance transitive. For $J(7,3)$, its automorphism group is $S_7.$
\item 
\label{JG:degConn}
$J(n,m)$ is regular of degree $m(n-m)$. Thus, it has vertex connectivity $m(n-m)$. See~\cite{DavenR99}.
\item
\label{JG:conj}
For vertex-transitive graphs with maximum degree  $\Delta\geq 13$
and clique number $\omega$, 	the Borodin-Kostochka conjecture,
$\chi\leq \max\set{\omega,\Delta-1 }$ was proved in~\cite{CranstonR15}.
Also, if $\omega<\Delta$ then $\chi\leq \Delta-1$.
\end{enumerate}

 \begin{figure}[h]
\centering
\includegraphics[scale=0.4]{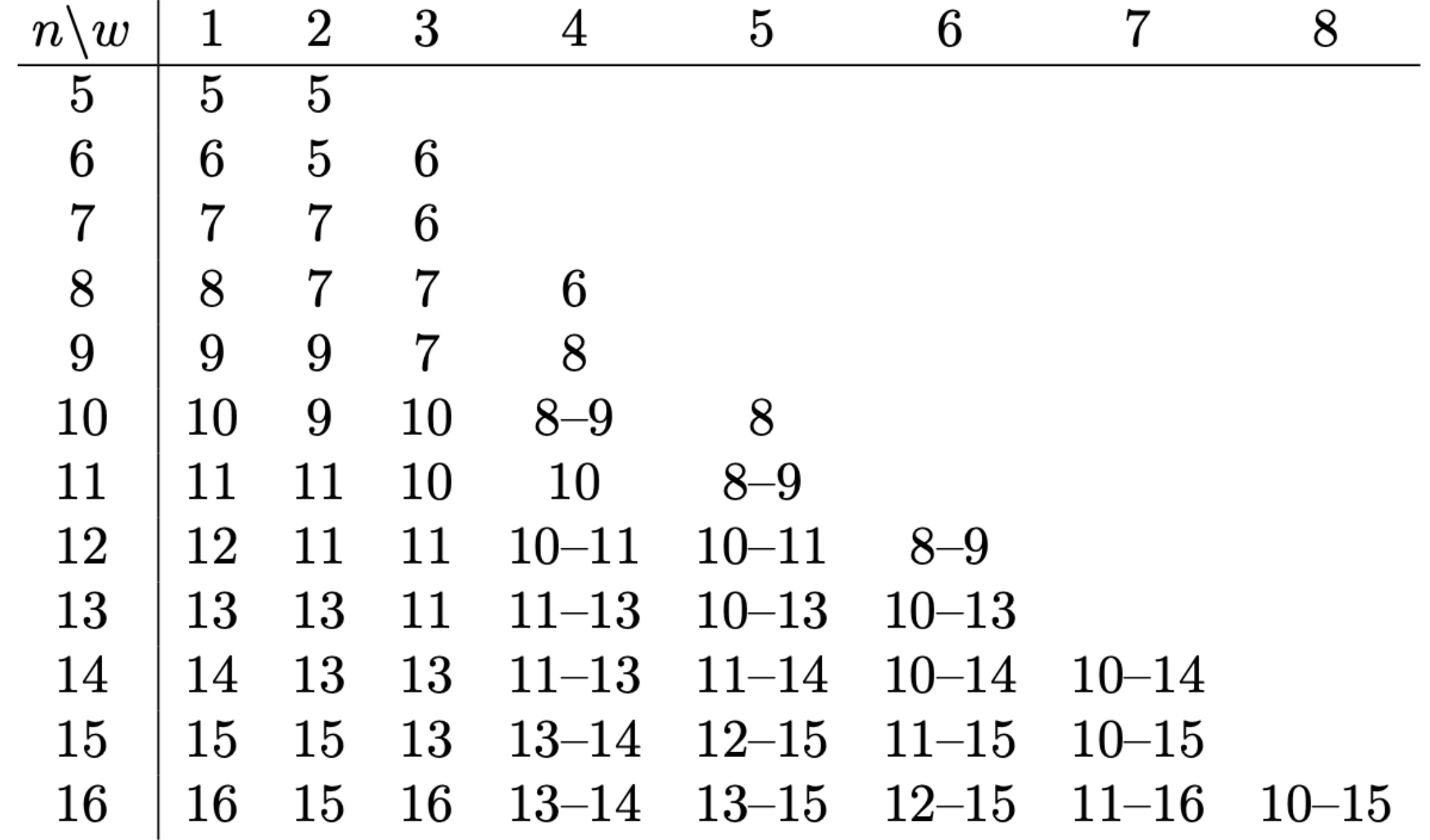} 
\caption{Bounds on the chromatic number of Johnson graphs~\cite[Table 4]{BE11}}. 
\label{fig-tableChrJohnsonG}
\end{figure}

\section{Impossibility of uniform solutions to the Russian cards problem}
\label{app:lowerBoundI}

Here we present additional details about six-message solutions to the Russian
cards problem, and the impossibility of Section~\ref{sec:rusCards:lower}, in which
Theorem~\ref{th:lowerB6} states that 
there is no uniform solution to the Russian cards problem with only six messages,
namely, where at most one color class is of size 7.

In the proof of Theorem~\ref{th:lowerB6}, it is shown that it it is not possible to design three color classes, where all
vertices of $G_0$ are of degree two. Figure~\ref{fig-6coloringRC-deg3ok} shows that
it is possible, using vertices of degree 3.

\begin{figure}[h]
\centering
\includegraphics[scale=0.45]{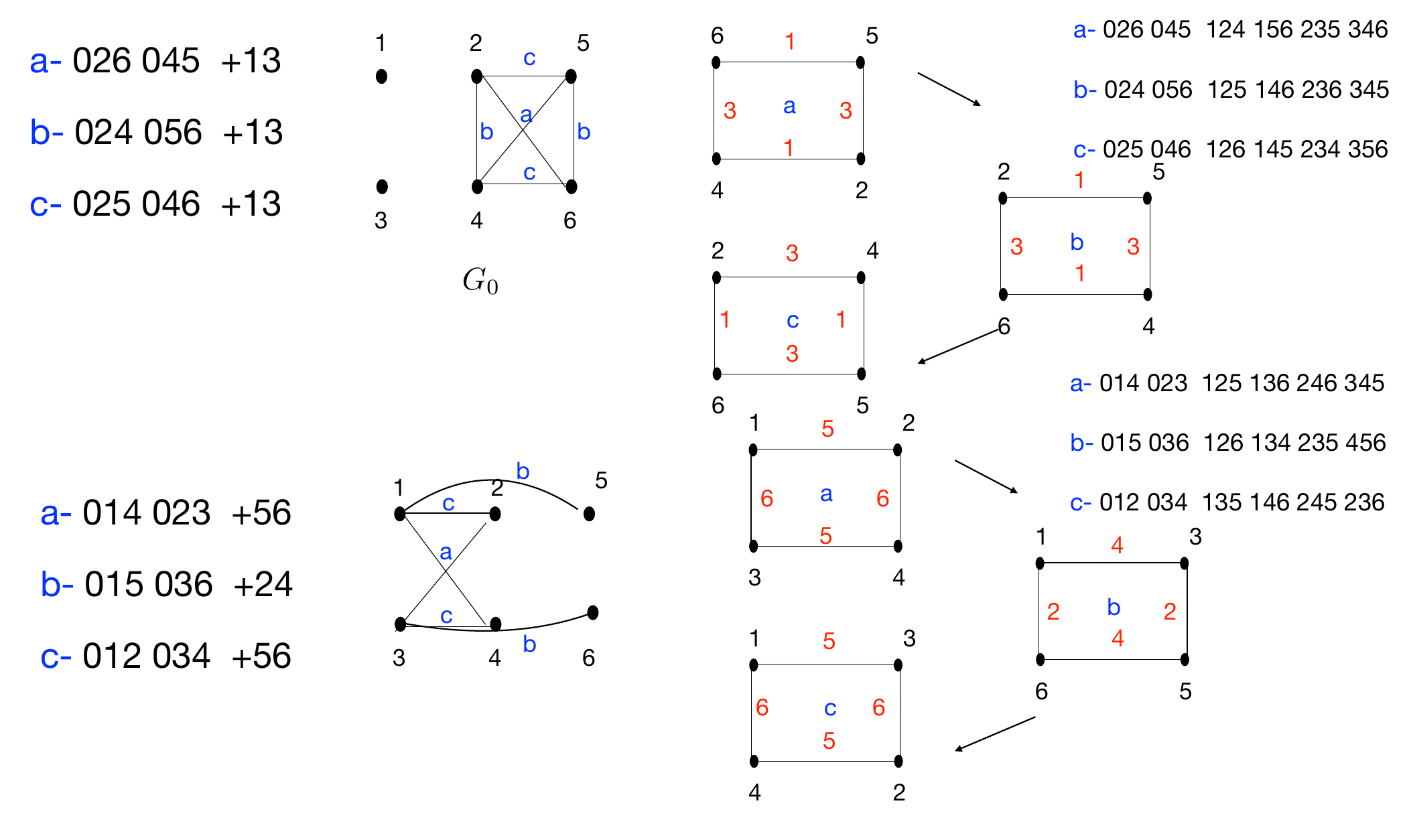}
\caption{It is possible to design three color classes with vertices of degree 3, here are two examples.}
\label{fig-6coloringRC-deg3ok}
\end{figure}

The full tree is  of configurations  $12,34$; $13,56$; $25,46$ is in Figure~\ref{fig-6coloringRC-tree}.

\begin{figure}[h]
\centering
\includegraphics[scale=0.45]{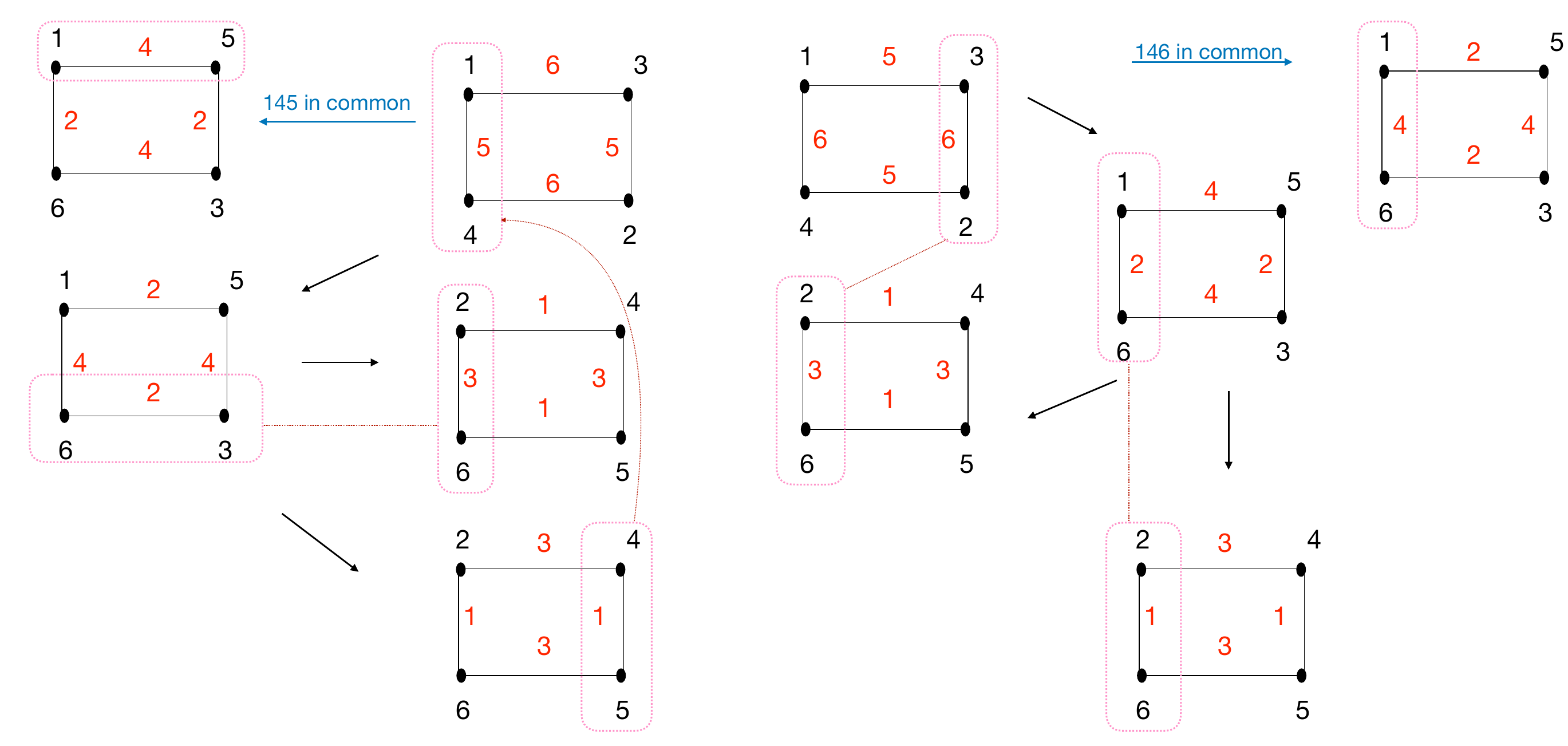}
\caption{The full tree for configuration $12,34$; $13,56$; $25,46$}
\label{fig-6coloringRC-tree}
\end{figure}

\section{Symmetric cases of the $\chi_{modn}$ protocol}
\label{app:symmModn}

We begin with the case of Remark~\ref{rm:prime} where $\Na,\Nn$ are relatively prime
(which includes~\cite[Corollary 9]{Cordon-FrancoDF12}).
In particular, this takes care of cases where $\Nn=2\Na+1$, such as the classic\footnote{
Interestingly, this is the case that had to be treated separately in~\cite{Cordon-FrancoDF12}.
The cases  $(4,3,1)$ and $(3,4,1)$ where checked using a Haskell script.
}
 $(3,3,1)$.

\begin{lemma}
\label{lm:safetyRussCprime}
The protocol $\chi_{modn}$ is safe when $\Nc+\Nr=1$, $\Na,\Nb\geq 3$, $\Nn\geq 7$
and $\Na,\Nn$ are relatively prime.
\end{lemma}

\begin{proof}
Assume that $\Na\leq \floor{\Nn/2}$, by the duality Theorem~\ref{th:dual}.

Consider an $M\in\mathbb{Z}_\Nn$ and $c\in\cards$. 
Let $y\in\bar{c}$.

Since   $\Na,\Nn$ are relatively prime (Remark~\ref{rm:prime}), 
there exists an $x$, such that $a=\set{x,x+1,\ldots,x+\Na-1}$ satisfies
 $\chi_{modn}(a_m)=m$.
Assume w.l.o.g. that $x=0$.

We consider several easy, similar cases,
where we  use Lemma~\ref{lem:basicTool} to obtain
$a'_1,a'_2\in\bar{c}$ such that $y\in a'_1\triangle a'_2$.

Case 1: assume that both $y$ and $c$ are in $a$.\\
Let $z_1=c$ and $z_2=y$, and apply the Lemma~\ref{lem:basicTool}  to obtain $a'_1\in\bar{c}$,
such that $y\not\in a'_1$.
Then, let $z_1=c$ and $z_2$ any card from $a$ different from $y$, and apply  Lemma~\ref{lem:basicTool}  to obtain $a'_2\in\bar{c}$,
such that $y\in a'_2$.

Case 2: assume that  $y\in a$ and $c\not\in a$.\\
In this case, we already have $a=a'_2\in\bar{c}$,
such that $y\in a'_2$. 
Thus, let $z_1=\Na-y-1$ and $z_2=y$. Notice that if $z_1\neq z_2$ then
 there exists an integer  $i$, $1\leq i \leq \floor{\ell_1/2}$  such that both 
$z_1+i\not\in a\cup c$ and
$z_2-i\not\in a\cup c$, and we can
apply  Lemma~\ref{lem:basicTool}  to obtain $a'_1\in\bar{c}$,
such that $y\not\in a'_1$.
Else, the conditions of the lemma hold for either   $z_1=y-1$ or $z_1=y+1$, with
 $z_2=y$, to obtain $a'_1\in\bar{c}$,
such that $y\not\in a'_1$.

Case 3: assume that  $y\not\in a$ and $c\not\in a$.\\
In this case, we already have $a=a'_1\in\bar{c}$,
such that $y\not\in a'_1$. 
Thus, let $z_1=0$ and $z_2=\Na-1$. If 
 there exists an integer  $i$, $1\leq i \leq \floor{\ell_1/2}$  such  
$z_1+i =y$ and
$z_2-i\neq c$,
 we can
apply  Lemma~\ref{lem:basicTool}  to obtain $a'_2\in\bar{c}$,
such that $y\in a'_2$.
Else, the conditions of the lemma hold for 
 $z_1=0$ and $z_2=\Na-2$,
to obtain $a'_2\in\bar{c}$,
such that $y\in a'_2$.

Case 4: assume that  $y\not\in a$ and $c\in a$ is similar.\\
\end{proof} 
 
 We now prove the symmetric case\footnote{
Interestingly, this is the case that had to be treated separately in~\cite{Cordon-FrancoDF12}.
The cases  $(4,3,1)$ and $(3,4,1)$ where checked using a Haskell script.
}
 where $\Na=\Nn/2$. 

 \begin{lemma}
 \label{lem:symmetricRCevenA}
If $\Na=\Nn/2-1$ with both $\Nn$ and $\Na$ even,
the protocol $\chi_{modn}$ is safe when $\Nc+\Nr=1$, $\Na \geq 3$, $\Nn\geq 7$.
 \end{lemma}
 \begin{proof}
 We have that $gcd(\Nn,\Na)=2$, 
 and hence for half of the values in $z_n$ there is an $m$ there are exactly two opposite sequences of $\Na$ cards, $a$ and $a'$
 with $\chi_{modn}(a)=\chi_{modn}(a')$, all even.
 For the other half, there is a sequence of $\Na-1$ consecutive, separated at the end by 1, 
 all odd.
 Thus, in either case, there are exactly two values that are not covered
 by these opposite $\Na$-sets.
And it is then easy to reach these two values.
 Figure~\ref{fig-safetyRussCardsSymEven} illustrates the two cases.
 \begin{figure}[h]
\centering
\includegraphics[scale=0.4]{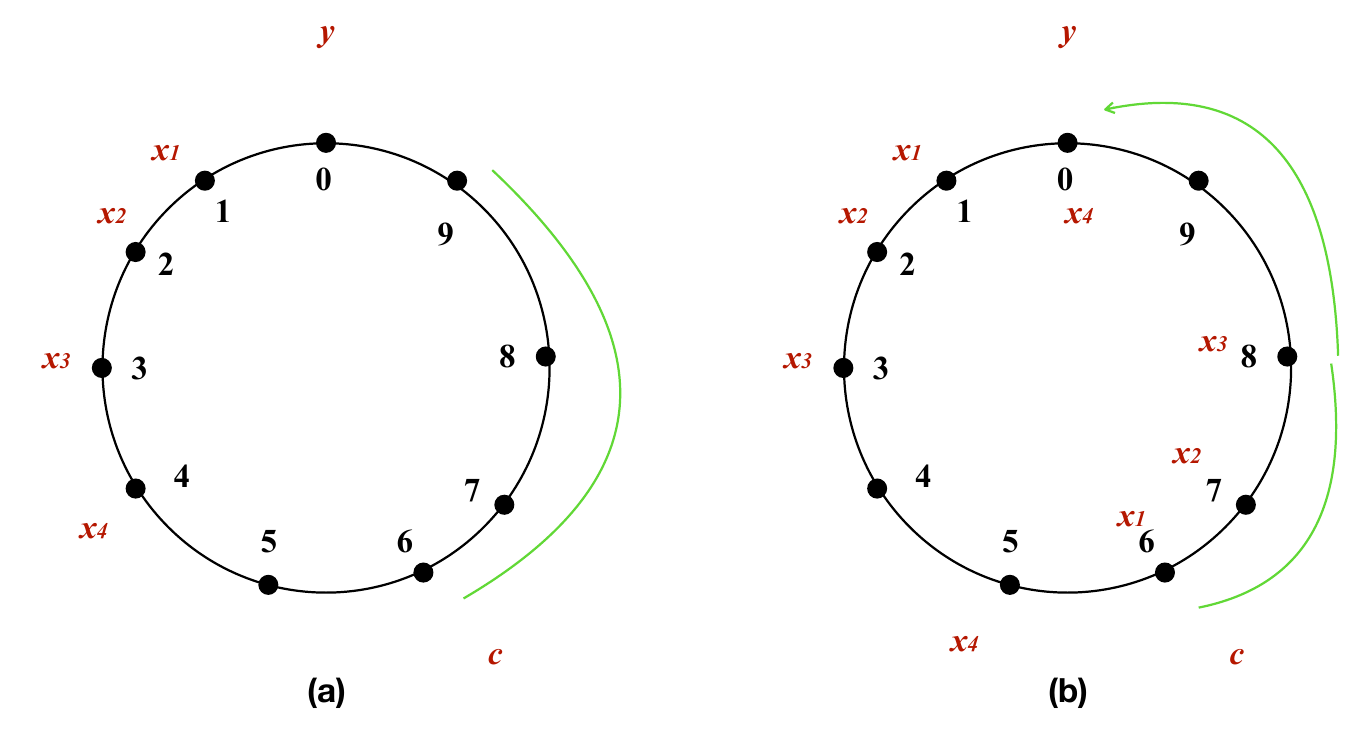} 
\caption{Symmetric case  $\Nn=10,\Na=4,\Nc=1,\Nr=0$}. 
\label{fig-safetyRussCardsSymEven}
\end{figure}
 \end{proof}

 \begin{lemma}
 \label{lem:symmetricRC}
If $\Nn=2\Na$, 
the protocol $\chi_{modn}$ is safe when $\Nc+\Nr=1$, $\Na \geq 3$, $\Nn\geq 7$.
 \end{lemma}

 \begin{proof}
The arguments are similar to the above, we present only a sketch.
We use Remark~\ref{rm:prime}, to choose
without loss of generality  $c=\Nn/2$.
Consider the two $\Na$-sets $a_1=\set{0,1,\ldots, c-1}$, and 
$a_2=\set{c+2,c+3,\ldots, 0,1}$.
Notice that $\chi_{modn}(a_1)=\chi_{modn}(a_2)$, because
$a_2=a_1+(\Nn/2+2)\Na$.
Thus, for each card $y\not\in\set{c,0,1}$,  $y\in a_1\triangle a_2$.
To complete the proof of this case, we use Lemma~\ref{lem:basicTool} as follows.
Consider $y=0$, and let $a'_1$ be
$$
a_1 \stackrel{-2, 0}{\longrightarrow} a_{1} \stackrel{c+1, c-1}{\longrightarrow} a'_1 .
$$
Thus, $y=0 \in a_1\triangle a'_1$, and $\chi_{modn}(a_1)=\chi_{modn}(a'_1)$.
Similarly, consider $y=1$,  and let $a'_2$ be
$$
a_2 \stackrel{2, 1}{\longrightarrow} a_{1} \stackrel{c+1, c+2}{\longrightarrow} a'_2 .
$$
Thus, $y=1 \in a_2\triangle a'_2$, and $\chi_{modn}(a_2)=\chi_{modn}(a'_2)$.
 \end{proof}

\end{document}